\documentclass[sigconf]{acmart}

\renewcommand\footnotetextcopyrightpermission[1]{} 

\usepackage{amsmath}
\usepackage{graphicx}
\usepackage{amsthm}
\usepackage{thmtools}
\usepackage{color}
\usepackage{enumitem}
\usepackage{bm}
\usepackage{multirow}
\usepackage[linesnumbered, lined, boxed]{algorithm2e}

\newcommand{\nats}{\mathbb{N}}
\newcommand{\nnints}{\mathbb{Z}_{\ge0}}
\newcommand{\reals}{\mathbb{R}}
\newcommand{\nnreals}{\mathbb{R}_{\ge0}}

\renewcommand{\epsilon}{\varepsilon}

\newcommand{\calG}{\mathcal{G}}

\newcommand{\calM}{\mathcal{M}}
\newcommand{\calR}{\mathcal{R}}

\newcommand{\calX}{\mathcal{X}}
\newcommand{\calY}{\mathcal{Y}}

\newcommand{\MSE}{\operatorname{\textsf{MSE}}}

\newcommand{\E}{\operatorname{\mathbb{E}}}
\newcommand{\V}{\operatorname{\mathbb{V}}}
\newcommand{\cov}{\operatorname{\text{Cov}}}
\newcommand{\bmA}{\mathbf{A}}

\newcommand{\bma}{\mathbf{a}}

\newcommand{\hf}{\hat{f}}

\newcommand{\td}{\tilde{d}}

\def\AlgWS{\textsf{WS}}
\def\AlgWSLE{\textsf{WSLE}}
\def\AlgWSTri{\textsf{WShuffle}$_{\triangle}$}
\def\AlgWSTriVR{\textsf{WShuffle}$_{\triangle}^*$}

\def\AlgWLTri{\textsf{WLocal}$_{\triangle}$}
\def\AlgARRTri{\textsf{ARR}$_{\triangle}$}
\def\AlgRRTri{\textsf{RR}$_{\triangle}$}
\def\AlgTwoRL{\textsf{2R-Large}$_{\triangle}$}
\def\AlgTwoRS{\textsf{2R-Small}$_{\triangle}$}

\def\AlgWSCyc{\textsf{WShuffle}$_{\square}$}
\def\AlgWLCyc{\textsf{WLocal}$_{\square}$}

\newcommand{\Gplus}{\textsf{Gplus}}
\newcommand{\IMDB}{\textsf{IMDB}}

\newcommand{\Lap}{\textrm{Lap}}

\SetKwInput{KwData}{Input}
\SetKwInput{KwResult}{Output}

\newif\ifconferenceon\conferenceonfalse
\ifconferenceon
\newcommand{\conference}[1]{#1}
\newcommand{\arxiv}[1]{}
\else
\newcommand{\conference}[1]{}
\newcommand{\arxiv}[1]{#1}
\usepackage{balance}
\fi

\AtBeginDocument{%
  }

\begin{document}

\title{Differentially Private Triangle and 4-Cycle Counting in the Shuffle Model}

\author{Jacob Imola}
\authornote{The first and second authors made equal contribution and are listed alphabetically.}
\email{jimola@eng.ucsd.edu}
\affiliation{%
  \institution{UC San Diego}
  \country{USA}
}

\author{Takao Murakami}
\authornotemark[1]
\email{takao-murakami@aist.go.jp}
\affiliation{%
  \institution{AIST}
  \country{Japan}
}

\author{Kamalika Chaudhuri}
\email{kamalika@cs.ucsd.edu}
\affiliation{%
  \institution{UC San Diego}
  \country{USA}
}

\begin{abstract}
Subgraph counting is fundamental for analyzing connection patterns or clustering tendencies in 
graph data. 
Recent studies have applied LDP (Local Differential Privacy) to subgraph counting to protect user privacy 
even against a data collector 
in 
social networks. 
However, existing local algorithms suffer from extremely large estimation errors or assume multi-round interaction between users and the data collector, which requires a lot of user effort and synchronization. 

In this paper, we focus on a one-round of interaction and propose accurate subgraph counting algorithms by introducing a recently studied shuffle model. 
We first propose a basic technique called \textit{wedge shuffling} to send wedge information, the main component of several subgraphs, with small noise. 
Then we apply our wedge shuffling to counting triangles and 4-cycles -- basic subgraphs for analyzing clustering tendencies -- with several additional techniques. 
We also show upper bounds on the estimation error for each algorithm. 
We show through comprehensive experiments that our one-round shuffle algorithms significantly outperform the one-round local algorithms in terms of accuracy and achieve small estimation errors with a reasonable privacy budget, e.g., smaller than 1 in edge DP. 
\end{abstract}

\maketitle
\pagestyle{plain}

\section{Introduction}
\label{sec:intro}
Graph statistics is useful for finding meaningful connection patterns in network data, and 
subgraph counting is known as a fundamental task in graph analysis. 
For example, a triangle is a cycle of size three, 
and a $k$-star consists of a central node connected to $k$ other nodes.  
These subgraphs 
can be used to calculate 
a 
clustering coefficient ($=\frac{3 \times \text{\#triangles}}{\text{\#2-stars}}$). 
In a social graph, the clustering coefficient measures 
the tendency of nodes (users) to form a cluster with each other. 
It also represents the average probability that a friend's friend is also a friend \cite{Newman_PRL09}. 
Therefore, the clustering coefficient is useful for analyzing the effectiveness of friend suggestions. 
Another example of the subgraph is a 4-cycle, 
a cycle of size four. 
The 4-cycle count is useful for measuring the clustering ability in bipartite graphs (e.g., 
online dating networks, mentor-student networks \cite{Kutty_WWW14}) 
where a triangle never appears \cite{Lind_PRE05,Robins_CMOT04,Sanei-Mehri_CIKM19}. 
Figure~\ref{fig:subgraphs} shows examples of triangles, 2-stars, and 4-cycles. 
Although these subgraphs are important for analyzing the connection patterns or clustering tendencies, their exact numbers can leak sensitive edges (friendships) \cite{Imola_USENIX21}. 

\begin{figure}[t]
  \centering
  \includegraphics[width=0.65\linewidth]{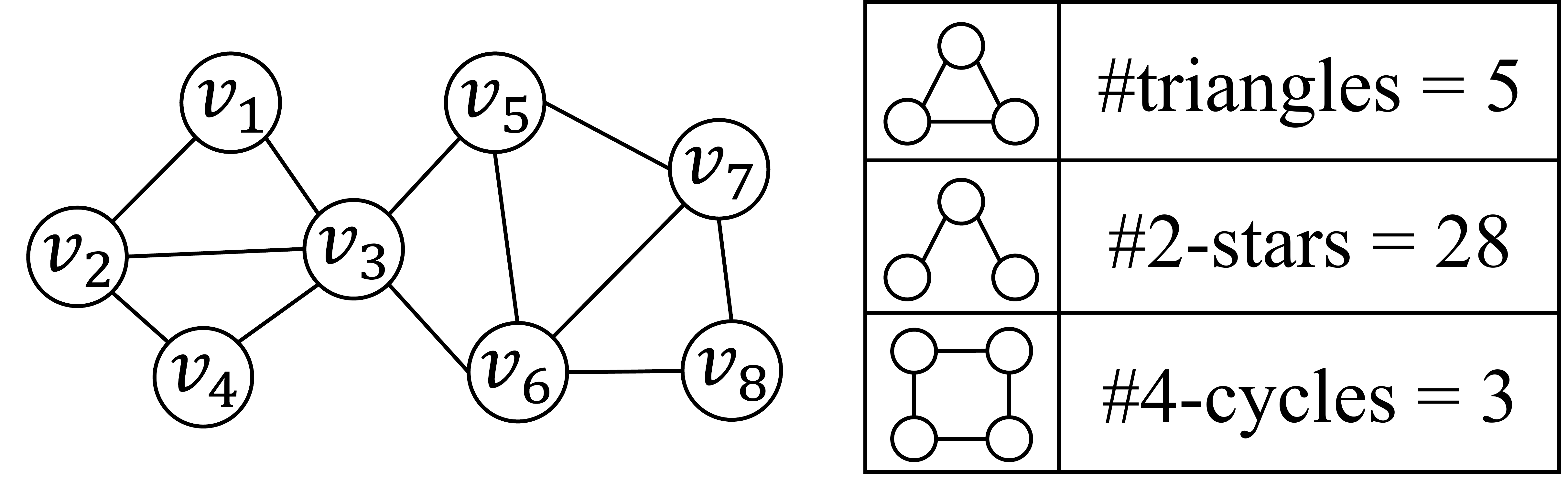}
  \vspace{-2mm}
  \caption{Examples of subgraph counts.
  }
  \label{fig:subgraphs}
\end{figure}

DP (Differential Privacy) \cite{Dwork_ICALP06,DP} -- the gold standard of privacy notions -- has been widely used to strongly protect edges in graph data \cite{Day_SIGMOD16,Ding_TKDE21,Hay_ICDM09,Imola_USENIX21,Imola_USENIX22,Karwa_PVLDB11,Kasiviswanathan_TCC13,Qin_CCS17,Sun_CCS19,Ye_ICDE20,Ye_TKDE21}. 
In particular, recent studies 
\cite{Imola_USENIX21,Imola_USENIX22,Qin_CCS17,Ye_ICDE20,Ye_TKDE21} 
have applied LDP (Local DP) \cite{Kasiviswanathan_FOCS08} to 
graph data. 
In the graph LDP model, each user obfuscates her neighbor list (friends list) by herself and sends the obfuscated neighbor list to a data collector. 
Then, the data collector estimates graph statistics, such as subgraph counts. 
Compared to central DP where a central server has personal data of all users (i.e., the entire graph), LDP does not have a risk that all personal data are leaked from the server by 
cyberattacks 
\cite{Henriquez_breach2021} or insider attacks \cite{Kohen_insider_threats}. 
Moreover, LDP can be applied to decentralized social networks \cite{Paul_CN14,Salve_CSR18} (e.g., diaspora* \cite{Diaspora}, Mastodon \cite{Mastodon}) where no server can access the entire graph; e.g., 
the entire graph is distributed across many servers, or no server has any original edges. 
It is reported in \cite{Imola_USENIX21} that $k$-star counts can be accurately estimated in this model. 

However, it is much more challenging to accurately count more complicated subgraphs such as triangles and 4-cycles under LDP. 
The root cause of this is its local property -- a user cannot see edges between others. 
For example, user $v_1$ cannot count triangles or 4-cycles including $v_1$, as she cannot see edges between others, e.g., $(v_2,v_3)$, $(v_2,v_4)$, and $(v_3,v_4)$. 
Therefore, the existing algorithms \cite{Imola_USENIX21,Imola_USENIX22,Ye_ICDE20,Ye_TKDE21} 
obfuscate 
each bit of the neighbor list 
rather than the subgraph count by 
the RR (Randomized Response) \cite{Warner_JASA65}, which randomly flips 0/1. 
As a result, their algorithms suffer from extremely large estimation errors because it makes all edges noisy. 
Some studies \cite{Imola_USENIX21,Imola_USENIX22} 
significantly improve the accuracy by introducing 
an additional round of interaction between users and the data collector. 
However, multi-round interaction may be impractical in many applications, as it requires a lot of user effort and synchronization; 
in  \cite{Imola_USENIX21,Imola_USENIX22}, every user must respond twice, and the data collector must wait for responses from all users 
in each round. 

In this work, we focus on a \textbf{one-round} of interaction between users and the data collector and propose accurate subgraph counting algorithms by introducing a recently studied privacy model: the \textit{shuffle model} \cite{Erlingsson_SODA19,Feldman_FOCS21}. 
In the shuffle model, each user sends her (encrypted) obfuscated data to an intermediate server called the shuffler. 
Then, the shuffler randomly shuffles the obfuscated data of all users and sends the shuffled data to the data collector (who decrypts them). 
The shuffling amplifies DP guarantees of the obfuscated data under the assumption that the shuffler and the data collector do not collude with each other. 
Specifically, it is known that DP strongly protects user privacy when a parameter (a.k.a. privacy budget) $\epsilon$ is small, e.g., $\epsilon \leq 1$ \cite{DP_Li}. 
The shuffling significantly reduces $\epsilon$ and therefore significantly improves utility at the same value of $\epsilon$. 
To date, the shuffle model has been successfully applied to tabular data \cite{Meehan_ICLR22,Wang_PVLDB20} and 
gradients \cite{Girgis_AISTATS21,Liu_AAAI21} in federated learning. 
We 
apply the shuffle model to graph data to accurately count subgraphs within one round. 

The main challenge in subgraph counting in the shuffle model is that each user's neighbor list is \textit{high-dimensional data}, i.e., $n$-dim binary string where $n$ is the number of users. 
Consequently, applying the RR to each bit of the neighbor list, as in the existing work \cite{Imola_USENIX21,Imola_USENIX22,Ye_ICDE20,Ye_TKDE21}, results in an extremely large privacy budget $\epsilon$ even after applying the shuffling (see Section~\ref{sub:technical} for more details). 

We address this issue by introducing a new, basic technique called \textit{wedge shuffling}. 
In graphs, a wedge between $v_i$ and $v_j$ is defined by a 2-hop path with endpoints $v_i$ and $v_j$. 
For example, in Figure~\ref{fig:subgraphs}, 
there are two wedges between $v_2$ and $v_3$: $v_2$-$v_1$-$v_3$ and $v_2$-$v_4$-$v_3$. 
In other words, users $v_1$ and $v_4$ have a wedge between $v_2$ and $v_3$, whereas $v_5, \ldots, v_8$ do not. 
Each user obfuscates such wedge information by the RR, and the shuffler randomly shuffles them. 
Because the wedge information (i.e., whether there is a wedge between a specific user-pair) 
is \textit{one-dimensional binary data}, it can be sent with small noise and small $\epsilon$. 
In addition, the wedge is the main component of several subgraphs, such as triangles, 4-cycles, and 3-hop paths \cite{Sun_CCS19}. 
Since the wedge has little noise, we can accurately count these subgraphs based on wedge shuffling. 

We apply wedge shuffling to triangle and 4-cycle counting tasks with several additional techniques. 
For triangles, we first 
propose an algorithm that counts triangles involving the user-pair at the endpoints of the wedges by locally sending an edge between the user-pair to the data collector. 
Then we propose an algorithm to count triangles in the entire graph by sampling disjoint user-pairs, which share no common users (i.e., 
no user falls in two pairs). 
We also propose a technique to reduce the variance of the estimate by ignoring sparse user-pairs, where either of the two users has a very small degree. 
For 4-cycles, we propose an algorithm to calculate an unbiased estimate of the 4-cycle count from that of the wedge count via bias correction. 

We provide upper bounds on the estimation error for our triangle and 4-cycles counting algorithms. 
Through comprehensive evaluation, we show that 
our algorithms 
accurately estimate these subgraph counts within one round under the shuffle model.

\smallskip
\noindent{\textbf{Our Contributions.}}~~Our contributions are as follows: 
\begin{itemize}
    \item We propose a wedge shuffle technique to enable privacy amplification 
    of graph data. 
    To our knowledge, we are the first to shuffle graph data (see Section~\ref{sec:related} for more details). 
    \item We propose one-round triangle and 4-cycle counting algorithms 
    based on our wedge shuffle technique. 
    For triangles, we propose three additional techniques: sending local edges, sampling disjoint user-pairs, and variance reduction by ignoring sparse user-pairs. 
    For 4-cycles, we propose a bias correction technique. 
    We show upper bounds on the estimation error for each algorithm. 
    \item We evaluate our algorithms using two real graph datasets. 
    Our experimental results show that our one-round shuffle algorithms 
    significantly outperform one-round local algorithms in terms of accuracy 
    and achieve a small estimation error (relative error $\ll 1$) with a reasonable privacy budget, e.g., smaller than $1$ in edge DP. 
\end{itemize}
In \conference{the full version of this paper \cite{Imola_CCSFull22}}\arxiv{Appendix~\ref{sec:cluster}}, we show that our triangle algorithm is also useful for accurately estimating the clustering coefficient within one round. 
We can use our algorithms to analyze the clustering tendency or the effectiveness of friend suggestions in decentralized social networks by introducing a shuffler. 
We implemented our algorithms in C/C++. 
Our code is available on GitHub \cite{SubgraphShuffle}. 
The proofs of all statements in the main body are given in \conference{\cite{Imola_CCSFull22}}\arxiv{Appendices~\ref{sec:triangle_proof} and \ref{sec:4cycle_proof}}. 

\section{Related Work}
\label{sec:related}
\noindent{\textbf{Non-private Subgraph Counting.}}~~Subgraph counting has been extensively studied in a non-private setting (see \cite{Ribeiro_CS21} for a recent survey). 
Examples of subgraphs include triangles \cite{Bera_PODS20,Eden_FOCS15,Kolountzakis_IM12,Wu_TKDE16},  4-cycles \cite{Bera_STACS17,Kallaugher_PODS19,Manjunath_ESA11,McGregor_PODS20}, $k$-stars \cite{Aliakbarpour_Alg18,Gonen_DM11}, and 
$k$-hop paths \cite{Bjorklund_ICALP19,Kartun-Giles_TWC18}. 

Here, the main challenge is to reduce the computational time of counting these subgraphs in large-scale graph data. 
One of the simplest approaches is edge sampling \cite{Bera_PODS20,Eden_FOCS15,Wu_TKDE16}, which randomly 
samples edges in a graph. 
Edge sampling outperforms other sampling methods 
(e.g., node sampling, triangle sampling) \cite{Wu_TKDE16} and is also adopted in \cite{Imola_USENIX22} 
for private triangle counting. 

Although our triangle algorithm also samples user-pairs, ours is different from edge sampling in two ways. 
First, our algorithm does not sample an edge but samples a pair of users who may or may not be a friend. 
Second, our algorithm samples user-pairs that share no common users to avoid the increase of the privacy budget $\epsilon$ as well as to reduce the time complexity 
(see Section~\ref{sec:triangle} for details). 

\smallskip
\noindent{\textbf{Private Subgraph Counting.}}~~Differentially private subgraph counting has been widely studied, and the previous work assumes either the central \cite{Ding_TKDE21,Karwa_PVLDB11,Kasiviswanathan_TCC13} or local \cite{Imola_USENIX21,Imola_USENIX22,Sun_CCS19,Ye_ICDE20,Ye_TKDE21} models. 
The central model assumes a centralized social network and has a data breach issue, as explained in Section~\ref{sec:intro}. 

Subgraph counting in the local model has recently attracted attention. 
Sun \textit{et al.} \cite{Sun_CCS19} propose subgraph counting algorithms 
assuming 
that each user knows all friends' friends. 
However, this assumption does not hold in many social networks; e.g., Facebook users can change their settings so that anyone cannot see their friend lists. 
Therefore, we make a minimal assumption -- each user knows only her friends.

In this setting, recent studies propose triangle \cite{Imola_USENIX21,Imola_USENIX22,Ye_ICDE20,Ye_TKDE21} and $k$-star \cite{Imola_USENIX21} counting algorithms. 
For $k$-stars, Imola \textit{et al.} \cite{Imola_USENIX21} propose a one-round algorithm that is order optimal and show that it provides a very small estimation error. 
For triangles, they
propose a one-round algorithm that applies the RR to 
each bit of the neighbor list 
and then calculates an unbiased estimate of triangles from the noisy graph. 
We call this algorithm \AlgRRTri{}. 
Imola \textit{et al.} \cite{Imola_USENIX22} show that \AlgRRTri{} provides a much smaller estimation error than the one-round triangle algorithms in \cite{Ye_ICDE20,Ye_TKDE21}. 
In \cite{Imola_USENIX22}, they also reduce the time complexity of \AlgRRTri{} 
by using the ARR (Asymmetric RR), which samples each 1 (edge) after applying the RR. 
We call this algorithm \AlgARRTri{}. 
In this paper, we use \AlgRRTri{} and \AlgARRTri{} as baselines in triangle counting. 
For 4-cycles, there is no existing algorithm under LDP, to our knowledge. 
Thus, we compare our shuffle algorithm with its local version, which does not shuffle the obfuscated data. 

For triangles, Imola \textit{et al.} also propose a two-round local algorithm in \cite{Imola_USENIX21} and significantly reduce its download cost in \cite{Imola_USENIX22}. 
Although we focus on one-round algorithms, 
we show in 
\conference{the full version \cite{Imola_CCSFull22}}\arxiv{Appendix~\ref{sec:two-round}} 
that our one-round algorithm is comparable to the two-round algorithm in \cite{Imola_USENIX22}, which requires a lot of user effort and synchronization, in terms of accuracy.

\smallskip
\noindent{\textbf{Shuffle Model.}}~~The privacy amplification by shuffling has been recently studied in \cite{Balle_CRYPTO19,Cheu_EUROCRYPT19,Erlingsson_SODA19,Feldman_FOCS21}. 
Among them, the privacy amplification bound by Feldman \textit{et al.} \cite{Feldman_FOCS21} is the state-of-the-art -- it provides a smaller $\epsilon$ than other bounds, such as \cite{Balle_CRYPTO19,Cheu_EUROCRYPT19,Erlingsson_SODA19}. 
Girgis \textit{et al.} \cite{Girgis_CCS21} 
consider multiple interactions between users and the data collector and 
show a better bound than the bound in \cite{Feldman_FOCS21} when used with composition. 
However, the bound in \cite{Feldman_FOCS21} outperforms the bound in \cite{Girgis_CCS21} when used without composition. 
Because our work focuses on a single interaction and does not use the composition, 
we use the bound in \cite{Feldman_FOCS21}. 

The shuffle model has been applied to tabular data \cite{Meehan_ICLR22,Wang_PVLDB20} and gradients \cite{Girgis_AISTATS21,Liu_AAAI21} in federated learning. 
Meehan \textit{et al.} \cite{Meehan_ICLR22} construct a graph from public auxiliary information and determine a permutation of obfuscated data using the graph to reduce re-identification risks. 
Liew \textit{et al.} \cite{Liew_SIGMOD22} propose network shuffling, which shuffles obfuscated data via random walks on a graph. 
Note that both \cite{Meehan_ICLR22} and \cite{Liew_SIGMOD22} use graph data to shuffle another type of data. 
To our knowledge, our work is the first to shuffle graph data itself. 

\section{Preliminaries}
\label{sec:preliminaries}
In this section, we describe some preliminaries for our work. 
Section~\ref{sub:notations} defines the basic notation used in this paper. 
Sections~\ref{sub:privacy} and \ref{sub:shuffle} introduce DP on graphs and the shuffle model, respectively. 
Section~\ref{sub:utility} explains utility metrics. 

\subsection{Notation}
\label{sub:notations}
Let $\reals$, $\nnreals$, $\nats$, and $\nnints$ be the sets of real numbers, non-negative real numbers, natural numbers, and non-negative integers, respectively. 
For $a\in\nats$, let $[a]$ be the set of natural numbers that do not exceed $a$, i.e., $[a] = \{1, 2, \ldots, a\}$. 

We consider an undirected social graph $G=(V,E)$, where $V$ represents a set of nodes (users) and $E \subseteq V \times V$ represents a set of edges (friendships). 
Let $n\in\nats$ be the number of nodes in $V$, and $v_i \in V$ be the $i$-th node, i.e., $V=\{v_1,\ldots,v_n\}$. 
Let $I_{-(i,j)}$ be the set of indices of users other than $v_i$ and $v_j$, i.e., $I_{-(i,j)} = [n]\setminus\{i,j\}$. 
Let 
$d_i \in \nnints$ be a degree of $v_i$, 
$d_{avg} \in \nnreals$ be the average degree of $G$, and $d_{max} \in \nats$ be the maximum degree of $G$. 
In most real graphs, $d_{avg} \ll d_{max} \ll n$ holds. 
We denote a set of graphs with $n$ nodes by $\calG$. 
Let $f^\triangle: \calG \rightarrow \nnints$ and $f^\square: \calG \rightarrow \nnints$ be triangle and 4-cycle count functions, respectively. 
The triangle count function takes $G \in \calG$ as input and outputs the number $f^\triangle(G)$ of triangles in $G$, 
whereas the 4-cycle count function takes $G$ as input and outputs the number $f^\square(G)$ of 4-cycles. 

Let $\bmA=(a_{i,j}) \in \{0,1\}^{n \times n}$ be 
an adjacency matrix corresponding to $G$. 
If $(v_i,v_j) \in E$, then $a_{i,j} = 1$; otherwise, $a_{i,j} = 0$. 
We call $a_{i,j}$ an \textit{edge indicator}. 
Let $\bma_i \in \{0,1\}^n$ be a neighbor list of user $v_i$, i.e., the $i$-th row of $\bmA$. 
\arxiv{Table~\ref{tab:notations} shows the basic notation in this paper.} 

\arxiv{\begin{table}[t]
\caption{Basic notation in this paper.}
\vspace{-4mm}
\centering
\hbox to\hsize{\hfil
\begin{tabular}{l|l}
\hline
Symbol		&	Description\\
\hline
$G=(V,E)$   &	    Undirected social graph.\\
$n$         &	    Number of nodes (users).\\
$v_i$       &       $i$-th user in $V$, i.e., $V=\{v_1,\ldots,v_n\}$.\\
$I_{-(i,j)}$    &   $=[n]\setminus\{i,j\}$.\\
$d_i$   &       Degree of $v_i$.\\
$d_{avg}$   &       Average degree in $G$.\\
$d_{max}$   &       Maximum degree in $G$.\\
$\calG$     &       Set of possible graphs with $n$ nodes.\\
$f^\triangle(G)$   &  Triangle count in graph $G$.\\
$f^\square(G)$   &  4-cycle count in graph $G$.\\
$\bmA=(a_{i,j})$	    &		Adjacency matrix.\\
$\bma_i$	&		Neighbor list of $v_i$, i.e., the $i$-th row of $\bmA$.\\
\hline
\end{tabular}
\hfil}
\label{tab:notations}
\end{table}}

\subsection{Differential Privacy}
\label{sub:privacy}
\noindent{\textbf{DP and LDP.}}~~We use differential privacy, and more specifically $(\epsilon,\delta)$-DP \cite{DP}, as a privacy metric: 

\begin{definition} [$(\epsilon,\delta)$-DP \cite{DP}] \label{def:DP} 
Let $n \in \nats$ be the number of users. 
Let $\epsilon \in \nnreals$ and $\delta \in [0,1]$. 
Let $\calX$ be the set of input data for each user. 
A randomized algorithm $\calM$ with domain $\calX^n$ 
provides \emph{$(\epsilon,\delta)$-DP} if for any neighboring databases $D,D' \in \calX^n$ that differ in a single user's data and any 
$S \subseteq \mathrm{Range}(\calM)$, 
\begin{align*}
\Pr[\calM(D) \in S] \leq e^\epsilon \Pr[\calM(D') \in S] + \delta.
\end{align*}
\end{definition}
$(\epsilon,\delta)$-DP guarantees that two neighboring datasets $D$ and $D'$ are almost equally likely when $\epsilon$ and $\delta$ are close to $0$. 
The parameter $\epsilon$ is called the privacy budget. 
It is well known that $\epsilon \leq 1$ is acceptable and $\epsilon \geq 5$ is unsuitable in many practical scenarios \cite{DP_Li}. 
In addition, the parameter $\delta$ needs to be much smaller than $\frac{1}{n}$ \cite{Barber_arXiv14,DP}. 

LDP \cite{Kasiviswanathan_FOCS08} is a special case of DP where $n=1$. 
In this case, a randomized algorithm is called a \textit{local randomizer}. 
We denote the local randomizer by $\calR$ to distinguish it from the randomized algorithm $\calM$ in the central model. 
Formally, LDP is defined as follows: 
\begin{definition} [$\epsilon$-LDP \cite{Kasiviswanathan_FOCS08}] \label{def:LDP} 
Let $\epsilon \in \nnreals$. 
Let $\calX$ be the set of input data for each user. 
A local randomizer $\calR$ with domain $\calX$ 
provides \emph{$\epsilon$-LDP} if for any $x,x' \in \calX$ and any $S \subseteq \mathrm{Range}(\calR)$, 
\begin{align}
\Pr[\calR(x) \in S] \leq e^\epsilon \Pr[\calR(x') \in S].
\label{eq:LDP}
\end{align}
\end{definition}

\smallskip
\noindent{\textbf{Randomized Response.}}~~We use Warner's RR (Randomized Response) \cite{Warner_JASA65} to provide 
LDP. 
Given $\epsilon \in \nnreals$, Warner's RR $\calR_{\epsilon}^W: \{0,1\} \rightarrow \{0,1\}$ maps $x \in \{0,1\}$ to $y \in \{0,1\}$ with the probability: 
\begin{align*}
    \Pr[\calR_{\epsilon}^W(x) = y] = 
    \begin{cases}
    \frac{e^\epsilon}{e^\epsilon + 1}   &   \text{(if $x=y$)} \\
    \frac{1}{e^\epsilon + 1}   &   \text{(otherwise)}.
    \end{cases}
\end{align*}
$\calR_{\epsilon}^W$ 
provides $\epsilon$-LDP in Definition~\ref{def:LDP}, where $\calX = \{0,1\}$. 
We refer to Warner's RR $\calR_{\epsilon}^W$ with parameter $\epsilon$ as \textit{$\epsilon$-RR}.

\smallskip
\noindent{\textbf{DP on Graphs.}}~~For graphs, we can consider two types of DP: 
\textit{edge DP} and \textit{node DP} \cite{Hay_ICDM09,Raskhodnikova_Encyclopedia16}. 
Edge DP hides the existence of one edge, whereas node DP hides the existence of one node along with its adjacent edges. 
In this paper, we focus on edge DP because existing one-round local triangle counting algorithms \cite{Imola_USENIX21,Imola_USENIX22,Ye_ICDE20,Ye_TKDE21} use edge DP. 
In other words, we are interested in 
how much the estimation error is reduced at the same value of $\epsilon$ in edge DP by shuffling. 
Although node DP is much stronger than edge DP, it is much harder to attain and often results in a much larger $\epsilon$ \cite{Chen_SIGMOD13,Sajadmanesh_arXiv22}. 
Thus, we leave an algorithm for shuffle node DP with small $\epsilon$ (e.g., $\epsilon \leq 1$) for future work. 
Another interesting avenue of future work is establishing a lower bound on the estimation error for node DP. 

Edge DP assumes that anyone (except for user $v_i$) can be an adversary who infers edges of user $v_i$ and that the adversary can obtain all edges except for edges of $v_i$ as background knowledge. 
Note that the central and local models have different definitions of neighboring data in edge DP. 
Specifically, edge DP in the central model \cite{Raskhodnikova_Encyclopedia16} considers two graphs that differ in one edge. 
In contrast, edge LDP 
\cite{Qin_CCS17} considers two neighbor lists that differ in one bit: 

\begin{definition} [$(\epsilon,\delta)$-edge DP \cite{Raskhodnikova_Encyclopedia16}] \label{def:edge_DP} 
Let $n \in \nats$, $\epsilon \in \nnreals$, and $\delta \in [0,1]$. 
A randomized algorithm $\calM$ with domain $\calG$ provides \emph{$(\epsilon, \delta)$-edge DP} 
if for any two neighboring graphs $G, G' \in \calG$ that differ in \textbf{one edge} and any $S \subseteq \mathrm{Range}(\calM)$, 
\begin{align*}
\Pr[\calM(G) \in S] \leq e^\epsilon \Pr[\calM(G') \in S] + \delta.
\end{align*}
\end{definition}

\begin{definition} [$\epsilon$-edge LDP~\cite{Qin_CCS17}] \label{def:edge_LDP} 
Let $\epsilon \in \nnreals$. 
A local randomizer $\calR$ with domain $\{0,1\}$ provides \emph{$\epsilon$-edge LDP} if for any two neighbor lists $\bma_i, \bma'_i \in \{0,1\}^n$ that differ in \textbf{one bit} and any $S \subseteq \mathrm{Range}(\calR)$, 
\begin{align*}
\Pr[\calR(\bma_i) \in S] \leq e^\epsilon \Pr[\calR(\bma'_i) \in S].
\end{align*}
\end{definition}

As with edge LDP,  
we define \textit{element DP}, which considers two adjacency matrices that differ in one bit, in the central model:

\begin{definition} [$(\epsilon,\delta)$-element DP] \label{def:element_DP} 
Let $n \in \nats$, $\epsilon \in \nnreals$, and $\delta \in [0,1]$. 
A randomized algorithm $\calM$ with domain $\calG$ provides \emph{$(\epsilon, \delta)$-element DP} 
if for any two neighboring graphs $G, G' \in \calG$ that differ in \textbf{one bit} in the corresponding adjacency matrices $\bmA, \bmA' \in \{0,1\}^{n \times n}$
and any $S \subseteq \mathrm{Range}(\calM)$, 
\begin{align*}
\Pr[\calM(G) \in S] \leq e^\epsilon \Pr[\calM(G') \in S] + \delta.
\end{align*}
\end{definition}

Although element DP and edge DP have different definitions of neighboring data, they are closely related to each other:

\begin{proposition}\label{prop:element_edge_DP}
If a randomized algorithm $\calM$ provides $(\epsilon, \delta)$-element DP, it also provides $(2\epsilon, 2\delta)$-edge DP. 
\end{proposition}
\begin{proof}
Adding or removing one edge affects two bits in an adjacency matrix. 
Thus, by group privacy \cite{DP}, any 
$(\epsilon, \delta)$-element DP 
algorithm $\calM$ 
provides $(2\epsilon, 2\delta)$-edge DP. 
\end{proof}
Similarly, if 
a randomized algorithm $\calM$ in the central model 
applies a local randomizer $\calR$ providing $\epsilon$-edge LDP to each neighbor list $\bma_i$ ($1 \leq i \leq n$), it provides $2\epsilon$-edge DP \cite{Imola_USENIX21}. 

In this work, we use the shuffling technique to provide $(\epsilon, \delta)$-element DP and then Proposition~\ref{prop:element_edge_DP} to provide $(2\epsilon, 2\delta)$-edge DP. 
We also compare our shuffle algorithms providing $(\epsilon, \delta)$-element DP and $(2\epsilon, 2\delta)$-edge DP with local algorithms providing $\epsilon$-edge LDP and $2\epsilon$-edge DP to see how much the estimation error is reduced by introducing the shuffle model and a very small $\delta$ ($\ll \frac{1}{n}$). 

\subsection{Shuffle Model}
\label{sub:shuffle}
We consider the following shuffle model. 
Each user $v_i \in V$ obfuscates her personal data 
using a local randomizer $\calR$ providing $\epsilon_L$-LDP for $\epsilon_L \in \nnreals$. 
Note that $\calR$ is common to all users. 
User $v_i$ encrypts the obfuscated data and sends it to a shuffler. 
Then, the shuffler randomly shuffles the encrypted data and sends the results to a data collector. 
Finally, the data collector decrypts them. 
The common assumption in the shuffle model is that the shuffler and the data collector do not collude with each other. 
Under this assumption, the shuffler cannot access the obfuscated data, and the data collector cannot link the obfuscated data to the users. 
Hereinafter, we omit the encryption/decryption process because it is clear from the context. 

We use the privacy amplification result by Feldman \textit{et al.} \cite{Feldman_FOCS21}:
\begin{theorem} [Privacy amplification by shuffling \cite{Feldman_FOCS21}] \label{thm:shuffle}
Let $n \in \nats$ and $\epsilon_L \in \nnreals$. 
Let $\calX$ be the set of input data for each user. 
Let $x_i \in \calX$ be input data of the $i$-th user, and 
$x_{1:n} = (x_1, \cdots, x_n) \in \calX^n$. 
Let $\calR: \calX \rightarrow \calY$ be a local randomizer providing $\epsilon_L$-LDP. 
Let $\calM_S: \calX^n \rightarrow \calY^n$ be an algorithm that given a dataset $x_{1:n}$, computes $y_i = \calR(x_i)$ for $i \in [n]$, samples a uniform random permutation $\pi$ over $[n]$, and outputs $y_{\pi(1)}, \ldots, y_{\pi(n)}$. 
Then for any $\delta \in [0,1]$ such that $\epsilon_L \leq \log (\frac{n}{16 \log (2/\delta)})$, $\calM_S$ provides $(\epsilon, \delta)$-DP, where
\begin{align}
\epsilon = f(n, \epsilon_L, \delta)
\label{eq:shuffle_epsilon_f}
\end{align}
and 
\begin{align}
f(n, \epsilon_L, \delta) = \log \left( 1 + \frac{e^{\epsilon_L}-1}{e^{\epsilon_L}+1} \left( \frac{8\sqrt{e^{\epsilon_L} \log(4/\delta)}}{\sqrt{n}} + \frac{8 e^{\epsilon_L}}{n} \right) \right).
\label{eq:shuffle_epsilon}
\end{align}
\end{theorem}
Thanks to the shuffling, the shuffled data $y_{\pi(1)}, \ldots, y_{\pi(n)}$ available to the data collector provides $(\epsilon, \delta)$-DP, where $\epsilon \ll \epsilon_L$. 

Feldman \textit{et al.} \cite{Feldman_FOCS21} also propose an efficient method to numerically compute a tighter upper bound than the closed-form upper bound in Theorem~\ref{thm:shuffle}. 
We use both the closed-form and numerical upper bounds in our experiments. 
Specifically, we use the numerical upper bounds in Section~\ref{sec:experiments} and compare the numerical bound with the closed-form bound in 
\conference{the full version \cite{Imola_CCSFull22}}\arxiv{Appendix~\ref{sec:numerical_closed}}. 

Assume that $\epsilon$ and $\delta$ in (\ref{eq:shuffle_epsilon}) are constants. 
Then, by solving for $\epsilon_L$ and changing to big $O$ notation, we obtain  $\epsilon_L = \log(n) + O(1)$. 
This is consistent with the upper bound $\epsilon = O(e^{\epsilon_L / 2} / \sqrt{n})$ in \cite{Feldman_FOCS21}, from which we obtain $\epsilon_L = \log(n) + O(1)$. 
Similarly, the privacy amplification bound in \cite{Cheu_EUROCRYPT19} can also be expressed as $\epsilon_L = \log(n) + O(1)$. 
We use the bound in \cite{Feldman_FOCS21} because it is the state-of-the-art, as described in Section~\ref{sec:related}.

\subsection{Utility Metrics}
\label{sub:utility}
We use 
the MSE (Mean Squared Error) 
in our theoretical analysis and the relative error 
in our experiments. 
The MSE is the expectation of the squared error 
between a true value and its estimate. 
Let $f: \calG \rightarrow \nnints$ be a subgraph count function that can be instantiated by $f^\triangle$ or $f^\square$. 
Let $\hf: \calG \rightarrow \reals$ be the corresponding estimator. 
Let 
$\MSE: \reals \rightarrow \nnreals$ be the MSE function, which maps the estimate $\hf(G)$ to the MSE. 
Then the MSE can be expressed as $\MSE(\hf(G)) = \E[(f(G) - \hf(G))^2]$, 
where the expectation is taken over the randomness in the estimator $\hf$. 
By the bias-variance decomposition \cite{mlpp}, 
the MSE can be expressed as a summation of the squared bias $(\E[\hf(G)] - f(G))^2$ and the variance $\V[\hf(G)] = \E[(\hf(G) - \E[\hf(G)])]^2$. 
Thus, for an unbiased estimator $\hf$ satisfying $\E[\hf(G)] = f(G)$, the MSE is equal to the variance, i.e., $\MSE(\hf(G)) = \V[\hf(G)]$. 

Although the MSE is suitable for theoretical analysis, it tends to be large when the number $n$ of users is large. 
This is because the true triangle and 4-cycle counts are very large when $n$ is large -- $f^\triangle(G) = O(n d_{max}^2)$ and $f^\square(G) = O(n d_{max}^3)$. 
Therefore, we use the relative error in our experiments. 
The relative error is an absolute error divided by the true value and is given by $\frac{|f^\triangle(G) - \hf^\triangle(G)|}{\min\{f^\triangle(G), \eta\}}$, where $\eta \in \nnreals$ is a small positive value. 
Following the convention \cite{Bindschaedler_SP16,Chen_CCS12,Xiao_SIGMOD11}, we set $\eta = \frac{n}{1000}$. 

When the relative error is well below $1$, the estimate is accurate. 
Note that the absolute error smaller than $1$ would be impossible under DP with meaningful $\epsilon$ (e.g., $\epsilon \leq 1$), as we consider counting queries. 
However, the relative error ($=$ absolute error / true count) much smaller than $1$ is possible under DP with meaningful $\epsilon$. 

\section{Shuffle Model for Graphs}
\label{sec:shuffle}
In this work, we apply the shuffle model to graph data to accurately estimate subgraph counts, such as triangles and 4-cycles. 
Section~\ref{sub:technical} explains our technical motivation. 
In particular, we explain why it is challenging to apply the shuffle model to graph data. 
Section~\ref{sub:wedge_shuffle} proposes a wedge shuffle technique to overcome the technical challenge. 

\subsection{Our Technical Motivation}
\label{sub:technical}

\begin{figure}[t]
  \centering
  \includegraphics[width=0.99\linewidth]{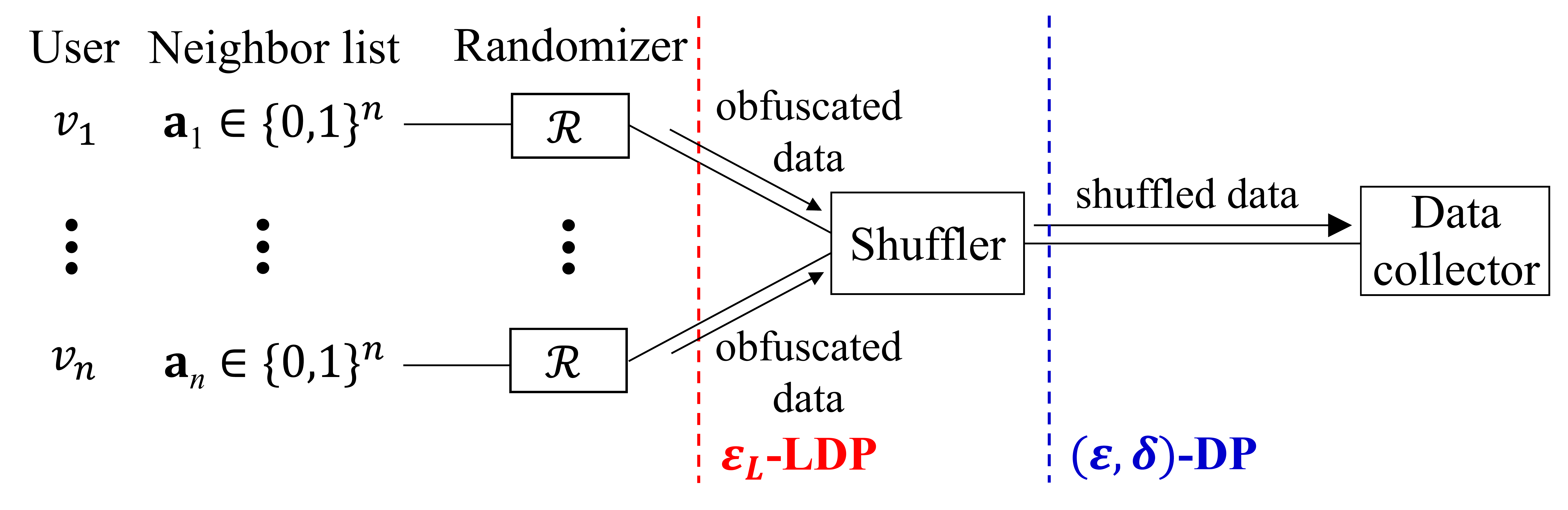}
  \vspace{-2mm}
  \caption{Shuffle model for graphs. 
  }
  \label{fig:shuffle_model}
\end{figure}

The shuffle model has been introduced to dramatically reduce the privacy budget $\epsilon$ (hence the estimation error at the same $\epsilon$) in tabular data \cite{Meehan_ICLR22,Wang_PVLDB20} or 
gradients 
\cite{Girgis_AISTATS21,Liu_AAAI21}. 
However, it is very challenging to apply the shuffle model to graph data, 
as explained below. 

Figure~\ref{fig:shuffle_model} shows the shuffle model for graph data, where each user $v_i$ has her neighbor list $\bma_i \in \{0,1\}^n$. 
The main challenge here 
is that 
the shuffle model uses a \textit{standard} definition of LDP for the local randomizer and that a neighbor list is \textit{high-dimensional data}, i.e., $n$-dim binary string. 
Specifically, LDP in Definition~\ref{def:LDP} requires any pair of inputs $x$ and $x'$ to be indistinguishable; i.e., 
the inequality (\ref{eq:LDP}) must hold for all pairs of possible inputs. 
Thus, if we use the entire neighbor list 
as input data (i.e., $\bma_i = x_i$ in Theorem~\ref{thm:shuffle}), 
either privacy or utility is destroyed for large $n$. 

To illustrate this, consider the following example. 
Assume that $n=10^5$ and $\delta=10^{-8}$. 
Each user $v_i$ applies 
$\epsilon_0$-RR 
with $\epsilon_0=1$ to each bit of her neighbor list $\bma_i$. 
This mechanism is called the randomized neighbor list \cite{Qin_CCS17} and provides $\epsilon_0$-edge LDP. 
However, the privacy budget $\epsilon_L$ in the standard LDP (Definition~\ref{def:LDP}) 
is extremely large -- by group privacy \cite{DP}, $\epsilon_L = n \epsilon_0 = 10^5$. 
Because 
$\epsilon_L$ 
is much larger than $\log (\frac{n}{16 \log (2/\delta)}) = 8.09$, we cannot use the privacy amplification result in Theorem~\ref{thm:shuffle}. 
This is evident from the fact that the shuffled data $y_{\pi(1)}, \ldots, y_{\pi(n)}$ are easily re-identified when $n$ is large. 
If we use 
$\epsilon_0$-RR 
with $\epsilon_0 = \frac{1}{n}$, we can use the amplification result (as $\epsilon_L = n \epsilon_0 = 1$). 
However, it makes obfuscated data almost a random string and destroys the utility because $\epsilon_0$ is too small. 

In this work, we address this issue by introducing a basic technique, which we call \textit{wedge shuffling}. 

\subsection{Our Approach: Wedge Shuffling}
\label{sub:wedge_shuffle}

Figure ~\ref{fig:wedge_shuffle} shows the overview of our wedge shuffle technique. 
This technique calculates the number of wedges (2-hop paths) between 
a specific pair of users $v_i$ and $v_j$. 

Algorithm~\ref{alg:WShuffle} shows our wedge shuffle algorithm, which we call \AlgWS{}. 
Given users $v_i$ and $v_j$, 
each of the remaining users $v_k$ ($k \ne i, j$) 
calculates a \textit{wedge indicator} $w_{i-k-j} = a_{k,i} a_{k,j}$, 
which 
takes $1$ if 
a wedge $v_i$-$v_k$-$v_j$ exists and $0$ otherwise (line 2). 
Then, $v_k$ obfuscates $w_{i-k-j}$ using $\epsilon_L$-RR and sends it to the shuffler (line 3). 
The shuffler randomly shuffles the noisy wedges using a random permutation $\pi$ over 
$I_{-(i,j)}$ ($=[n]\setminus\{i,j\}$) 
to provide $(\epsilon, \delta)$-DP with $\epsilon \ll \epsilon_L$ (line 5). 
Finally, the shuffler sends the shuffled wedges to the data collector (line 6). 
The only information available to the data collector is  the number of wedges from $v_i$ to $v_j$, i.e., 
common friends of $v_i$ and $v_j$. 

Our wedge shuffling has two main features. 
First, 
the wedge indicator $w_{i-k-j}$ is \textit{one-dimensional binary data}. 
Therefore, it can be sent with small noise and small $\epsilon$, unlike the $n$-dimensional neighbor list. 
For example, when $n=10^5$, $\delta=10^{-8}$, and $\epsilon=1$, the value of $\epsilon_L$ in (\ref{eq:shuffle_epsilon_f}) and (\ref{eq:shuffle_epsilon}) is $\epsilon_L = 5.44$. In this case, $\epsilon_L$-RR rarely flips $w_{i-k-j}$ -- the flip probability is $0.0043$. 
In other words, the shuffled wedges are almost free of noise. 

\begin{figure}[t]
  \centering
  \includegraphics[width=0.99\linewidth]{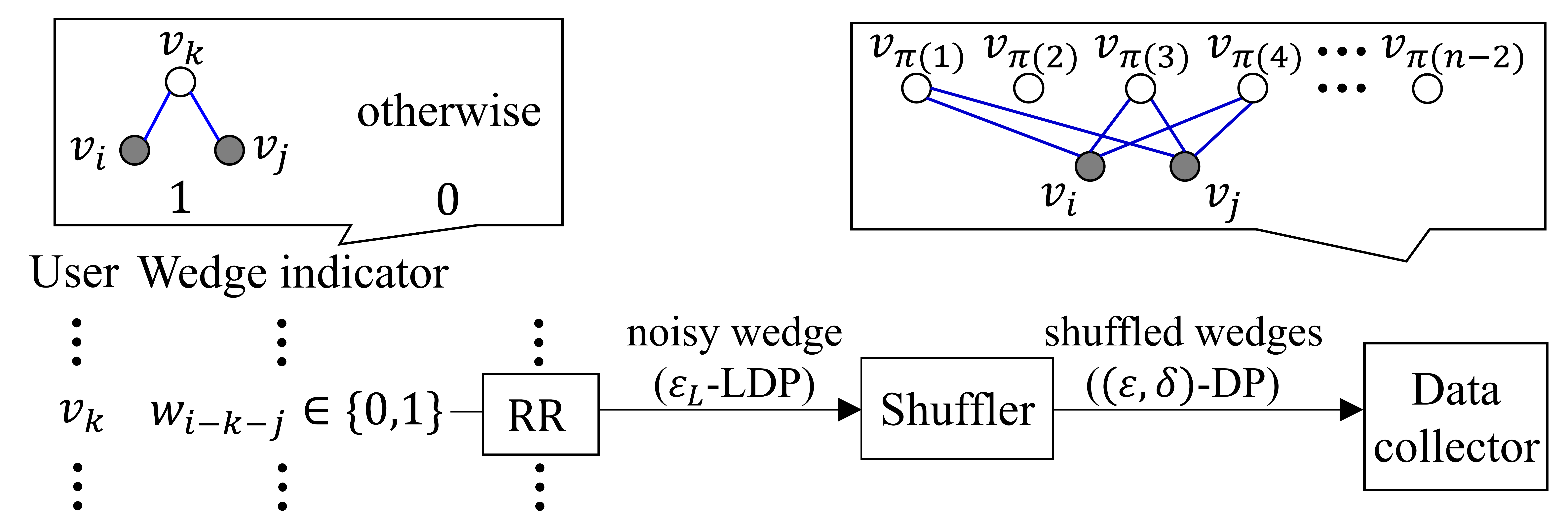}
  \vspace{-2mm}
  \caption{Overview of wedge shuffling with inputs $v_i$ and $v_j$. 
  }
  \label{fig:wedge_shuffle}
\end{figure}

\setlength{\algomargin}{5mm}
\begin{algorithm}[t]
  \SetAlgoLined
  \KwData{Adjacency matrix $\bmA \in \{0,1\}^{n \times n}$, $\epsilon_L \in \nnreals$, user-pair $(v_i,v_j)$.}
  \KwResult{Shuffled wedges $\{y_{\pi(k)} | k \in I_{-(i,j)}\}$.}
  \ForEach{$k \in I_{-(i,j)}$}{
    [$v_k$] $w_{i-k-j} \leftarrow a_{k,i} a_{k,j}$\;
    [$v_k$] $y_k \leftarrow \calR_{\epsilon_L}^W(x)(w_{i-k-j})$; Send $y_k$ to the shuffler\;
  }
  [s] Sample a random permutation $\pi$ over $I_{-(i,j)}$\;
  [s] Send $\{y_{\pi(k)} | k \in I_{-(i,j)}\}$ to the data collector\;
  [d] \KwRet{$\{y_{\pi(k)} | k \in I_{-(i,j)}\}$}
  \caption{Our wedge shuffle algorithm \AlgWS{}. 
  [$v_k$], [s], and [d] represent that the process is run by user $v_i$, the shuffler, and the data collector, respectively. 
  }\label{alg:WShuffle}
\end{algorithm}

Second, the wedge is the main component of many subgraphs such as triangles, 
$k$-triangles \cite{Karwa_PVLDB11}, 
3-hop paths \cite{Sun_CCS19}, 
and 4-cycles. 
For example, a triangle consists of one wedge and one edge, e.g., $v_i-v_{\pi(3)}-v_j$ and $(v_i, v_j)$ in Figure ~\ref{fig:wedge_shuffle}. 
More generally, a $k$-triangle consists of $k$ triangles sharing one edge. 
Thus, it can be decomposed into $k$ wedges and one edge. 
A 3-hop path consists of one wedge and one edge. 
A 4-cycle consists of two wedges, e.g., $v_i-v_{\pi(1)}-v_j$ and $v_i-v_{\pi(3)}-v_j$ in Figure ~\ref{fig:wedge_shuffle}. 
Because the shuffled wedges have little noise, we can accurately count these subgraphs based on wedge shuffling, compared to local algorithms in which all edges are noisy. 

In this work, we focus on triangles and 4-cycles and 
present algorithms with upper bounds on the estimation error based on our wedge shuffle technique.

\section{Triangle Counting Based on Wedge Shuffling}
\label{sec:triangle}
Based on our wedge shuffle technique, we first propose a
one-round triangle counting algorithm.
Section~\ref{sub:triangle_overview} describes the overview of our algorithms.
Section~\ref{sub:wedge} proposes an algorithm for counting triangles involving a specific user-pair
as a building block of our triangle counting algorithm.
Section~\ref{sub:triangle} proposes
our triangle counting algorithm.
Section~\ref{sub:var_red} proposes a technique to significantly reduce the variance in our triangle counting algorithm. 
Section~\ref{sub:summary} summarizes the performance guarantees of our triangle algorithms.

\subsection{Overview}
\label{sub:triangle_overview}
Our wedge shuffle technique tells the data collector the number of common friends of $v_i$ and $v_j$.
However, this information is not sufficient to count triangles in the entire graph.
Therefore, we introduce
three additional techniques:
\textit{(i) sending local edges},
\textit{(ii) sampling disjoint user-pairs},
and
\textit{(iii) variance reduction by ignoring sparse user-pairs}.
Below, we briefly explain each technique.

\begin{figure}[t]
  \centering
  \includegraphics[width=0.99\linewidth]{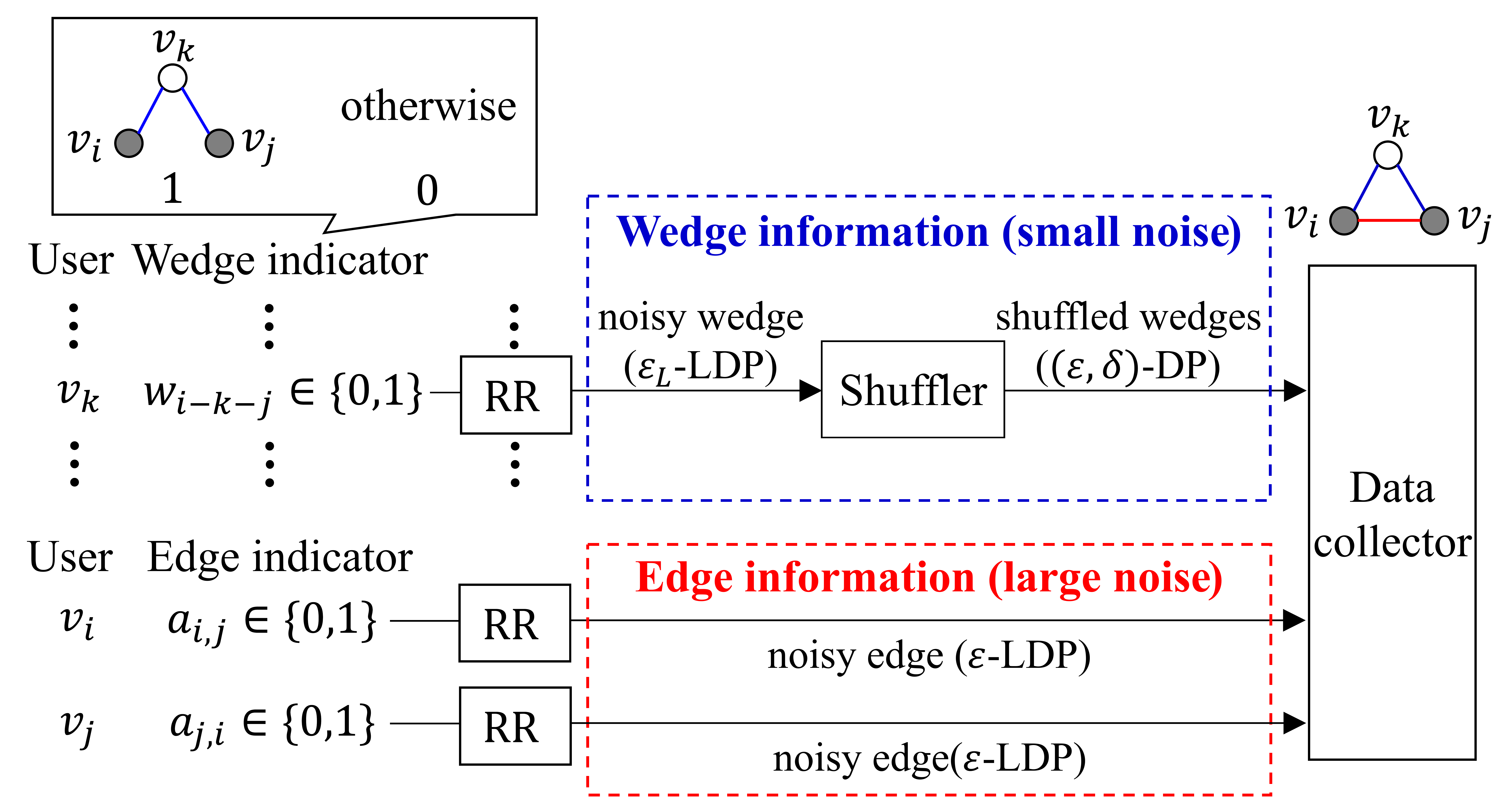}
  \vspace{-2mm}
  \caption{Overview of our \AlgWSLE{} (Wedge Shuffling with Local Edges) algorithm with inputs $v_i$ and $v_j$.
  }
  \label{fig:local_edges}
\vspace{2mm}
  \centering
  \includegraphics[width=0.7\linewidth]{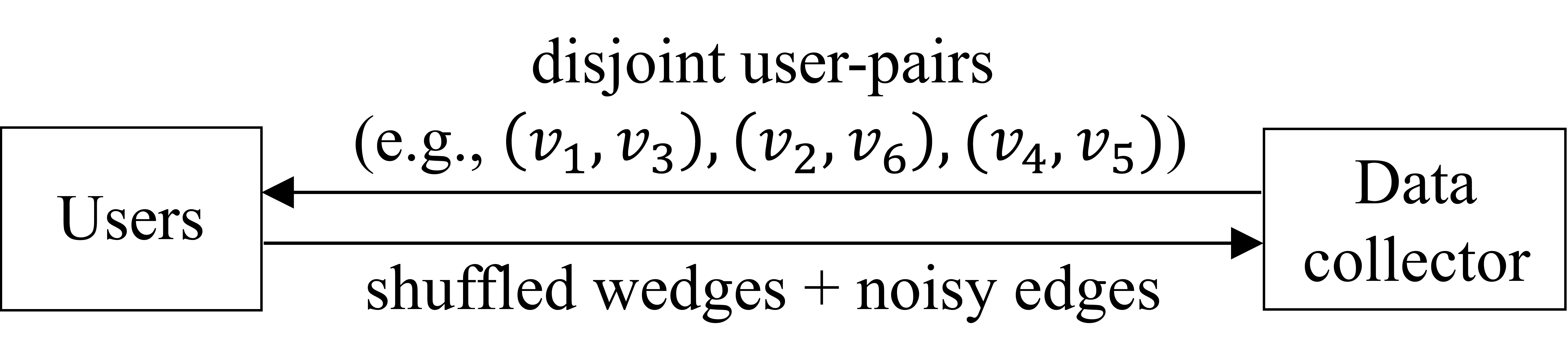}
  \vspace{-2mm}
  \caption{Overview of our triangle counting algorithm.
  We use our
  \AlgWSLE{}
  algorithm with each user-pair.
  }
  \label{fig:triangle_count}
\end{figure}

\smallskip
\noindent{\textbf{Sending Local Edges.}}~~First, we consider the problem of counting triangles involving a specific user-pair $(v_i, v_j)$ and propose an algorithm to send
\textit{local edges} between $v_i$ and $v_j$, along with shuffled wedges, to the data collector.
We call this the \AlgWSLE{} (Wedge Shuffling with Local Edges) algorithm.

Figure~\ref{fig:local_edges} shows the overview of \AlgWSLE{}.
In this algorithm, users $v_i$ and $v_j$ obfuscate edge indicators $a_{i,j}$ and $a_{j,i}$, respectively, using $\epsilon$-RR and send them to the data collector directly (or through the shuffler without shuffling).
Then, the data collector calculates an unbiased estimate of the triangle count
from the shuffled wedges and the noisy edges.
Because $\epsilon$ is small, a large amount of noise is added to the edge indicators.
However, \textit{only one edge} is noisy (the other two have little noise) in any triangle the data collector sees.
This brings us an advantage over the one-round local algorithms in which all three edges are noisy.

\smallskip
\noindent{\textbf{Sampling Disjoint User-Pairs.}}~~Next, we consider the problem of counting triangles in the entire graph $G$.
A naive solution to this problem is to use our
\AlgWSLE{} algorithm
with all $\binom{n}{2}$ user-pairs as input.
However, it results in very large $\epsilon$ and $\delta$ because it uses each element of the adjacency matrix $\bmA$ many times.
To address this issue,
we propose a triangle counting algorithm that samples
disjoint user-pairs, ensuring that no user falls in two pairs. 

Figure~\ref{fig:triangle_count} shows the overview of our triangle algorithm.
The data collector sends the sampled user-pairs to users.
Then, users apply
\AlgWSLE{}
with each user-pair
and send the results to the data collector.
Finally, the data collector calculates an unbiased estimate of the triangle count from the results.
Because our triangle algorithm uses each element of
$\bmA$ \textit{at most once}, it provides $(\epsilon,\delta)$-element DP hence $(2\epsilon,2\delta)$-edge DP.
In addition, our triangle algorithm reduces the time complexity from $O(n^3)$ to $O(n^2)$ by sampling user-pairs rather than using all user-pairs.

We prove that the MSE of our triangle counting algorithm is
$O(n^3)$ when we ignore the factor of $d_{max}$.
When we do not shuffle wedges, the
MSE is
$O(n^4)$.
In addition, the MSE of the existing one-round local algorithm \cite{Imola_USENIX22} with the same time complexity is $O(n^6)$, as proved in
\conference{the full version \cite{Imola_CCSFull22}}\arxiv{
Appendix~\ref{sec:upper}}.
Thus, our
algorithm
provides a dramatic improvement over the local algorithms.

\smallskip
\noindent{\textbf{Variance Reduction.}}~~Although our
algorithm
dramatically
improves
the MSE, the factor of $n^3$ may still be large.
Therefore, we propose a variance reduction technique that ignores sparse user-pairs, where either of the two users has a very small degree.
Our basic idea is that the number of triangles involving such a user-pair is very small
and can be approximated by $0$.
By ignoring the sparse user-pairs, we can significantly reduce the variance at the cost of introducing a small bias.
We prove that our variance reduction technique reduces the MSE from $O(n^3)$ to $O(n^\gamma)$ where $\gamma\in[2,3)$ and makes one-round triangle counting more accurate.

\subsection{WSLE (Wedge Shuffling with Local Edges)}
\label{sub:wedge}

\noindent{\textbf{Algorithm.}}~~We first
propose
the \AlgWSLE{} algorithm
as a building block of our triangle counting algorithm.
\AlgWSLE{} counts
triangles involving a specific user-pair $(v_i,v_j)$.

Algorithm~\ref{alg:WSLE} shows \AlgWSLE{}.
Let $f_{i,j}^\triangle: \calG \rightarrow \nnints$ be a function that takes $G \in \calG$ as input and outputs the number $f_{i,j}^\triangle(G)$ of triangles involving $(v_i,v_j)$ in $G$.
Let $\hf_{i,j}^\triangle(G) \in \reals$ be an estimate of $f_{i,j}^\triangle(G)$.

\setlength{\algomargin}{5mm}
\begin{algorithm}[t]
  \SetAlgoLined
  \KwData{Adjacency matrix $\bmA \in \{0,1\}^{n \times n}$,
    $\epsilon \in \nnreals$, $\delta \in [0,1]$,
    user-pair $(v_i,v_j)$.
  }
  \KwResult{Estimate $\hf_{i,j}^\triangle(G)$ of the number $f_{i,j}^\triangle(G)$ of triangles involving $(v_i,v_j)$.}
  $\epsilon_L \leftarrow \texttt{LocalPrivacyBudget}(n,\epsilon,\delta)$\;
  \tcc{Wedge shuffling}
  $\{y_{\pi(k)} | k \in I_{-(i,j)}\} \leftarrow$ \AlgWS{}$(\bmA, \epsilon_L, (v_i, v_j))$\;
  \tcc{Send local edges}
  [$v_i$] $z_i \leftarrow \calR_{\epsilon}^W(x)(a_{i,j})$; Send $z_i$ to the data collector\;
  [$v_j$] $z_j \leftarrow \calR_{\epsilon}^W(x)(a_{j,i})$; Send $z_j$ to the data collector\;
  \tcc{Calculate an unbiased estimate}
  [d] $q_L \leftarrow \frac{1}{e^{\epsilon_L}+1}$; $q \leftarrow
  \frac{1}{e^\epsilon+1}$\;
  [d] $\hf_{i,j}^\triangle(G) \leftarrow \frac{(z_i+z_j-2q)\sum_{k \in
  I_{-(i,j)}} (y_{k} - q_L)}{2(1-2q)(1-2q_L)}$\;
  [d] \KwRet{$\hf_{i,j}^\triangle(G)$}
  \caption{\AlgWSLE{}
  (Wedge Shuffling with Local Edges).
  \AlgWS{} is shown in Algorithm~\ref{alg:WShuffle}.
  }\label{alg:WSLE}
\end{algorithm}

We first call the function \texttt{LocalPrivacyBudget}, which calculates a local privacy budget $\epsilon_L$ from $n$, $\epsilon$, and $\delta$ (line 1).
Specifically, this function calculates $\epsilon_L$
such that $\epsilon$ is a closed-form upper bound (i.e., $\epsilon = f(n-2, \epsilon_L, \delta)$ in (\ref{eq:shuffle_epsilon_f})) or numerical upper bound in the shuffle model with $n-2$ users.
Given $\epsilon_L$, we can easily calculate the closed-form or numerical upper bound $\epsilon$ by (\ref{eq:shuffle_epsilon}) and the open source code in \cite{Feldman_FOCS21}\footnote{\url{https://github.com/apple/ml-shuffling-amplification}.}, respectively.
Thus, we can also easily calculate $\epsilon_L$ from $\epsilon$ by calculating a lookup table for pairs $(\epsilon, \epsilon_L)$ in advance.

Then, we run our wedge shuffle algorithm \AlgWS{} in Algorithm~\ref{alg:WShuffle} (line 2); i.e., each user $v_k \in I_{-(i,j)}$ sends her obfuscated wedge indicator
$y_k = \calR_{\epsilon_L}^W(w_{i-k-j})$ to the shuffler, and the shuffler sends
shuffled wedge indicators $\{y_{\pi(k)} | k \in I_{-(i,j)}\}$ to the data collector.
Meanwhile, user $v_i$ obfuscates her edge indicator $a_{i,j}$ using $\epsilon$-RR $\calR_{\epsilon}^W$ and sends the result $z_i = \calR_{\epsilon}^W(a_{i,j})$ to the data collector
(line 3).
Similarly, $v_j$ sends $z_j = \calR_{\epsilon}^W(a_{j,i})$ to the data collector (line 4).

Finally, the data collector estimates $f_{i,j}^\triangle(G)$ from $\{y_{\pi(k)} | k \in I_{-(i,j)}\}$, $z_i$, and $z_j$.
Specifically, the data collector calculates the estimate $\hf_{i,j}^\triangle(G)$ as follows:
\begin{align}
    \textstyle{\hf_{i,j}^\triangle(G) = \frac{(z_i+z_j-2q)\sum_{k \in I_{-(i,j)}} (y_{k} -
    q_L)}{2(1-2q)(1-2q_L)},}
    \label{eq:hfij_triangle}
\end{align}
where $q_L = \frac{1}{e^{\epsilon_L}+1}$ and $q = \frac{1}{e^\epsilon+1}$ (lines 5-6).
Note that this estimate involves simply summing over the set $\{y_{\pi(k)}\}$ and does not require knowing the value of $\pi$. This is
consistent with the shuffle model.
As we prove later, $\hf_{i,j}^\triangle(G)$ in (\ref{eq:hfij_triangle}) is an unbiased estimate of $f_{i,j}^\triangle(G)$.

\smallskip
\noindent{\textbf{Theoretical Properties.}}~~Below, we show some theoretical properties of \AlgWSLE{}.
First, we prove that
the estimate $\hf_{i,j}^\triangle(G)$
is unbiased:
\begin{theorem}
\label{thm:unbiased_I}
  For any indices $i,j \in [n]$, the estimate produced by
  $\AlgWSLE{}$ satisfies $\E[\hf_{i,j}^\triangle(G)] = f_{i,j}^\triangle(G)$.
\end{theorem}

Next,
we show the MSE ($=$ variance). 
Recall that in the shuffle model, $\epsilon_L = \log n + O(1)$ when $\epsilon$ and $\delta$ are constants.
We show
the MSE for a general case and for the shuffle model:
\begin{theorem}
\label{thm:l2-loss_I}
  For any indices $i,j \in [n]$, the estimate produced by
  \AlgWSLE{} provides the following utility guarantee:
\begin{align}
& \MSE(\hf_{i,j}^\triangle(G)) = \V[\hf_{i,j}^\triangle(G)] \nonumber\\
  & \leq \frac{n q_L + q(1-2q_L)^2 d_{max}^2}{(1-2q)^2(1-2q_L)^2} \triangleq
  err_{\AlgWSLE}(n,d_{max},q,q_L).
\label{eq:l2_I_general}
\end{align}
  When $\epsilon$ and $\delta$ are constants and $\epsilon_L = \log n + O(1)$, we
  have
\begin{align}
  & err_{\AlgWSLE{}}(n,d_{max}, q, q_L) = O(d_{max}^2).
\label{eq:l2_I_shuffle}
\end{align}
\end{theorem}
The equation (\ref{eq:l2_I_shuffle}) follows from (\ref{eq:l2_I_general}) because $q_L = \frac{1}{e^{\epsilon_L}+1} = \frac{1}{n e^{O(1)} + 1}$. 
Because \AlgWSLE{} is a building block for our triangle counting algorithms, we
introduce the notation $err_{\AlgWSLE{}}(n,d_{max}, q, q_L)$ for our upper bound
in~\eqref{eq:l2_I_general}. Observing
\eqref{eq:l2_I_general}, if we do not use the shuffling technique (i.e.,
$\epsilon_L = \epsilon$), then $err_{\AlgWSLE{}}(n,d_{max}, q, q_L) = O(n + d_{max}^2)$ when we
treat $\epsilon$ and $\delta$ as constants.
In contrast, in the shuffle model where we have $\epsilon_L = \log n + O(1)$,
then $err_{\AlgWSLE{}}(n,d_{max}, q, q_L) = O(d_{max}^2)$.
This means that wedge shuffling reduces the MSE from $O(n + d_{max}^2)$ to $O(d_{max}^2)$, which is significant when $d_{max} \ll n$.

\subsection{Triangle Counting}
\label{sub:triangle}
\noindent{\textbf{Algorithm.}}~~Based on
\AlgWSLE{},
we propose an algorithm that counts triangles in the entire graph $G$.
We denote this algorithm by \AlgWSTri{}, as it applies wedge shuffling to triangle counting.

Algorithm~\ref{alg:wshuffle_triangle} shows \AlgWSTri{}.
First, the data collector samples disjoint user-pairs, 
ensuring that no user falls in two pairs. 
Specifically, it calls the function \texttt{RandomPermutation}, which samples a uniform random permutation $\sigma$ over $[n]$ (line 1).
Then, it samples disjoint user-pairs as
$(v_{\sigma(1)}, v_{\sigma(2)}), (v_{\sigma(3)}, v_{\sigma(4)}), \ldots, (v_{\sigma(2t-1)}, \allowbreak v_{\sigma(2t)})$, where $t \in [\lfloor \frac{n}{2} \rfloor]$.
The parameter $t$ represents the number of user-pairs and controls the trade-off between the MSE and the time complexity;
when $t = \lfloor \frac{n}{2} \rfloor$, the MSE is minimized and the time complexity is maximized.
The data collector sends the sampled user-pairs to users (line 2).

\setlength{\algomargin}{5mm}
\begin{algorithm}[t]
  \SetAlgoLined
  \KwData{Adjacency matrix $\bmA \in \{0,1\}^{n \times n}$, $\epsilon \in \nnreals$, $\delta \in [0,1]$, $t \in [\lfloor \frac{n}{2} \rfloor]$.
  }
  \KwResult{Estimate $\hf^\triangle(G)$ of $f^\triangle(G)$.}
  \tcc{Sample disjoint user-pairs}
  [d] $\sigma \leftarrow$\texttt{RandomPermutation}$(n)$\;
  [d] Send $(v_{\sigma(1)}, v_{\sigma(2)}), \ldots, (v_{\sigma(2t-1)}, v_{\sigma(2t)})$ to users\;
  \ForEach{$i \in \{1, 3, \ldots, 2t-1\}$}{
    $\hf_{\sigma(i), \sigma(i+1)}^\triangle(G) \leftarrow$ \AlgWSLE{}$(\bmA, \epsilon, \delta, (v_{\sigma(i)}, v_{\sigma(i+1)}))$\;
  }
  \tcc{Calculate an unbiased estimate}
  [d] $\hf^\triangle(G) \leftarrow \frac{n(n-1)}{6t} \sum_{i=1, 3, \ldots, 2t-1} \hf_{\sigma(i),\sigma(i+1)}^\triangle(G)$\;
  [d] \KwRet{$\hf^\triangle(G)$}
  \caption{Our triangle counting algorithm \AlgWSTri{}.
  \AlgWSLE{} is shown in Algorithm~\ref{alg:WSLE}.
  }\label{alg:wshuffle_triangle}
\end{algorithm}

Then, we run our wedge algorithm \AlgWSLE{} in Algorithm~\ref{alg:WSLE} with each sampled user-pair as input (lines 3-5).
Finally, the data collector estimates the triangle count $f^\triangle(G)$ as follows:
\begin{align}
    \textstyle{\hf^\triangle(G) = \frac{n(n-1)}{6t} \sum_{i=1, 3, \ldots, 2t-1} \hf_{\sigma(i),\sigma(i+1)}^\triangle(G)}
   \label{eq:hf_triangle_II}
\end{align}
(line 6). 
Note that a single triangle is never counted by more than one user-pair, as the user-pairs never overlap. 
Later, we prove that $\hf^\triangle(G)$ in (\ref{eq:hf_triangle_II}) is unbiased.

\smallskip
\noindent{\textbf{Theoretical Properties.}}~~We prove that
\AlgWSTri{} provides DP:
\begin{theorem}
\label{thm:DP_II}
\AlgWSTri{} provides $(\epsilon, \delta)$-element DP and $(2\epsilon, 2\delta)$-edge DP.
\end{theorem}
Theorem~\ref{thm:DP_II} comes from the fact that
\AlgWSLE{} with a user-pair $(v_i,v_j)$ provides $(\epsilon,\delta)$-DP for each element in the $i$-th and $j$-th columns of the adjacency matrix $\bmA$ and that \AlgWSTri{} samples disjoint user-pairs, i.e., it uses each element of $\bmA$ at most once.

Note that running \AlgWSLE{} with all $\binom{n}{2}$ user-pairs 
provides $((n-2) \epsilon, (n-2) \delta)$-DP, as it uses each element of $\bmA$ at most $n-2$ times.
The privacy budget is very large, even using the advanced composition \cite{DP,Kairouz_ICML15}.
We avoid this issue by sampling user-pairs that share no common users.

We also prove that
\AlgWSTri{} provides an unbiased estimate:
\begin{theorem}
\label{thm:unbiased_II}
The estimate produced by \AlgWSTri{} satisfies $\allowbreak \E[\hf^\triangle(G)] = f^\triangle(G)$.
\end{theorem}

Next, we analyze the MSE ($=$ variance) of \AlgWSTri{}.
This analysis is 
non-trivial 
because \AlgWSTri{} samples each user-pair \textit{without replacement}.
In this case, the sampled user-pairs are not independent.
However, we can prove that
$t$ estimates in (\ref{eq:hf_triangle_II}) are negatively correlated with each other (\conference{see the full version \cite{Imola_CCSFull22} for details}\arxiv{Lemma~\ref{lem:sampling_replacement_var} in Appendix~\ref{sub:l2-loss_II_proof}}).
Thus, the variance of the sum of $t$ estimates
in (\ref{eq:hf_triangle_II}) is upper bounded by the sum of their variances, each of which is given by Theorem~\ref{thm:l2-loss_I}.
This brings us to the following result:

\begin{theorem}
\label{thm:l2-loss_II}
The estimate produced by \AlgWSTri{} provides the following utility guarantee:
\begin{align}
& \MSE(\hf^\triangle(G)) = \V[\hf^\triangle(G)] \nonumber\\
  &\leq
  \frac{n^4}{36t}err_{\AlgWSLE{}}(n,d_{max},q,q_L) +
  \frac{n^3}{36t}d_{max}^3, \label{eq:thm:l2_II}
\end{align}
where $err_{\AlgWSLE{}}(n,d_{max},q,q_L)$ is given by (\ref{eq:l2_I_general}).
  When $\epsilon$ and $\delta$ are constants,  $\epsilon_L = \log(n) + O(1)$, and
  $t = \lfloor\frac{n}{2}\rfloor$, we have
\begin{align}
\MSE(\hf^\triangle(G))
  &\leq
  O(n^3d_{max}^2). \label{eq:thm:l2_II_simp}
\end{align}
\end{theorem}

The inequality (\ref{eq:thm:l2_II_simp}) follows from (\ref{eq:l2_I_shuffle}) and (\ref{eq:thm:l2_II}). 
The first and second terms in~(\ref{eq:thm:l2_II}) are caused by
Warner's RR
and the sampling of disjoint user-pairs, respectively.
In other words, the MSE of \AlgWSTri{} can be decomposed into two factors: the RR
and user-pair sampling.

For example, assume that $t = \lfloor \frac{n}{2} \rfloor$.
When we do not shuffle wedges (i.e., $\epsilon_L = \epsilon$), then
$err_{\AlgWSLE{}}(n,d_{max},q,q_L) = O(n + d_{max}^2)$, and
MSE in (\ref{eq:thm:l2_II}) is $O(n^4 + n^3 d_{max}^2)$.
When we shuffle wedges, the MSE is $O(n^3 d_{max}^2)$.
Thus, when we ignore the factor of $d_{max}$, our wedge shuffle technique reduces the MSE from $O(n^4)$ to $O(n^3)$ in triangle counting.
The factor of $n^3$ is caused by the RR for local edges.
This is intuitive because a large amount of noise is added to the local edges.

Finally, we analyze the time complexity of \AlgWSTri{}.
The time complexity of running \AlgWSLE{} with all $\binom{n}{2}$ user-pairs is $O(n^3)$, as there are $O(n^2)$ user-pairs in total and \AlgWSLE{} requires the time complexity of $O(n)$.
In contrast, the time complexity of \AlgWSTri{} with $t = \lfloor \frac{n}{2} \rfloor$ is $O(n^2)$ because it samples $O(n)$ user-pairs.
Thus, \AlgWSTri{} reduces the time complexity from $O(n^3)$ to $O(n^2)$ by user-pair sampling.
We can further reduce the time complexity at the cost of increasing the MSE by setting $t$ small,
i.e., $t \ll \lfloor \frac{n}{2} \rfloor$.

\subsection{Variance Reduction}
\label{sub:var_red}
\noindent{\textbf{Algorithm.}}~~\AlgWSTri{}
achieves the MSE of $O(n^3)$ when we ignore the factor of $d_{max}$.
To provide a smaller estimation error,
we propose a variance reduction technique that ignores sparse user-pairs.
We denote our triangle counting algorithm with the variance reduction technique by \AlgWSTriVR{}.

As explained in Section~\ref{sub:triangle}, the factor of $n^3$ is caused by the RR for local edges.
However, most user-pairs $v_i$ and $v_j$ have a very small minimum degree
$\min\{d_i, d_j\} \ll d_{max}$,
and there is no edge $(v_i, v_j)$ between them in almost all cases.
In addition, even if there is an edge $(v_i, v_j)$, the number of triangles involving the sparse user-pair is very small
(at most $\min\{d_i, d_j\}$)
and can be approximated by $0$.
By ignoring such sparse user-pairs, we can dramatically reduce the variance of the RR for local edges at the cost of a small bias.
This is an intuition behind our variance reduction technique.

Algorithm~\ref{alg:wshuffle_triangle_vr} shows \AlgWSTriVR{}.
This algorithm detects sparse user-pairs based on the degree information.
However, user $v_i$'s degree $d_i$ can leak the information about edges of $v_i$.
Thus, \AlgWSTriVR{} calculates a differentially private estimate of $d_i$ within one round.

\setlength{\algomargin}{5mm}
\begin{algorithm}[t]
  \SetAlgoLined
  \KwData{Adjacency matrix $\bmA \in \{0,1\}^{n \times n}$, $\epsilon_1, \epsilon_2 \in \nnreals$, $\delta \in [0,1]$, $t \in [\lfloor \frac{n}{2} \rfloor]$, $c \in \nnreals$.
  }
  \KwResult{Estimate $\hf^\triangle(G)$ of $f^\triangle(G)$.}
  \tcc{Sample disjoint user-pairs}
  [d] $\sigma \leftarrow$\texttt{RandomPermutation}$(n)$\;
  [d] Send $(v_{\sigma(1)}, v_{\sigma(2)}), \ldots, (v_{\sigma(2t-1)}, v_{\sigma(2t)})$ to users\;
  \ForEach{$i \in \{1, 3, \ldots, 2t-1\}$}{
    $\hf_{\sigma(i), \sigma(i+1)}^\triangle(G) \leftarrow$ \AlgWSLE{}$(\bmA, \epsilon_2, \delta, (v_{\sigma(i)}, v_{\sigma(i+1)}))$\;
  }
  \tcc{Send noisy degrees}
  \For{$i=1$ \KwTo $n$}{
    [$v_i$] $\td_i \leftarrow d_i + \Lap(\frac{1}{\epsilon_1})$; Send $\td_i$ to the data collector\;
  }
  \tcc{Calculate a variance-reduced estimate}
  [d] $\td_{avg} \leftarrow \frac{1}{n} \sum_{i=1}^n \td_i$; $d_{th} \leftarrow c \td_{avg}$\;
  [d] $D \leftarrow \{i | i = 1, 3, \ldots, 2t-1,
  \min\{\td_{\sigma(i)}, \td_{\sigma(i+1)}\} > d_{th}\}$\;
  [d] $\hf^\triangle(G) \leftarrow \frac{n(n-1)}{6t} \sum_{i \in D} \hf_{\sigma(i),\sigma(i+1)}^\triangle(G)$\;
  [d] \KwRet{$\hf^\triangle(G)$}
  \caption{Our triangle counting algorithm with variance reduction \AlgWSTriVR{}.
  \AlgWSLE{} is shown in Algorithm~\ref{alg:WSLE}.
  }\label{alg:wshuffle_triangle_vr}
\end{algorithm}

Specifically, \AlgWSTriVR{} uses two privacy budgets: $\epsilon_1, \epsilon_2 \in \nnreals$.
The first budget $\epsilon_1$ is for privately estimating $d_i$, whereas the second budget $\epsilon_2$ is for
\AlgWSLE{}.
Lines 1 to 5 in Algorithm~\ref{alg:wshuffle_triangle_vr} are the same as those in Algorithm~\ref{alg:wshuffle_triangle}, except that Algorithm~\ref{alg:wshuffle_triangle_vr} uses $\epsilon_2$
to provide $(\epsilon_2, \delta)$-element DP.
After these processes, each user $v_i$ adds the Laplacian noise $\Lap(\frac{1}{\epsilon_1})$ with mean $0$ and scale $\frac{1}{\epsilon_1}$ to her degree $d_i$ and sends the noisy degree $\td_i$ ($= d_i + \Lap(\frac{1}{\epsilon_1})$) to the data collector (lines 6-8).
Because the sensitivity \cite{DP} of $d_i$ (the maximum distance of $d_i$ between two neighbor lists that differ in one bit) is $1$, adding $\Lap(\frac{1}{\epsilon_1})$ to $d_i$ provides $\epsilon_1$-element DP.

Then, the data collector estimates the average degree $d_{avg}$ as $\td_{avg} = \frac{1}{n}\sum_{i=1}^n \td_i$ and sets a threshold $d_{th}$ of the minimum degree to $d_{th} = c \td_{avg}$, where $c \in \nnreals$ is a small positive number, e.g.,
$c \in [1,10]$ (line 9).
Finally, the data collector estimates $f^\triangle(G)$ as
\begin{align}
    \textstyle{\hf^\triangle(G) = \frac{n(n-1)}{6t} \sum_{i \in D} \hf_{\sigma(i),\sigma(i+1)}^\triangle(G),}
   \label{eq:hf_triangle_II_ast}
\end{align}
where
\begin{align*}
D = \{i | i = 1, 3, \ldots, 2t-1,
  \min\{\td_{\sigma(i)}, \td_{\sigma(i+1)}\} > d_{th}\}
\end{align*}
(lines 10-11).
The difference between (\ref{eq:hf_triangle_II}) and (\ref{eq:hf_triangle_II_ast}) is that (\ref{eq:hf_triangle_II_ast}) ignores sparse user-pairs $v_{\sigma(i)}$ and $v_{\sigma(i+1)}$ such that $\min\{\td_{\sigma(i)}, \td_{\sigma(i+1)}\} \leq d_{th}$.
Since
$d_{avg} \ll d_{max}$ in practice,
$d_{th} \ll d_{max}$
holds for small $c$.

The parameter $c$ controls the trade-off between the bias and variance of the estimate $\hf^\triangle(G)$.
The larger $c$ is, the more user-pairs are ignored.
Thus, as $c$ increases, the bias is increased, and the variance is reduced.
In practice,
a small $c$ not less than $1$ 
results in a small MSE 
because most real graphs are scale-free networks that have a power-law degree distribution \cite{NetworkScience}.
In the scale-free networks, most users' degrees are smaller than the average degree $d_{avg}$.
For example, in the BA (Barab\'{a}si-Albert) graph model \cite{NetworkScience,Hagberg_SciPy08},
most users' degrees are $\frac{d_{avg}}{2}$.
Thus, if we set
$c \in [1,10]$, for example,
then most user-pairs are ignored (i.e., $|D| \ll t$), which leads to a significant reduction of the variance at the cost of a small bias.

Recall that the parameter $t$ in \AlgWSTriVR{} controls the trade-off between the MSE and the time complexity. 
Although \AlgWSTriVR{} always samples $t$ disjoint user-pairs, we can modify \AlgWSTriVR{} so that it stops sampling user-pairs right after the estimate $\hf^\triangle(G)$ in (\ref{eq:hf_triangle_II_ast}) is converged. 
We can also sample dense user-pairs $(v_i, v_j)$ with large noisy degrees $\td_i$ and $\td_j$ at the beginning (e.g., by sorting users in descending order of noisy degrees) to improve the MSE for small $t$. 
Evaluating such improved algorithms is left for future work. 

\smallskip
\noindent{\textbf{Theoretical Properties.}}~~As with \AlgWSTri{},
\AlgWSTriVR{} provides the following privacy guarantee:

\begin{theorem}
\label{thm:DP_II_ast}
\AlgWSTriVR{} provides $(\epsilon_1 + \epsilon_2, \delta)$-element DP and $(2(\epsilon_1 + \epsilon_2), 2\delta)$-edge DP.
\end{theorem}

Next, we analyze the bias of \AlgWSTriVR{}. Here, we assume most users have a small degree 
using parameters $\lambda \in \nnreals$ and $\alpha \in [0,1)$: 

\begin{theorem}
\label{thm:bias_II_ast}
  Suppose that in $G$, there exist $\lambda \in \nnreals$ and $\alpha \in [0,1)$
  such that at most $n^\alpha$ users have a degree larger than
  $\lambda d_{avg}$. Suppose \AlgWSTriVR{} is run with $c \geq \lambda$.
  Then, the estimator produced by
  \AlgWSTriVR{} provides the following bias guarantee:
\begin{align}
    Bias[\hf^\triangle(G)]
    = |\E[\hf^\triangle(G)] - f^\triangle(G)|
    \leq
    \frac{n c^2 d_{avg}^2}{3} +
    \frac{4 n^\alpha}{3 \epsilon_1^2}.
    \label{eq:thm:bias_II_ast}
\end{align}
\end{theorem}
The values of $\lambda$ and $\alpha$ depend on the original graph $G$.
In the scale-free networks, $\alpha$ is small for a moderate value of $\lambda$.
For example, in the BA graph with $n=107614$ and $d_{avg}=200$ used in \conference{the full version \cite{Imola_CCSFull22}}\arxiv{Appendix~\ref{sec:BA-graph}}, $\alpha = 0.5$, $0.6$, $0.8$, and $0.9$ when $\lambda = 10.1$, $5.4$, $1.6$, and $0.9$, respectively.
When $c$ and $\epsilon_1$ are constants,
the bias can be expressed as $O(n d_{avg}^2)$.

Finally, we show the variance of \AlgWSTriVR{}. 
This result assumes that $c$ is bigger $(= \frac{(1-\alpha) \log n}{\epsilon_1 d_{avg}})$ than $\lambda$. 
We assume this 
because otherwise, many sparse users (with $d_i \leq \lambda d_{avg}$) have a noisy degree $\td_i \geq c \td_{avg}$, causing the set $D$ to
be noisy. In practice, the gap between $c$ and $\lambda$ is small because 
$\log n$ is much smaller than $d_{avg}$.

\begin{theorem}
\label{thm:var_II_ast}
  Suppose that in $G$, there exist $\lambda \in \nnreals$ and $\alpha \in [0,1)$
  such that at most $n^\alpha$ users have a degree larger than
  $\lambda d_{avg}$. Suppose \AlgWSTriVR{} is
  run with $c \geq \lambda + \frac{(1-\alpha) \log n}{\epsilon_1 d_{avg}}$.
  Then, the estimator produced by
  \AlgWSTriVR{} provides the following variance guarantee:
\begin{align}
  &\V[\hf^\triangle(G)] \leq \nonumber\\
  &\frac{n^2 d_{max}^4}{9} + \frac{2n^{2+2\alpha}}{9t}err_{\AlgWSLE}(n,d_{max},q,q_L) + \frac{n^{2+\alpha}
  d_{max}^3 }{36t}.
  \label{eq:thm:l2_II_ast}
\end{align}
When $\epsilon_1$, $\epsilon_2$, and $\delta$ are constants,
$\epsilon_L = \log n + O(1)$, and $t = \lfloor\frac{n}{2}\rfloor$,
\begin{align}
  \V[\hf^\triangle(G)] \leq O(n^2 d_{max}^4 + n^{1+2\alpha} d_{max}^2).
  \label{eq:thm:l2_II_ast_simp}
\end{align}
\end{theorem}

The first\footnote{The first term in (\ref{eq:thm:l2_II_ast}) is actually $\frac{(\sum_{i=1}^n d_i^2)^2}{9}$ and is much smaller than $\frac{n^2 d_{max}^4}{9}$. 
We express it as $O(n^2 d_{max}^4)$ in (\ref{eq:thm:l2_II_ast_simp}) for simplicity. 
See \conference{the full version \cite{Imola_CCSFull22}}\arxiv{Appendix~\ref{sub:var_II_ast_proof}} for details.}, second, and third terms in (\ref{eq:thm:l2_II_ast}) are caused by the randomness in the choice of $D$, the RR, and user-pair sampling, respectively. 
By (\ref{eq:thm:l2_II_ast_simp}), our variance reduction technique reduces the variance from $O(n^3)$ to $O(n^\gamma)$ where $\gamma\in[2,3)$ when we ignore the factor of $d_{max}$.
Because the MSE is the sum of the squared bias and the variance, it is also $O(n^\gamma)$.

The value of $\gamma$ in our bound $O(n^\gamma)$ depends on the parameter $c$ in \AlgWSTriVR{}.
For example, in the BA graph ($n=107614$, $d_{avg} = 200$), $\gamma=2$, $2.2$, $2.6$, and $2.8$ ($\alpha=0.5$, $0.6$, $0.8$, and $0.9$) when $c=10.4$, $5.6$, $1.7$, and $1.0$, respectively, and $\epsilon_1=0.1$.
Thus, the variance decreases with increase in $c$.
However, by (\ref{eq:thm:bias_II_ast}), a larger $c$ results in a larger bias. 
In our experiments, we show that \AlgWSTriVR{} provides a small estimation error when $c=1$ to $4$.
When $c=1$, \AlgWSTriVR{} empirically works well despite a large $\gamma$ because most users' degrees are smaller than $d_{avg}$ in practice, as explained above.
This indicates that our upper bound in (\ref{eq:thm:l2_II_ast_simp}) might not be tight when $c$ is around $1$.
Improving the bound is left for future work.

\subsection{Summary}
\label{sub:summary}

Table~\ref{tab:upper_bounds_triangle} summarizes the
performance guarantees
of one-round triangle algorithms providing edge DP.
Here,
we consider a variant of \AlgWSTri{} that does not shuffle wedges (i.e., $\epsilon_L = \epsilon$) as a one-round local algorithm.
We call this variant \AlgWLTri{} (Wedge Local).
We also show the variance of \AlgARRTri{} \cite{Imola_USENIX22} and \AlgRRTri{} \cite{Imola_USENIX21}.
The time complexity of \AlgRRTri{} is $O(n^3)$\footnote{Technically speaking, the algorithms of \AlgRRTri{} and
the one-round local algorithms in \cite{Ye_ICDE20,Ye_TKDE21} involve counting the number of triangles in a dense
graph. This can be done in time $O(n^\omega)$, where $\omega \in [2,3)$ and $O(n^\omega)$ is the time required for matrix multiplication. However, these algorithms are of theoretical interest,
and they do not outperform naive matrix multiplication except
for very large matrices~\cite{Alman_2021}. Thus, we assume implementations that use naive matrix multiplication in
$O(n^3)$ time.}, and that
of \AlgARRTri{} is $O(n^2)$ when we set the sampling probability $p_0 \in [0,1]$ of the ARR to $p_0=O(n^{-1/3})$.
We prove the variance of \AlgARRTri{} in this case and \AlgRRTri{}
in \conference{the full version \cite{Imola_CCSFull22}}\arxiv{Appendix~\ref{sec:upper}}.
We do not show the other one-round local algorithms \cite{Ye_ICDE20,Ye_TKDE21} in Table~\ref{tab:upper_bounds_triangle} for two reasons: (i) they have the time complexity of $O(n^3)$ and suffer from a larger estimation error than \AlgRRTri{} \cite{Imola_USENIX22};
(ii) their upper-bounds on the variance and bias are unclear.

\begin{table}[t]
  \caption{Performance guarantees
  of one-round triangle counting algorithms providing edge DP.
  $\alpha \in [0,1)$. 
  See also footnote 2 for the variance of \AlgWSTriVR{}.
  }
  \vspace{-4mm}
  \centering
  \begin{tabular}{|l|c|c|c|c|}
    \hline
    Algorithm & \hspace{-0.8mm}Model\hspace{-0.8mm} & Variance & Bias & Time \\ \hline
    \AlgWSTriVR{} & \hspace{-0.8mm}shuffle\hspace{-0.8mm} & \hspace{-1mm}$O(n^2
    d_{max}^4 \hspace{-0.5mm}+\hspace{-0.5mm} n^{1+2\alpha} d_{max}^2)$\hspace{-1mm} & \hspace{-1.5mm}$O(n d_{avg}^2)$\hspace{-1.5mm} & $O(n^2)$\\
    \hline
    \AlgWSTri{} & \hspace{-0.8mm}shuffle\hspace{-0.8mm} & $O(n^3 d_{max}^2)$ & $0$ & $O(n^2)$ \\ \hline
    \AlgWLTri{} & local & $O(n^4 + n^3 d_{max}^2)$ & $0$ & $O(n^2)$ \\ \hline
    \AlgARRTri{} \cite{Imola_USENIX22} & local & $O(n^6)$ & $0$ & $O(n^2)$ \\ \hline
    \AlgRRTri{} \cite{Imola_USENIX21} & local & $O(n^4)$ & $0$ & $O(n^3)$ \\ \hline
  \end{tabular}
  \label{tab:upper_bounds_triangle}
\end{table}

Table~\ref{tab:upper_bounds_triangle} shows that our \AlgWSTriVR{} dramatically outperforms the three local algorithms -- when we ignore $d_{max}$, the MSE of \AlgWSTriVR{} is 
$O(n^\gamma)$ where $\gamma\in[2,3)$, 
whereas that of the local algorithms is $O(n^4)$ or $O(n^6)$.
We also show this through experiments. 

Note that both \AlgARRTri{} and \AlgRRTri{} provide pure DP ($\delta=0$), whereas our shuffle algorithms provide approximate DP ($\delta > 0$). 
However, it would not make a noticeable difference, as $\delta$ is sufficiently small (e.g., $\delta = 10^{-8} \ll \frac{1}{n}$ in our experiments). 

\smallskip
\noindent{\textbf{Comparison with the Central Model.}}~~Finally, we note that our \AlgWSTriVR{} is worse than algorithms in the central model in terms of the estimation error. 

Specifically, Imola \textit{et al.} \cite{Imola_USENIX21} consider a central algorithm that 
adds the Laplacian noise $\Lap(\frac{d_{max}}{\varepsilon})$ to the true count $f^\triangle(G)$ and 
outputs $f^\triangle(G) + \Lap(\frac{d_{max}}{\epsilon})$\footnote{Here, we assume that $d_{max}$ is publicly available; e.g., $d_{max}=5000$ in Facebook \cite{Facebook_Limit}. 
When $d_{max}$ is not public, the algorithm in \cite{Imola_USENIX21} outputs $f(G) + \Lap(\frac{\td_{max}}{\epsilon})$, where $\td_{max} = \max_{i=1,\ldots,n} \td_i$, i.e., the maximum of noisy degrees.}. 
This central algorithm provides 
$(\epsilon, 0)$-edge DP. 
In addition, the estimate is unbiased, and the variance is $\frac{2d_{max}^2}{\epsilon^2} = O(d_{max}^2)$. 
Thus, the central algorithm provides a much smaller MSE ($=$ variance) than \AlgWSTriVR{}. 

However, our \AlgWSTriVR{} is preferable to central algorithms in terms of the trust model -- the central model assumes that a single party accesses personal data of all users and therefore has a risk that the entire graph is leaked from the party. 
\AlgWSTriVR{} can also be applied to decentralized social networks, as described in Section~\ref{sec:intro}.

\section{4-Cycle Counting Based on Wedge Shuffling}
\label{sec:4cycle}

Next, we propose a one-round 4-cycle counting algorithm in the shuffle model.
Section~\ref{sub:4cycle_overview} explains its overview.
Section~\ref{sub:4cycle_details} proposes our 4-cycle counting algorithm and shows its theoretical properties. 
Section~\ref{sub:summary_4cycle} summarizes the performance guarantees of our 4-cycle algorithms. 

\subsection{Overview}
\label{sub:4cycle_overview}
We apply our wedge shuffling technique to 4-cycle counting with two additional techniques: \textit{(i) bias correction} and \textit{(ii) sampling disjoint user-pairs}. 
Below, we briefly explain each of them. 

\smallskip
\noindent{\textbf{Bias Correction.}}~~As with triangles, we begin with the problem of counting 4-cycles involving specific users $v_i$ and $v_j$. 
We can leverage the noisy wedges output by our wedge shuffle algorithm \AlgWS{} to estimate such a 4-cycle count. 
Specifically, let $f^\square_{i,j}: \calG \rightarrow \nnints$ be a function that, given $G \in \calG$, outputs the number $f^\square_{i,j}(G)$ of $4$-cycles for which users $v_i$ and $v_j$ are opposite nodes, i.e. the number of unordered pairs $(k,k')$ such that $v_i-v_k-v_j-v_{k'}-v_i$ is a path in $G$.
Each pair $(k,k')$ satisfies the above requirement if and only if $v_i-v_k-v_j$ and $v_i-v_{k'}-v_j$ are wedges in $G$.
Thus, we have $f^\square_{i,j}(G) = \binom{f^\wedge_{i,j}}{2}$, where $f^\wedge_{i,j}$ is the number of wedges between $v_i$ and $v_j$. 
Based on this, we calculate an unbiased estimate $\hf^\wedge_{i,j}$ of the wedge count using \AlgWS{}. 
Then, we calculate an estimate of the 4-cycle count as $\binom{\hf^\wedge_{i,j}}{2}$. 
Here, it should be noted that the estimate $\binom{\hf^\wedge_{i,j}}{2}$ 
is \textit{biased}, as proved later. 
Therefore, 
we perform bias correction -- we subtract a positive value from the estimate to obtain an unbiased estimate $\hf_{i,j}^\square(G)$ of the 4-cycle count. 

Note that unlike \AlgWSLE{}, no edge between $(v_i,v_j)$ needs to be sent. 
In addition, thanks to the privacy amplification by shuffling, all wedges can be sent with 
small noise.

\smallskip
\noindent{\textbf{Sampling Disjoint User-Pairs.}}~~Having an estimate $\hf_{i,j}^\square(G)$, we turn our attention to estimating 4-cycle count $f^\square(G)$ in the entire graph $G$.
As with triangles, a naive solution using estimates $\hf_{i,j}^\square(G)$ for all $\binom{n}{2}$ user-pairs $(v_i,v_j)$ results in very large $\epsilon$ and $\delta$. 
To avoid this, we sample disjoint user-pairs and obtain an unbiased estimate of $f^\square(G)$ from them. 

\subsection{4-Cycle Counting}
\label{sub:4cycle_details}

\setlength{\algomargin}{5mm}
\begin{algorithm}[t]
  \SetAlgoLined
  \KwData{Adjacency matrix $\bmA \in \{0,1\}^{n \times n}$, $\epsilon \in \nnreals$, $\delta \in [0,1]$, $t \in [\lfloor \frac{n}{2} \rfloor]$.
  }
  \KwResult{Estimate $\hf^\square(G)$ of $f^\square(G)$.}
  $\epsilon_L \leftarrow \texttt{LocalPrivacyBudget}(n,\epsilon,\delta)$\;
  [d] $q_L \leftarrow \frac{1}{e^{\epsilon_L}+1}$\;
  \tcc{Sample disjoint user-pairs}
  [d] $\sigma \leftarrow$\texttt{RandomPermutation}$(n)$\;
  [d] Send $(v_{\sigma(1)}, v_{\sigma(2)}), \ldots, (v_{\sigma(2t-1)}, v_{\sigma(2t)})$ to users\;
  \ForEach{$i \in \{1, 3, \ldots, 2t-1\}$}{
    $\{y_{\pi_i(k)} | k \hspace{-1mm}\in I_{-(\sigma(i),\sigma(i+1))}\} \hspace{-1mm}\leftarrow\hspace{-1mm}$ \AlgWS{}$(\bmA, \epsilon_L, (v_{\sigma(i)}, v_{\sigma(i+1)}))$\;
    [d] $\hf_{\sigma(i),\sigma(i+1)}^\wedge(G) \leftarrow \sum_{k \in I_{-(\sigma(i),\sigma(i+1))}} \frac{y_{k} - q_L}{1-2q_L}$\;
    [d] $\hf_{\sigma(i),\sigma(i+1)}^\square(G) \leftarrow  \frac{\hf_{\sigma(i),\sigma(i+1)}^\wedge(G)(\hf_{\sigma(i),\sigma(i+1)}^\wedge(G)-1)}{2} - \frac{n-2}{2}\frac{q_L(1-q_L)}{(1-2q_L)^2}$\;
  }
  \tcc{Calculate an unbiased estimate}
  [d] $\hf^\square(G) \leftarrow \frac{n(n-1)}{4t}\sum_{i=1,3,\ldots,2t-1} \hf_{\sigma(i),\sigma(i+1)}^\square(G)$\;
  [d] \KwRet{$\hf^\square(G)$}
  \caption{Our 4-cycle counting algorithm \AlgWSCyc{}.
  \AlgWS{} is shown in Algorithm~\ref{alg:WShuffle}.
  }\label{alg:wshuffle_cycle}
\end{algorithm}

\noindent{\textbf{Algorithm.}}~~Algorithm \ref{alg:wshuffle_cycle} shows our 4-cycle counting algorithm. 
We denote it by \AlgWSCyc{}. 
First, we set a local privacy budget $\epsilon_L$ from $n$, $\epsilon$, and $\delta$ in the same way as \AlgWSLE{} (line 1). 
Then, we sample $t$ disjoint pairs of users using the permutation $\sigma$ (lines 3-4). 
Each pair is given by $(v_{\sigma(i), \sigma(i+1)})$ for $i \in \{1, 3, \ldots, 2t-1\}$.

For each $i \in \{1, 3, \ldots, 2t-1\}$, we compute an unbiased estimate $\hf_{\sigma(i),
\sigma(i+1)}^\square(G)$ of the 4-cycle count involving $v_{\sigma(i)}$ and $v_{\sigma(i+1)}$ (lines 5-9). 
To do this, 
we call \AlgWS{} on
$(v_{\sigma(i)} v_{\sigma(i+1)})$ to obtain an unbiased estimate $\hf_{\sigma(i),
\sigma(i+1)}^\wedge(G)$ of the wedge count (lines 6-7). 
We calculate an estimate $\hf_{i,j}^\wedge(G)$ of the number $f_{i,j}^\wedge(G)$ of wedges between $v_i$ and $v_j$ in $G$ as follows:
\begin{equation}\label{eq:wedge-estimate}
  \textstyle{\hf^\wedge_{i,j}(G) = \sum_{k \in I_{-(i,j)}} \frac{y_k - q_L}{1-2q_L}.}
\end{equation}
Later, we will prove that $\hf^\wedge_{i,j}(G)$ is an unbiased estimator. 
As with (\ref{eq:hfij_triangle}), this estimate involves the sum over the set $\{y_{\pi(k)}\}$ and does not require knowing the permutation $\pi$ produced by the shuffler. 
Then, we obtain an unbiased estimator of $f_{\sigma(i), \sigma(i+1)}^\square$ as follows:
\begin{equation}\label{eq:4-cycle-ij-estimate}
  \textstyle{\hf_{i,j}^\square(G) = \frac{\hf^\wedge_{i, j}(G)(\hf^\wedge_{i, j}(G)-1)}{2} - \frac{n-2}{2}\frac{q_L(1-q_L)}{(1-2q_L)^2}}
\end{equation}
(line 8). 
Note that 
there is a quadratic relationship between $f_{i, j}^\square(G)$ and $f_{i, j}^\wedge(G)$, i.e., $f_{i, j}^\square(G) = \binom{f_{i, j}^\wedge(G)}{2}$. 
Thus, 
even though $\hf_{i, j}^\wedge(G)$ is unbiased, 
we must subtract a term from $\binom{\hf_{i, j}^\wedge(G)}{2}$ (i.e., bias correction) to obtain an unbiased estimator $\hf_{i, j}^\square(G)$. 
This forms the righthand side of
~\eqref{eq:4-cycle-ij-estimate} and ensures that $\hf_{i,j}^\square(G)$ is unbiased. 

Finally, 
we sum and scale 
$\hf_{\sigma(i), \sigma(i+1)}^\square(G)$ for each $i$ to obtain an estimate $\hf^\square(G)$ of the 4-cycle count $f^\square(G)$ in the entire graph $G$: 

\begin{equation}\label{eq:4-cycle-estimate}
  \textstyle{\hf^\square(G) = \frac{n(n-1)}{4t}\sum_{i=1,3,\ldots,2t-1} \hf_{\sigma(i),\sigma(i+1)}^\square(G)}
\end{equation}
(line 10). 
Note that it is possible that a single 4-cycle is counted twice; e.g., 
a 4-cycle $v_i$-$v_j$-$v_k$-$v_l$-$v_i$ is possibly counted by $(v_i,v_k)$ and $(v_j,v_l)$ if these user-pairs are selected. 
However, this is not an issue, because all 4-cycles are equally likely to be counted zero times, once, or twice. 
We also prove later that $\hf^\square(G)$ in (\ref{eq:4-cycle-estimate}) is an unbiased estimate of $f^\square(G)$. 

\smallskip
\noindent{\textbf{Theoretical Properties.}}~~First, \AlgWSCyc{} guarantees DP:
\begin{theorem}\label{thm:DP_IV}
\AlgWSCyc{} provides $(\epsilon, \delta)$-element DP and $(2\epsilon, 2\delta)$-edge DP.
\end{theorem}

In addition, thanks to the design of~\eqref{eq:4-cycle-ij-estimate}, we can show that \AlgWSCyc{} produces an unbiased estimate of $f^\square(G)$:
\begin{theorem}\label{thm:unbiased_IV}
  The estimate produced by \AlgWSCyc{} satisfies $\allowbreak \E[\hf^\square(G)] =
  f^\square(G)$.
\end{theorem}

Finally, we show the MSE ($=$ variance) of $f^\square(G)$:

\begin{theorem}\label{thm:l2-loss_IV}
  The estimate produced by \AlgWSCyc{} satisfies 
  \begin{align}
    & \MSE(\hf^\square(G)) = \V[\hf^\square(G)] \nonumber \\
    &\leq \frac{9n^5 q_L(d_{max} +
    nq_L)^2}{16t(1-2q_L)^4} + \frac{n^3d_{max}^6}{64t}. \label{eq:4-cycle-var}
  \end{align}
  When $\epsilon$ and $\delta$ are constants, $\epsilon_L = \log n + O(1)$, and $t
  = \lfloor \frac{n}{2} \rfloor$, we have
  \begin{equation}\label{eq:4-cycle-var-simp}
    \MSE(\hf^\square(G)) = \V[\hf^\square(G)] 
    = O\left( n^3 d_{max}^2 + n^2d_{max}^6 \right).
  \end{equation}
\end{theorem}
The first and second terms in (\ref{eq:4-cycle-var}) are caused by the RR and the sampling of disjoint user-pairs, respectively. 

\begin{table}[t]
  \caption{Performance guarantees of one-round 4-cycle counting algorithms providing edge DP.
  }
  \vspace{-4mm}
  \centering
  \begin{tabular}{|l|c|c|c|c|}
    \hline
    Algorithm & Model & Variance & Bias & Time \\ \hline
    \AlgWSCyc{} & shuffle & $O(n^3 d_{max}^2 + n^2 d_{max}^6)$ & $0$ & $O(n^2)$ \\ \hline
    \AlgWLCyc{} & local & $O(n^6 + n^2 d_{max}^6)$ & $0$ & $O(n^2)$ \\ \hline
  \end{tabular}
  \label{tab:upper_bounds_4cycle}
\end{table}

\subsection{Summary}
\label{sub:summary_4cycle}
Table~\ref{tab:upper_bounds_4cycle} summarizes the performance guarantees of the 4-cycle counting algorithms. 
As a one-round local algorithm, we consider a local model version of \AlgWSCyc{} that does not shuffle wedges (i.e., $\epsilon_L = \epsilon$). 
We denote it by \AlgWLCyc{}. 
To our knowledge, \AlgWLCyc{} is the first local 4-cycle counting algorithm. 

By (\ref{eq:4-cycle-var}), when $t = \lfloor \frac{n}{2} \rfloor$, the MSE of \AlgWLCyc{} can be expressed as $O(n^6 + n^2 d_{max}^6)$. 
Thus, our wedge shuffle technique dramatically reduces the MSE from $O(n^6 + n^2 d_{max}^6)$ to $O(n^3 d_{max}^2 + n^2 d_{max}^6)$. 
Note that the square of the true count $f^\square(G)$ is $O(n^2 d_{max}^6)$. 
This indicates that our \AlgWSCyc{} may not work well in an extremely sparse graph where $d_{max} < n^{\frac{1}{4}}$. 
However, $d_{max} \gg n^{\frac{1}{4}}$ holds in most social graphs; e.g., the maximum number $d_{max}$ of friends is much larger than $100$ when $n=10^8$. 
In this case, \AlgWSCyc{} can accurately estimate the 4-cycle count, as shown in our experiments. 

\smallskip
\noindent{\textbf{Comparison with the Central Model.}}~~As with triangles, our \AlgWSCyc{} is worse than algorithms in the central model in terms of the estimation error. 

Specifically, 
analogously to the central algorithm for triangles~\cite{Imola_USENIX21}, 
we can consider a central algorithm that outputs $f^\square(G) + \Lap(\frac{d_{max}^2}{\epsilon})$. 
This algorithm provides 
$(\epsilon, 0)$-edge DP and the variance of $\frac{2d_{max}^4}{\epsilon^2} = O(d_{max}^4)$. 
Because $d_{max}$ is much smaller than $n$, this central algorithm provides a much smaller MSE ($=$ variance) than \AlgWSCyc{}. 
This indicates 
that there is a trade-off between the trust model and the estimation error.

\section{Experimental Evaluation}
\label{sec:experiments}
Based on the performance guarantees summarized in Tables~\ref{tab:upper_bounds_triangle} and \ref{tab:upper_bounds_4cycle}, we pose the following research questions: 

\begin{description}[leftmargin=8.75mm]
    \item[RQ1.] 
    How much do our entire algorithms (\AlgWSTriVR{} and \AlgWSCyc{}) outperform the local algorithms?
    \item[RQ2.] 
    For triangles, how much does our variance reduction technique decrease the relative error?
    \item[RQ3.] 
    How small relative errors do our entire algorithms 
    achieve with a small privacy budget?
\end{description}
We designed experiments to answer these questions. 

\subsection{Experimental Set-up}
\label{sub:set-up}
We used the following two real graph datasets: 
\begin{itemize}
    \item \textbf{Gplus}: The first dataset is the Google+ dataset \cite{McAuley_NIPS12} denoted by \Gplus{}. 
    This dataset includes a social graph $G=(V,E)$ with $n=107614$ users and $12238285$ edges, where an edge $(v_i, v_j) \in E$ represents that a user $v_i$ follows or is followed by $v_j$. 
    The average and maximum degrees are $d_{avg} = 227.4$ and $d_{max} = 20127$, respectively. 
    \item \textbf{IMDB}: The second dataset is the IMDB (Internet Movie Database)~\cite{IMDB_GD05} denoted by \IMDB{}. 
    This dataset includes a bipartite graph between $896308$ actors and $428440$ movies. 
    From this, we extracted a graph $G=(V,E)$ with $n=896308$ actors and $57064358$ edges, where an edge represents that two actors have played in the same movie. 
    The average and maximum degrees are $d_{avg} = 127.3$ and $d_{max} = 15451$, respectively; i.e., \IMDB{} is more sparse than \Gplus{}. 
\end{itemize}
In \conference{the full version \cite{Imola_CCSFull22}}\arxiv{Appendix~\ref{sec:BA-graph}}, we also evaluate our algorithms using the Barab\'{a}si-Albert graphs \cite{NetworkScience,Hagberg_SciPy08}, which have a power-law degree distribution. 
Moreover, in \conference{\cite{Imola_CCSFull22}}\arxiv{Appendix~\ref{sec:bipartite}}, we evaluate our 4-cycle algorithms using bipartite graphs generated from \Gplus{} and \IMDB{}. 

For triangle counting, we evaluated the following four one-round algorithms: \AlgWSTriVR{}, \AlgWSTri{}, \AlgWLTri{}, and \AlgARRTri{} \cite{Imola_USENIX22}. 
We did not evaluate \AlgRRTri{} \cite{Imola_USENIX21}, because it was too inefficient -- it was reported in \cite{Imola_USENIX21} that when $n=10^6$, \AlgRRTri{} would require over $30$ years even on a supercomputer. 
The same applies to the one-round local algorithms in \cite{Ye_ICDE20,Ye_TKDE21} with the same time complexity ($=O(n^3)$). 

For 4-cycle counting, we compared \AlgWSCyc{} with \AlgWLCyc{}. 
Because \AlgWLCyc{} is the first local 4-cycle counting algorithm (to our knowledge), we did not evaluate other algorithms. 

In our shuffle algorithms \AlgWSTriVR{}, \AlgWSTri{}, and \AlgWSCyc{}, we set $\delta = 10^{-8}$ ($\ll \frac{1}{n}$) and $t=\frac{n}{2}$. 
We used the numerical upper bound in \cite{Feldman_FOCS21} for calculating $\epsilon$ in the shuffle model. 
In \AlgWSTriVR{}, we set 
$c\in[0.1,4]$ 
and 
divided the total privacy budget $\epsilon$ as $\epsilon_1 = \frac{\epsilon}{10}$ and $\epsilon_2 = \frac{9\epsilon}{10}$. 
Here, we assigned a small budget to $\epsilon_1$ because a degree $d_i$ has a very small sensitivity ($=1$) and $\Lap(\frac{1}{\epsilon_1})$ is very small. 
In \AlgARRTri{}, we set the sampling probability $p_0$ to $p_0 = n^{-1/3}$ or $0.1n^{-1/3}$ so that the time complexity is $O(n^2)$. 

We ran each algorithm $20$ times and evaluated the average relative error over the $20$ runs. 
In \conference{the full version \cite{Imola_CCSFull22}}\arxiv{Appendix~\ref{sec:standard_error}}, we show that the standard error of the average relative error is small. 

\subsection{Experimental Results}
\label{sub:results}

\noindent{\textbf{Relative Error vs. $\epsilon$.}}~~We first evaluated the relation between the relative error and 
$\epsilon$ in element DP or edge LDP, i.e., $2\epsilon$ in edge DP. 
We also measured the time to estimate the triangle/4-cycle count from the adjacency matrix $\bmA$ using a supercomputer \cite{ABCI} with two Intel Xeon Gold 6148 processors (2.40 GHz, 20 Cores) and 412 GB main memory. 

\begin{figure}[t]
  \centering
  \includegraphics[width=0.99\linewidth]{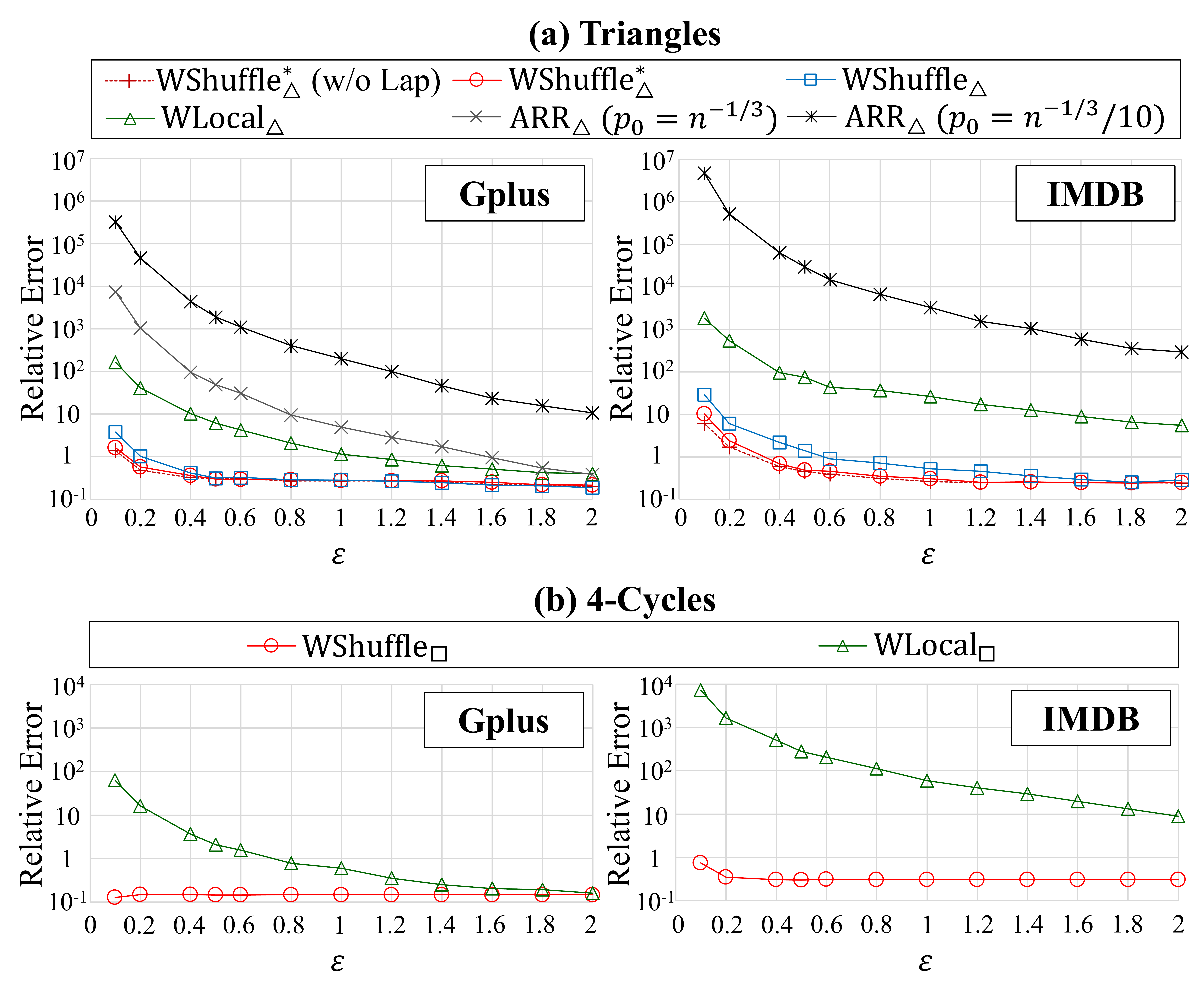}
  \vspace{-4mm}
  \caption{Relative error vs. $\epsilon$ 
  ($n=107614$ in \Gplus{}, $n=896308$ in \IMDB{}, $c=1$). 
  $p_0$ is the sampling probability in the ARR. 
  }
  \label{fig:res1_eps}
\end{figure}

\begin{table}[t]
  \caption{Relative error (RE) when $\epsilon=0.5$ or $1$ and computational time ($n=107614$ in \Gplus{}, $n=896308$ in \IMDB{}, $c=1$). 
  The lowest relative error is highlighted in bold.
  }
  \vspace{-4mm}
  \centering
  (a) \Gplus{}\\
  \begin{tabular}{|c|c|c|c|}
    \hline
    & RE ($\epsilon=0.5$) & RE ($\epsilon=1$) & Time (sec)\\ \hline
    \AlgWSTriVR{} & $\bm{2.98 \times 10^{-1}}$ & $\bm{2.77 \times 10^{-1}}$ & $3.60 \times 10^1$ \\ \hline
    \AlgWSTri{} & $3.12 \times 10^{-1}$ & $2.79 \times 10^{-1}$ & $3.62 \times 10^1$ \\ \hline
    \AlgWLTri{} & $6.10 \times 10^0$ & $1.14 \times 10^0$ & $5.83 \times 10^1$ \\ \hline
    \AlgARRTri{} ($p_0=n^{-1/3}$) & $4.90 \times 10^1$ & $4.93 \times 10^0$ & $7.15 \times 10^2$ \\ \hline
    \hspace{-0.5mm}\AlgARRTri{} ($p_0=0.1n^{-1/3}$)\hspace{-0.5mm} & $1.88 \times 10^3$ & $1.97 \times 10^2$ & $3.48 \times 10^1$ \\ \hline \hline
    \AlgWSCyc{} & $\bm{1.45 \times 10^{-1}}$ & $\bm{1.47 \times 10^{-1}}$ & $3.47 \times 10^1$ \\ \hline
    \AlgWLCyc{} & $2.08 \times 10^0$ & $5.96 \times 10^{-1}$ & $5.70 \times 10^1$ \\ \hline
  \end{tabular}\\
  (b) \IMDB{}\\
  \begin{tabular}{|c|c|c|c|}
    \hline
    & RE ($\epsilon=0.5$) & RE ($\epsilon=1$) & Time (sec)\\ \hline
    \AlgWSTriVR{} & $\bm{4.88 \times 10^{-1}}$ & $\bm{3.08 \times 10^{-1}}$ & $2.39 \times 10^3$ \\ \hline
    \AlgWSTri{} & $1.41 \times 10^0$ & $5.22 \times 10^{-1}$ & $2.40 \times 10^3$ \\ \hline
    \AlgWLTri{} & $7.46 \times 10^1$ & $2.63 \times 10^1$ & $3.96 \times 10^3$ \\ \hline
    \hspace{-0.5mm}\AlgARRTri{} ($p_0=0.1n^{-1/3}$)\hspace{-0.5mm} & $2.98 \times 10^4$ & $3.27 \times 10^3$ & $2.81 \times 10^3$ \\ \hline \hline
    \AlgWSCyc{} & $\bm{3.03 \times 10^{-1}}$ & $\bm{3.08 \times 10^{-1}}$ & $2.29 \times 10^3$ \\ \hline
    \AlgWLCyc{} & $2.82 \times 10^2$ & $5.91 \times 10^1$ & $3.91 \times 10^3$ \\ \hline
  \end{tabular}
  \label{tab:res1_eps_tri_time}
\end{table}

Figure~\ref{fig:res1_eps} shows the relative error ($c=1$). 
Here, we show the performance of \AlgWSTri{} when we do not add the Laplacian noise (denoted by \AlgWSTri{} (w/o Lap)). 
In \IMDB{}, we do not show \AlgARRTri{} with $p_0 = n^{-1/3}$, because it takes too much time (longer than one day). 
Table~\ref{tab:res1_eps_tri_time} highlights the relative error when $\epsilon=0.5$ or $1$. 
It also shows the running time of counting triangles or 4-cycles when $\epsilon=1$ (we verified that the running time had little dependence on $\epsilon$). 

Figure~\ref{fig:res1_eps} and Table~\ref{tab:res1_eps_tri_time} show that our shuffle algorithms dramatically improve the local algorithms. 
In triangle counting, 
\AlgWSTriVR{} outperforms \AlgWLTri{} by one or two orders of magnitude and \AlgARRTri{} by even more\footnote{Note that \AlgARRTri{} uses only the lower-triangular part of the adjacency matrix $\bmA$ and therefore provides 
$\epsilon$-edge DP (rather than $2\epsilon$-edge DP); i.e., it does not suffer from the doubling issue explained in Section~\ref{sub:privacy}. However, Figure~\ref{fig:res1_eps} shows that \AlgWSTriVR{} significantly outperforms \AlgARRTri{} 
even if we double $\epsilon$ for only \AlgWSTriVR{}.}. 
\AlgWSTriVR{} also requires less running time than \AlgARRTri{} with $p_0 = n^{-1/3}$. 
Although the running time of \AlgARRTri{} can be improved by using a smaller $p_0$, it results in a higher relative error. 
In 4-cycle counting, 
\AlgWSCyc{} significantly outperforms \AlgWLCyc{}. 
The difference between our shuffle algorithms and the local algorithms is larger in \IMDB{} because it is more sparse; i.e., the difference between $d_{max}$ and $n$ is larger in \IMDB{}. 
This is consistent with our theoretical results in Tables~\ref{tab:upper_bounds_triangle} and \ref{tab:upper_bounds_4cycle}. 

Figure~\ref{fig:res1_eps} and Table~\ref{tab:res1_eps_tri_time} also show that \AlgWSTriVR{} outperforms \AlgWSTri{}, especially when $\epsilon$ is small. 
This is because the variance is large when $\epsilon$ is small. 
In addition, \AlgWSTriVR{} significantly outperforms \AlgWSTri{} in \IMDB{} because \AlgWSTriVR{} significantly reduces the variance when $d_{max} \ll n$, as shown in Table~\ref{tab:upper_bounds_triangle}. 
In other words, this is also consistent with our theoretical results. 
For example, when $\epsilon=0.5$, our variance reduction technique reduces the relative error from $1.41$ to $0.488$ (about one-third) in \IMDB{}. 

Furthermore, Figure~\ref{fig:res1_eps} shows that the relative error of \AlgWSTriVR{} is hardly changed by adding the Laplacian noise. 
This is because the sensitivity of each user's degree $d_i$ is very small ($=1$). 
In this case, the Laplacian noise is also very small. 

\begin{figure}[t]
  \centering
  \includegraphics[width=0.99\linewidth]{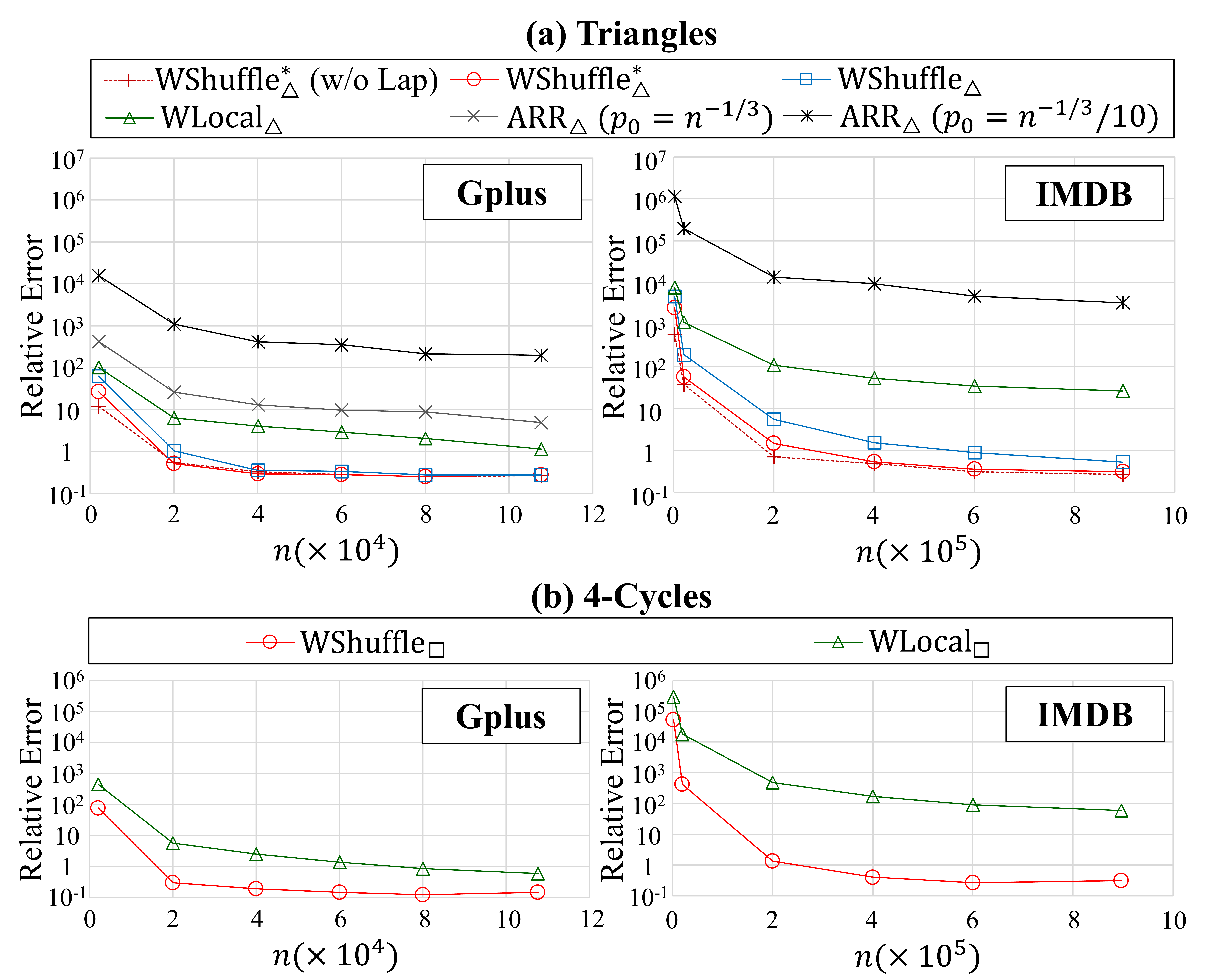}
  \vspace{-4mm}
  \caption{Relative error vs. $n$ ($\epsilon=1$, $c=1$).
  }
  \label{fig:res2_n}
\end{figure}

Our \AlgWSTriVR{} achieves a relative error of $0.3$ ($\ll 1$) 
when the privacy budget is $\epsilon = 0.5$ or $1$ in element DP ($2\epsilon = 1$ or $2$ in edge DP). 
\AlgWSCyc{} achieve a relative error of $0.15$ to $0.3$ with a smaller privacy budget (e.g., $\epsilon = 0.2$) because it  does not send local edges -- the error of \AlgWSCyc{} is mainly caused by user-pair sampling that is independent of $\epsilon$. 

In summary, our \AlgWSTriVR{} and \AlgWSCyc{} significantly outperform the local algorithms and achieve a relative error much smaller than $1$ with a reasonable privacy budget, i.e., $\epsilon \leq 1$. 

\smallskip
\noindent{\textbf{Relative Error vs. $n$.}}~~Next, we evaluated the relation between the relative error and $n$. 
Specifically, we randomly selected $n$ users from all users and extracted a graph with $n$ users. 
Then we set $\epsilon = 1$ and changed $n$ to various values starting from $2000$. 

Figure~\ref{fig:res2_n} shows the results ($c=1$). 
When $n=2000$, \AlgWSTri{} and \AlgWSCyc{} provide 
relative errors close to \AlgWLTri{} and \AlgWLCyc{}, respectively. 
This is because the privacy amplification effect is limited when $n$ is small. 
For example, when $n=2000$ and $\epsilon=1$, 
the numerical bound is $\epsilon_L=1.88$. 
The value of $\epsilon_L$ increases with increase in $n$; e.g., when $n=107614$ and $896308$, the numerical bound is $\epsilon_L= 5.86$ and $7.98$, respectively. 
This explains the reason that our shuffle algorithms significantly outperform the local algorithms when $n$ is large in Figure~\ref{fig:res2_n}.

\smallskip
\noindent{\textbf{Parameter $c$ in \AlgWSTriVR{}.}}~~Finally, we evaluated our \AlgWSTriVR{} while changing the parameter $c$ that controls the bias and variance. 
Recall that as $c$ increases, the bias is increased, and the variance is reduced. 
We set $\epsilon=0.1$ or $1$ and changed $c$ from $0.1$ to $4$. 

Figure~\ref{fig:res4_thr} shows the results. 
Here, we also show the relative error of \AlgWSTri{}. 
We observe that the optimal $c$ is different for $\epsilon=0.1$ and $\epsilon=1$. The optimal $c$ is around $3$ to $4$ for $\epsilon=0.1$, whereas the optimal $c$ is around $0.5$ to $1$ for $\epsilon=1$. 
This is because the variance of \AlgWSTri{} is large (resp.~small) when $\epsilon$ is small (resp.~large). 
For a small $\epsilon$, a large $c$ is effective in significantly reducing the variance. 
For a large $\epsilon$, a small $c$ is effective in keeping a small bias. 

We also observe that \AlgWSTriVR{} is always better than (or almost the same as) \AlgWSTri{} when $c=1$ or $2$. 
This is because most users' degrees are smaller than the average degree $d_{avg}$, as described in Section~\ref{sub:var_red}. 
When $c=1$ or $2$, most user-pairs are ignored. 
Therefore, we can significantly reduce the variance at the cost of a small bias. 

\begin{figure}[t]
  \centering
  \includegraphics[width=0.99\linewidth]{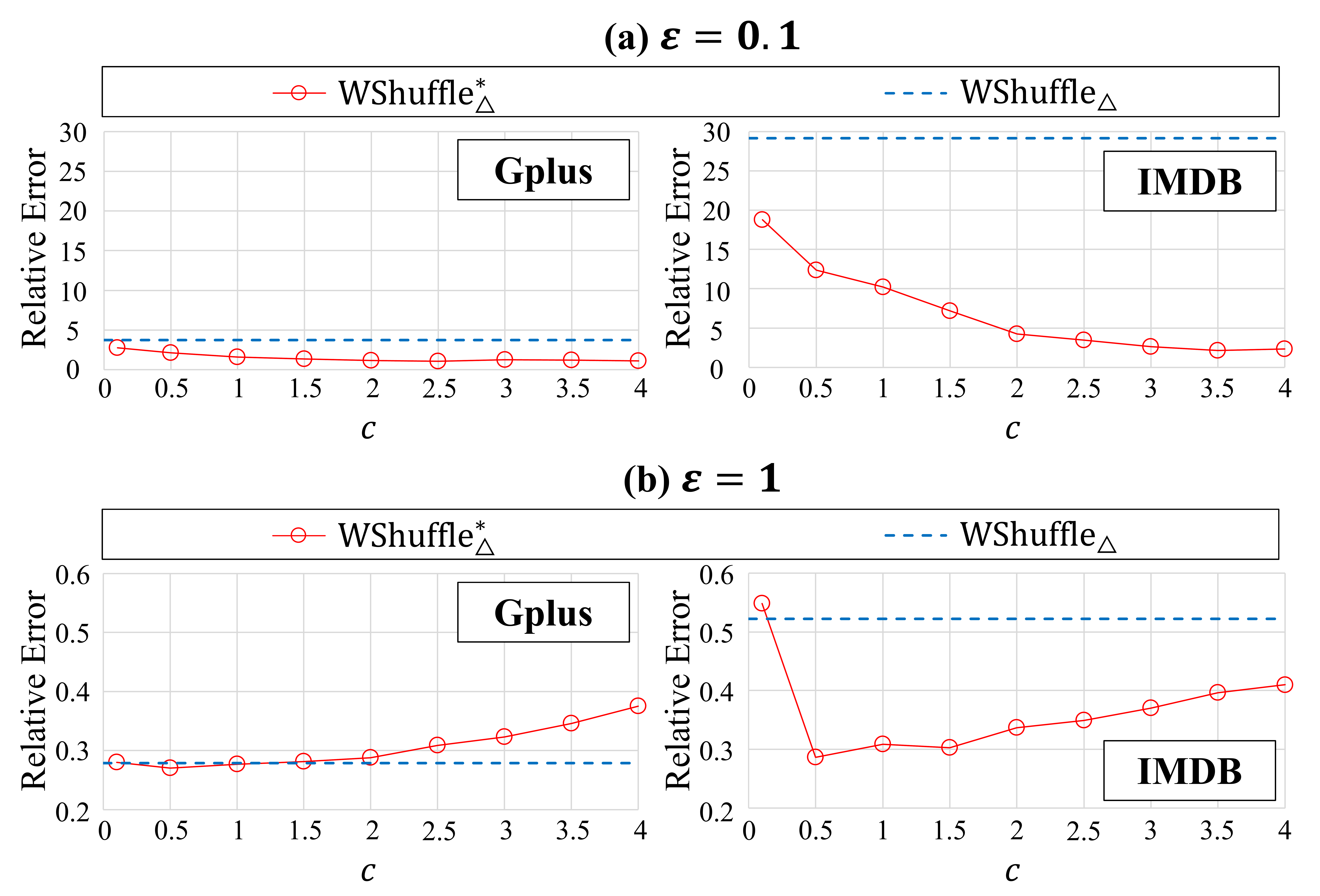}
  \vspace{-4mm}
  \caption{Relative error vs. parameter $c$ in \AlgWSTriVR{} ($n=107614$ in \Gplus{}, $n=896308$ in \IMDB{}).
  }
  \label{fig:res4_thr}
\end{figure}

\smallskip
\noindent{\textbf{Summary.}}~~In summary, our answers to the three questions at the beginning of Section~\ref{sec:experiments} are as follows. 
RQ1: Our \AlgWSTriVR{} and \AlgWSCyc{} outperform the one-round local algorithms by one or two orders of magnitude (or even more). 
RQ2: Our variance reduction technique significantly reduces the relative error (e.g., by about one-third) 
for a small $\epsilon$ in a sparse dataset. 
RQ3: 
\AlgWSTriVR{} achieves a relative error of $0.3$ ($\ll 1$) when $\epsilon=0.5$ or $1$ in element DP ($2\epsilon=1$ or $2$ in edge DP). 
\AlgWSCyc{} achieves a relative error of $0.15$ to $0.3$ with a smaller privacy budget: $\epsilon=0.2$. 

\section{Conclusion}
\label{sec:conclusion}
In this paper, we made the first attempt (to our knowledge) to 
shuffle graph data for privacy amplification. 
We proposed wedge shuffling as a basic technique and then applied it to 
one-round triangle and 4-cycle counting with several additional techniques. 
We showed upper bounds on 
the MSE 
for each algorithm. 
We also showed through comprehensive experiments that our one-round shuffle algorithms significantly outperform the one-round local algorithms and achieve a small relative error with a reasonable privacy budget, e.g., smaller than $1$ in edge DP. 

For future work, we would like to apply wedge shuffling to other subgraphs such as 3-hop paths \cite{Sun_CCS19} and $k$-triangles \cite{Karwa_PVLDB11}. 

\begin{acks}
Kamalika Chaudhuri and Jacob Imola would like to thank ONR under N00014-20-1-2334 and UC Lab Fees under LFR 18-548554  for research support.
Takao Murakami was supported in part by JSPS KAKENHI JP22H00521 and JP19H01109.
\end{acks}

\bibliographystyle{ACM-Reference-Format}
\bibliography{main}

\appendix
\allowdisplaybreaks[1]

\section{Experiments of the Clustering Coefficient}
\label{sec:cluster}
In Section~\ref{sec:experiments}, we showed that our triangle counting algorithm \AlgWSTriVR{} accurately estimates the triangle count within one round.
We also show that we can accurately estimate the clustering coefficient within one round by using \AlgWSTriVR{}.

We calculated the clustering coefficient as follows.
We used \AlgWSTriVR{} ($c=1$) for triangle counting and
the one-round local algorithm in \cite{Imola_USENIX21} with edge clipping \cite{Imola_USENIX22}
for 2-star counting.
The 2-star algorithm works as follows.
First, each user $v_i$ adds the Laplacian noise $\Lap(\frac{1}{\epsilon_1})$ and a non-negative constant $\eta \in \nnreals$ to her degree $d_i$ to obtain a noisy degree $\td_i = d_i + \Lap(\frac{1}{\epsilon_1}) + \eta$ with $\epsilon_1$-edge LDP.
If $\td_i < d_i$, then $v_i$ randomly removes $d_i - \lfloor \td_i \rfloor$ neighbors from her neighbor list.
This is called edge clipping in \cite{Imola_USENIX22}.
Then, $v_i$ calculates the number $r_i \in \nnints$ of 2-stars of which she is a center.
User $v_i$ adds $\Lap(\frac{\td_i}{\epsilon_2})$ to $r_i$ to obtain a noisy 2-star count $\tilde{r}_i = r_i + \Lap(\frac{\td_i}{\epsilon_2})$.
Because the sensitivity of the $k$-star count is $\binom{\td_i}{k-1}$, the noisy 2-star count $\tilde{r}_i$ provides $\epsilon_2$-edge LDP.
User $v_i$ sends the noisy degree $\td_i$ and the noisy 2-star count $\tilde{r}_i$ to the data collector.
Finally, the data collector estimates the 2-star count as $\sum_{i=1}^n \tilde{r}_i$.
By composition, this algorithm provides $(\epsilon_1 + \epsilon_2)$-edge LDP.
As with \cite{Imola_USENIX22}, we set $\eta = 150$ and divided the total privacy budget $\epsilon$ as $\epsilon_1 = \frac{\epsilon}{10}$ and $\epsilon_1 = \frac{9\epsilon}{10}$.

Let $\hf^\triangle(G)$ be the estimate of the triangle count by \AlgWSTriVR{} and $\hf^{2*}(G)$ be the estimate of the 2-star count by the above algorithm.
Then we estimated the clustering coefficient as $\frac{3 \hf^\triangle(G)}{\hf^{2*}(G)}$.

Figure~\ref{fig:res6_clst} shows the relative errors of the triangle count, 2-star count, and clustering coefficient in \Gplus{} and \IMDB{}.
We observe that the relative error of the clustering coefficient is almost the same as that of the triangle count.
This is because the 2-star algorithm is very accurate, as shown in Figure~\ref{fig:res6_clst}.
2-stars are much easier to count than triangles in the local model, as each user can count her 2-stars.
As a result, the error in the clustering coefficient is mainly caused by the error in $\hf^\triangle(G)$, which explains the results in Figure~\ref{fig:res6_clst}.

In Figure~\ref{fig:res6_clst}, we use the privacy budget $\epsilon$ for both $\hf^\triangle(G)$ and $\hf^{2*}(G)$.
In this case, we need $2\epsilon$ to calculate the clustering coefficient.
However, as shown in Figure~\ref{fig:res6_clst}, we can accurately estimate the 2-star count with a very small $\epsilon$; e.g., the relative error is around $10^{-2}$ when $\epsilon=0.1$.
Therefore, we can accurately calculate the clustering coefficient
with a very small additional budget
by using such a small $\epsilon$ for 2-stars.

In summary, our triangle algorithm \AlgWSTriVR{} is useful for accurately calculating the clustering coefficient within one round.

\begin{figure}[t]
  \centering
  \includegraphics[width=0.99\linewidth]{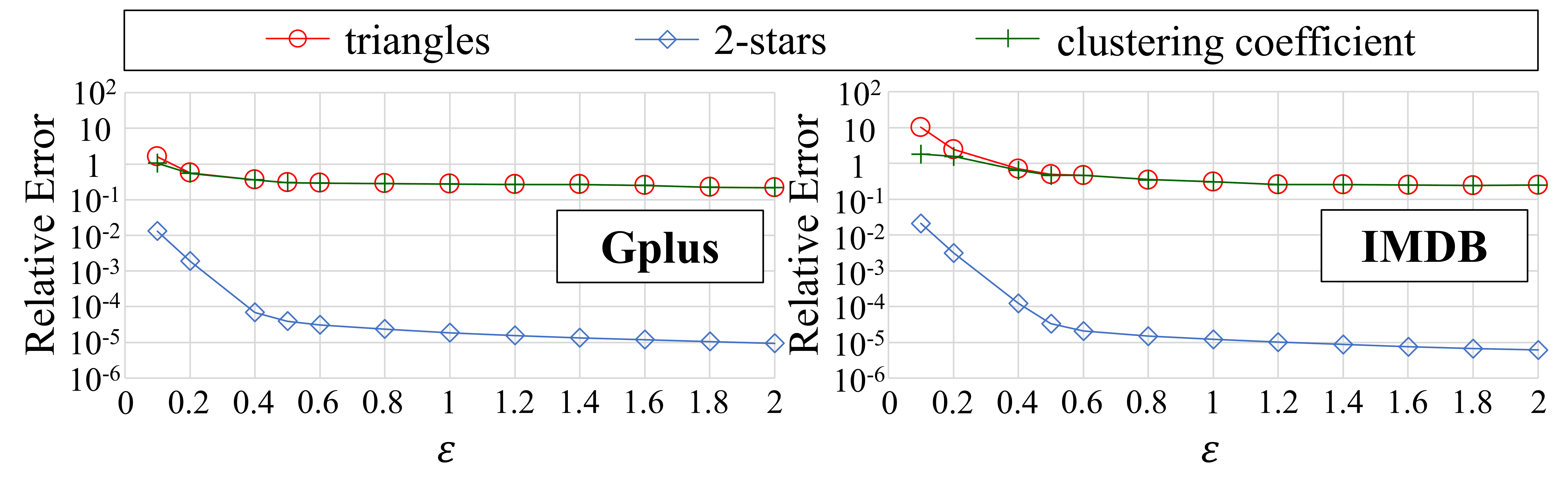}
  \vspace{-2mm}
  \caption{Relative errors of the triangle count, 2-star count, and clustering coefficient when \AlgWSTriVR{} and the one-round local 2-star algorithm in \cite{Imola_USENIX21} with edge clipping are used ($n=107614$ in \Gplus{}, $n=896308$ in \IMDB{}, $c=1$).
  }
  \label{fig:res6_clst}
\end{figure}

\section{Comparison with Two-Round Local Algorithms}
\label{sec:two-round}
In this work, we focus on one-round algorithms because multi-rounds algorithms require a lot of user effort and synchronization.
However, it is interesting to see how our one-round algorithms compare with the existing two-round local algorithms \cite{Imola_USENIX21,Imola_USENIX22} in terms of accuracy, as the existing two-round algorithms provide high accuracy.
Since they focus on triangle counting, we focus on this task.

We evaluate the two-round local algorithm in \cite{Imola_USENIX22} because it outperforms \cite{Imola_USENIX21} in terms of both the accuracy and communication efficiency.
The algorithm in \cite{Imola_USENIX22} works as follows.
At the first round, each user $v_i$ obfuscates bits $a_{i,1}, \ldots, a_{i,i-1}$ for smaller user IDs in her neighbor list $\bma_i$ (i.e., lower triangular part of $\bmA$) by the ARR and sends the noisy neighbor list to the data collector.
The data collector constructs a noisy graph $G'=(V,E')$ from the noisy neighbor lists.
At the second round, each user $v_i$ downloads some noisy edges $(v_j,v_k) \in E'$
from the data collector and counts noisy triangles $(v_i, v_j, v_k)$ so that
only one edge $(v_j, v_k)$ is noisy.
User $v_i$ adds the Laplacian noise to the noisy triangle count and sends it to the data collector.
Finally, the data collector calculates an unbiased estimate of the triangle count.
This algorithm provides $\epsilon$-edge LDP.
The authors in \cite{Imola_USENIX22} propose some strategies to select noisy edges to download at the second round.
We use a strategy to download noisy edges $(v_j,v_k) \in E'$ such that a noisy edge is connected from $v_k$ to $v_i$ (i.e., $(v_i,v_k) \in E'$) because it provides the best performance.

The algorithm in \cite{Imola_USENIX22} controls the trade-off between the accuracy and the download cost (i.e., the size of noisy edges) at the second round by changing the sampling probability $p_0$ in the ARR.
It is shown in \cite{Imola_USENIX22} that
when $p_0 = 1$, the MSE is $O(n d_{max}^3)$ and the download cost of each user is $\frac{(n-1)(n-2)}{2}$ bits.
In contrast,
when $p_0 = O(n^{-1/2})$, the MSE is $O(n^2 d_{max}^3)$ and the download cost is $O(n \log n)$.
We evaluated these two settings.
For the latter setting,
we set $p_0 = \frac{1}{q \sqrt{n}}$ where $q=\frac{e^\epsilon}{e^\epsilon+1}$ so that the download cost is $n \log n$ bits.
We denote the two-round algorithm with $p_0=1$ and $\frac{1}{q \sqrt{n}}$ by \AlgTwoRL{} and \AlgTwoRS{}, respectively.
\AlgTwoRL{} requires a larger download cost.

\begin{figure}[t]
  \centering
  \includegraphics[width=0.99\linewidth]{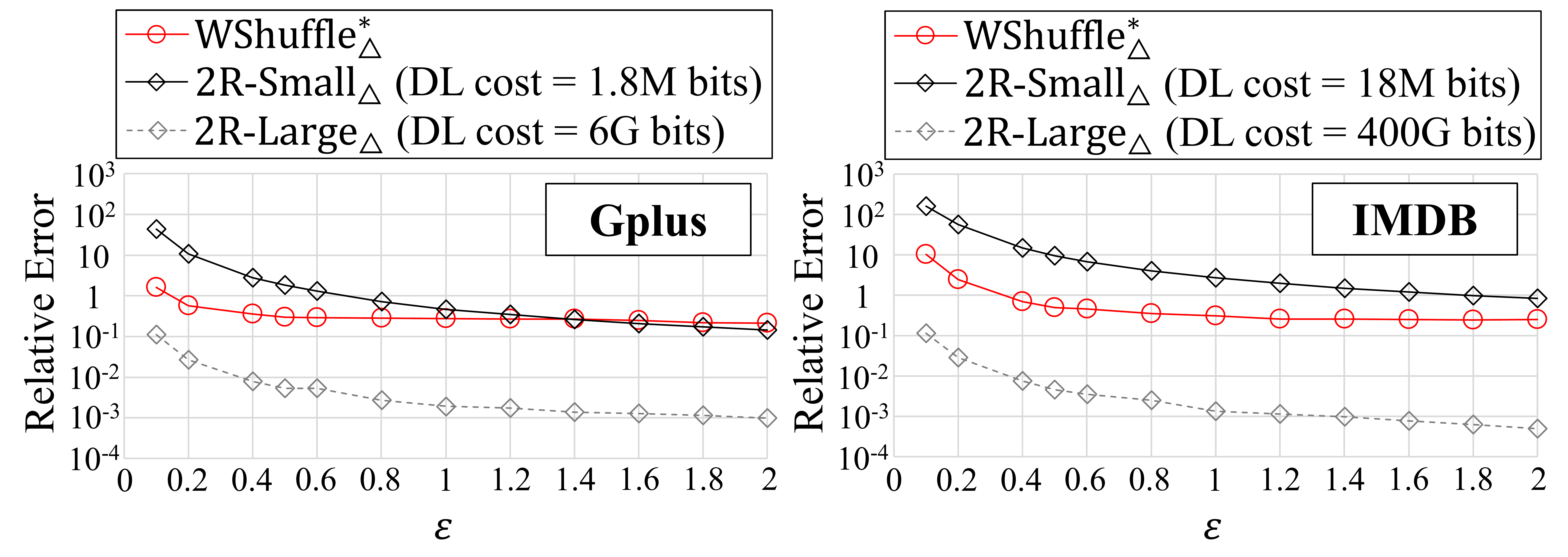}
  \vspace{-2mm}
  \caption{Comparison with the two-round local algorithm in \cite{Imola_USENIX22}.
  The download costs of \AlgTwoRS{} and \AlgTwoRL{} are $n \log n$ and $\frac{(n-1)(n-2)}{2}$ bits, respectively
  ($n=107614$ in \Gplus{}, $n=896308$ in \IMDB{}, $c=1$).
  }
  \label{fig:res3_2rounds}
\end{figure}

Figure~\ref{fig:res3_2rounds} shows the results.
We observe that our \AlgWSTriVR{} is outperformed by \AlgTwoRL{}.
This is expected, as \AlgWSTriVR{} and \AlgTwoRL{} provide the MSE of $O(n^2)$ and $O(n)$, respectively (when we ignore $d_{max}$).
However, \AlgTwoRL{} is impractical because it requires a too large download cost: 6G and 400G bits per user in \Gplus{} and \IMDB{}, respectively.
\AlgTwoRS{} is much more efficient (1.8M and 18M bits in \Gplus{} and \IMDB{}, respectively), and our \AlgWSTriVR{} is comparable to or outperforms \AlgTwoRS{}\footnote{As with \AlgARRTri{}, \AlgTwoRL{} and \AlgTwoRS{} provide $\epsilon$-edge DP (rather than $2\epsilon$-edge DP) because it uses only the lower-triangular part of $\bmA$.
However, our conclusion is the same even if we double $\epsilon$ for only \AlgWSTriVR{}.}.
This is also consistent with the theoretical results because both \AlgWSTriVR{} and \AlgTwoRS{} provide the MSE of $O(n^2)$.

In summary, our \AlgWSTriVR{} is comparable to the two-round local algorithm in \cite{Imola_USENIX22} (\AlgTwoRS{}), which requires a lot of user effort and synchronization, in terms of accuracy.

\section{Comparison between the Numerical Bound and the Closed-form Bound}
\label{sec:numerical_closed}
In Section~\ref{sec:experiments}, we used the numerical upper bound in \cite{Feldman_FOCS21} for calculating $\epsilon$ in the shuffle model.
Here, we compare the numerical bound with the closed-form bound in Theorem~\ref{thm:shuffle}.

\begin{figure}[t]
  \centering
  \includegraphics[width=0.99\linewidth]{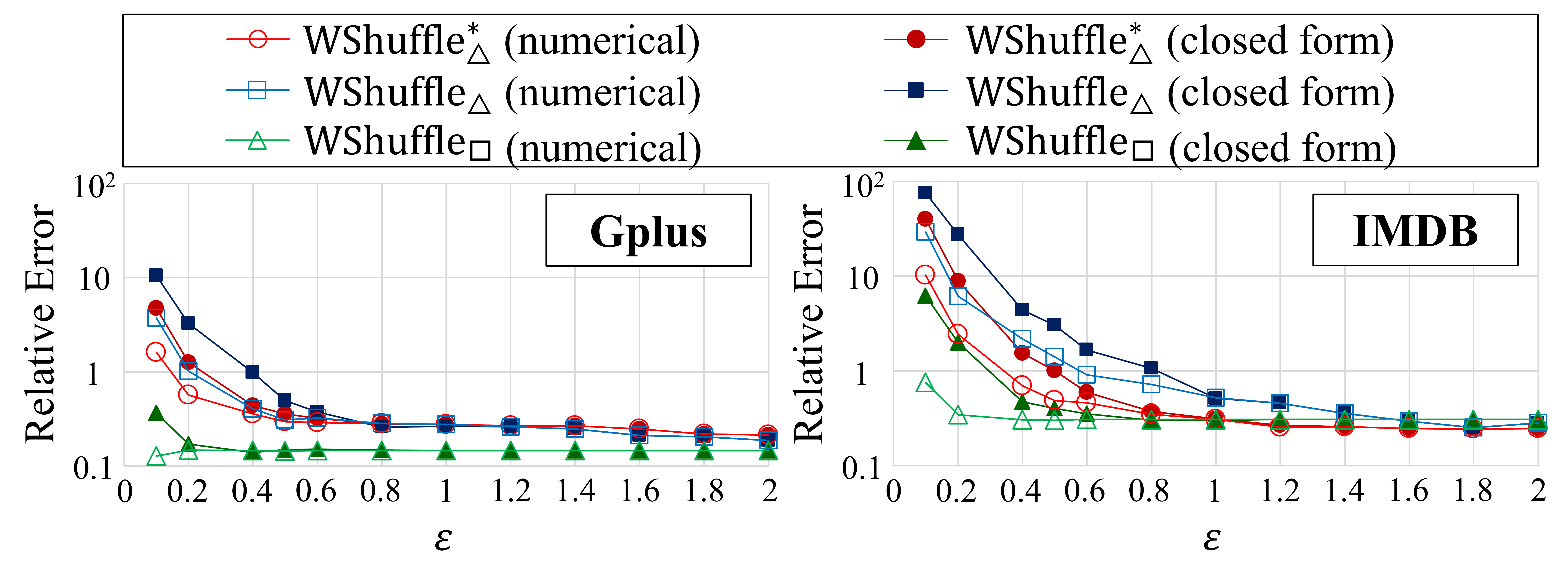}
  \vspace{-2mm}
  \caption{Numerical bound vs. closed-form bound ($n=107614$ in \Gplus{}, $n=896308$ in \IMDB{}, $c=1$).
  }
  \label{fig:res5_closed}
\end{figure}

Figure~\ref{fig:res5_closed} shows the results for \AlgWSTriVR{}, \AlgWSTri{}, and \AlgWSCyc{}.
We observe that the numerical bound provides a smaller relative error than the closed-form bound when $\epsilon$ is small.
However, when $\epsilon \geq 1$, the relative error is almost the same between the numerical bound and the closed-form bound.
This is because when $\epsilon \geq 1$, the corresponding $\epsilon_L$ is close to the maximum value $\log (\frac{n}{16 \log (2/\delta)})$ ($=5.86$ in \Gplus{} and $7.98$ in \IMDB{}) in both cases.
Thus, for a large $\epsilon$, the closed-form bound is sufficient.
For a small $\epsilon$, the numerical bound is preferable.

\section{Experiments on the Barab\'{a}si-Albert Graphs}
\label{sec:BA-graph}
In Section~\ref{sec:experiments}, we used \Gplus{} and \IMDB{} as datasets.
We also evaluated our algorithms using synthetic datasets based on the BA (Barab\'{a}si-Albert) graph model \cite{NetworkScience}, which has a power-law degree distribution.

The BA graph model generates a graph by adding new nodes one at a time.
Each new node has $m \in \nats$ new edges, and each new edge is randomly connected an existing node with probability proportional to its degree.
The average degree is almost $d_{avg} = 2m$, and most users' degrees are $m$.
We set $m=100$ or $200$ and used the NetworkX library \cite{Hagberg_SciPy08} (barabasi\_albert\_graph function) to generate a synthetic graph based on the BA model.
For the number $n$ of users, we set $n=107614$ (same as \Gplus{}) to compare the results between \Gplus{} and the BA graphs.
Table~\ref{tab:BA_graphs} shows some statistics of \Gplus{} and the BA graphs.
It is well known that the BA model has a low clustering coefficient \cite{Holme_PRE02}.
Thus, the BA graphs have much smaller triangles and 4-cycles than \Gplus{}.

\begin{table}[t]
  \caption{Statistics of \Gplus{} and the BA graphs ($n=107614$).
  }
  \vspace{-4mm}
  \centering
  \begin{tabular}{|c|c|c|c|c|}
    \hline
    & $d_{avg}$ & $d_{max}$ & \#triangles & \#4-cycles\\ \hline
    \Gplus{} & $227.4$ & $20127$ & $1.07 \times 10^{9}$ & $1.42 \times 10^{12}$ \\ \hline
    BA ($m=100$) & $199.8$ & $5361$ & $1.56 \times 10^7$ & $5.31 \times 10^9$ \\ \hline
    BA ($m=200$) & $399.3$ & $7428$ & $9.86 \times 10^7$ & $6.21 \times 10^{10}$ \\ \hline
  \end{tabular}
  \label{tab:BA_graphs}
\end{table}

Figure~\ref{fig:res7_BA} shows the results in the BA graphs.
We observe that the relative error is smaller when $m=200$.
This is because the BA graph with $m=200$ includes larger numbers of true triangles and 4-cycles,
as shown in Table~\ref{tab:BA_graphs}.
In this case, the denominator in the relative error is larger, and consequently the relative error becomes smaller.
By Figures~\ref{fig:res1_eps} and \ref{fig:res7_BA}, the relative error in the BA graph with $m=100$ is larger than the relative error in \Gplus{}.
The reason for this is the same -- \Gplus{} includes larger numbers of triangles and 4-cycles, as shown in Table~\ref{tab:BA_graphs}.
These results show that the relative error tends to be smaller in a dense graph that includes a larger number of subgraphs.

\begin{figure}[t]
  \centering
  \includegraphics[width=0.99\linewidth]{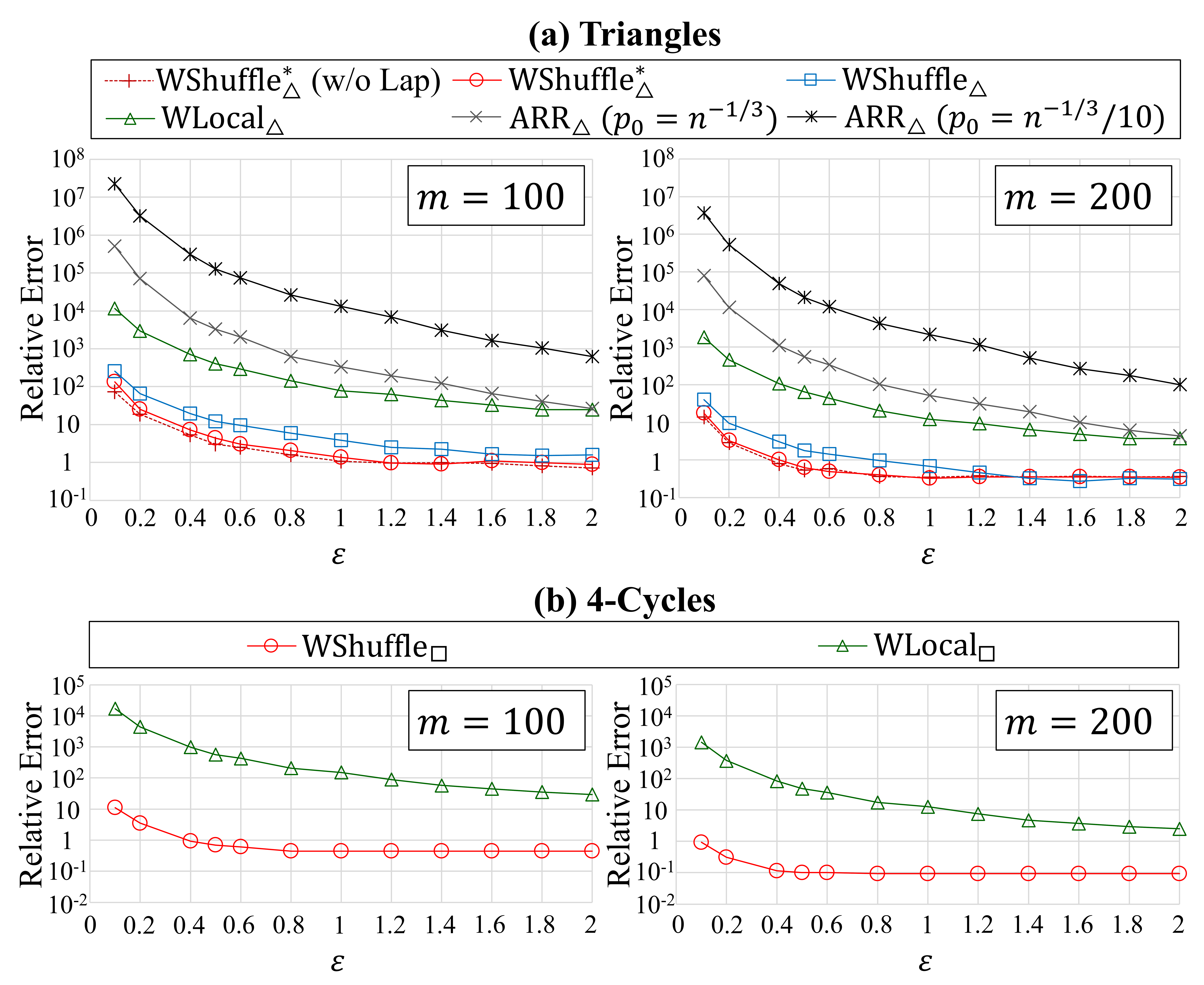}
  \vspace{-2mm}
  \caption{Relative error in the BA graph data
  ($n=107614$, $c=1$).
  $p_0$ is the sampling probability in the ARR.
  }
  \label{fig:res7_BA}
\vspace{2mm}
  \centering
  \includegraphics[width=0.8\linewidth]{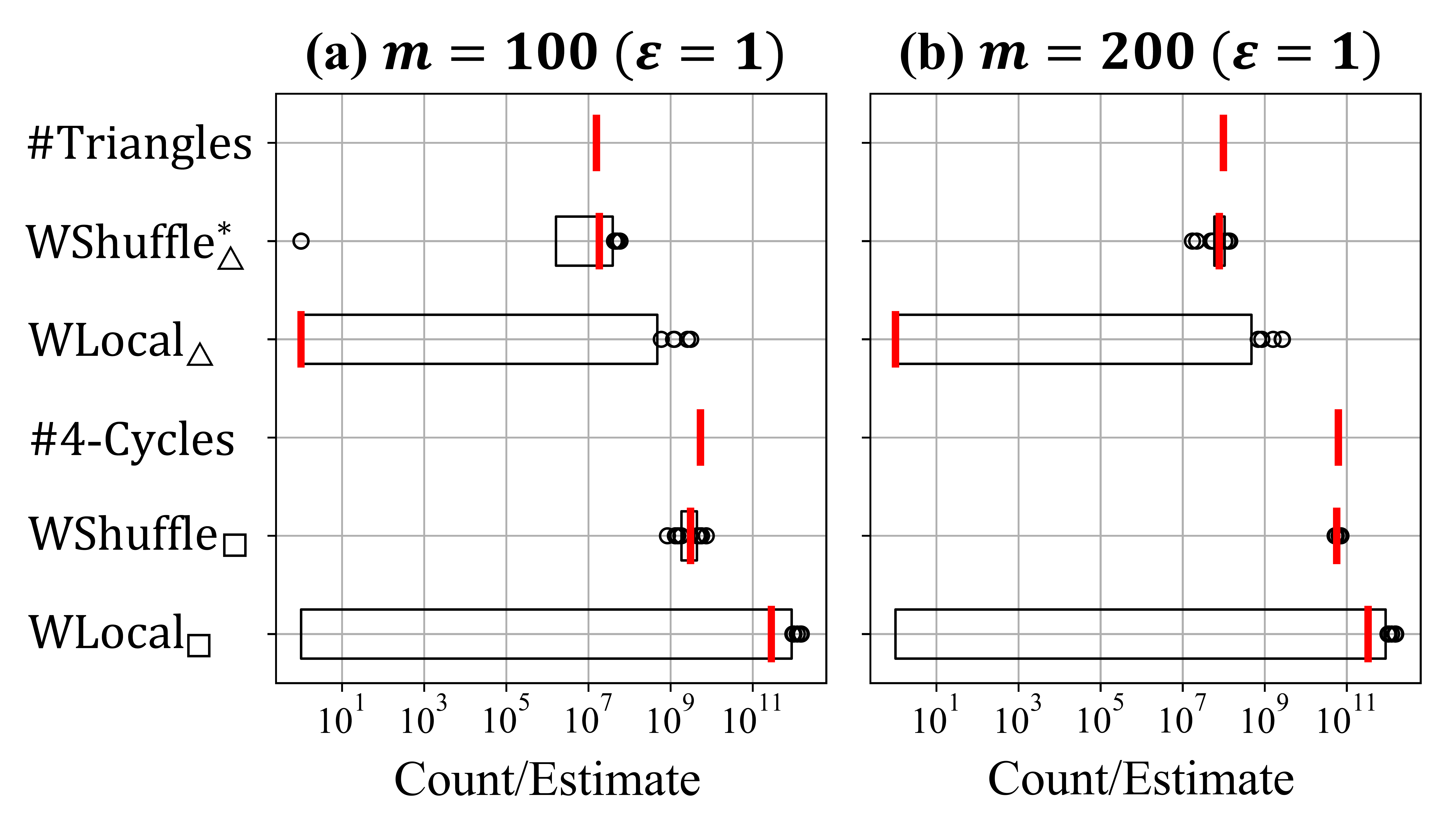}
  \vspace{-2mm}
  \caption{Box plots of counts/estimates in the BA graph data
  ($n=107614$, $c=1$).
  \#Triangles and \#4-Cycles represent the true triangle and 4-cycle counts, respectively.
  The box plot of each algorithm represents the median (red), lower/upper quartile, and outliers (circles) of $20$ estimates.
  The leftmost values are smaller than $1$.
  }
  \label{fig:res7_BA_box}
\end{figure}

Figures~\ref{fig:res7_BA} shows that when $\epsilon=1$, the relative errors of \AlgWSTriVR{} ($m=100$), \AlgWSCyc{} ($m=100$), \AlgWSTriVR{} ($m=200$), and \AlgWSCyc{} ($m=200$) are $1.36$, $0.447$, $0.323$, and $0.0928$, respectively.
Although the relative error of \AlgWSTriVR{} is about $1$ when $\epsilon=1$ and $m=100$, we argue that it is still useful for calculating a rough estimate of the triangle count.
To explain this, we show box plots of counts or estimates in the BA graphs in Figure~\ref{fig:res7_BA_box}.
This figure shows that the true triangle count is about $10^7$ and that \AlgWSTriVR{} ($m=100$) successfully calculates a rough estimate ($10^6 \sim 10^8$) in most cases ($15$ out of $20$ cases).
\AlgWSCyc{} ($m=100$), \AlgWSTriVR{} ($m=200$), and \AlgWSCyc{} ($m=200$) are much more accurate and successfully calculate an estimate in all cases.
In contrast, the local algorithms \AlgWLTri{} and \AlgWLCyc{} fail to calculate a rough estimate.

In summary, our shuffle algorithms significantly outperform the local algorithms and calculate a (rough) estimate of the triangle/4-cycle count with a reasonable privacy budget (e.g., $\epsilon=1$) in the BA graph data.

\section{Experiments on the Bipartite Graphs}
\label{sec:bipartite}
As described in Section~\ref{sec:intro}, the 4-cycle count is useful for measuring the clustering tendency in a bipartite graph where no triangles appear.
Therefore, we also evaluated our 4-cycle counting algorithms using bipartite graphs generated from \Gplus{} and \IMDB{}.

Specifically, for each dataset, we randomly divided all users into two groups with equal number of users.
The number of users in each group is $53807$ in \Gplus{} and $448154$ in \IMDB{}.
Then, we constructed a bipartite graph by removing edges within each group.
We refer to the bipartite versions of \Gplus{} and \IMDB{} as the \textit{bipartite \Gplus{}} and \textit{bipartite \IMDB{}}, respectively.
Using these datasets, we evaluated the relative errors of \AlgWSCyc{} and \AlgWLCyc{}.
Note that we did not evaluate the triangle counting algorithms, because there are no triangles in these graphs.

\begin{figure}[t]
  \centering
  \includegraphics[width=0.99\linewidth]{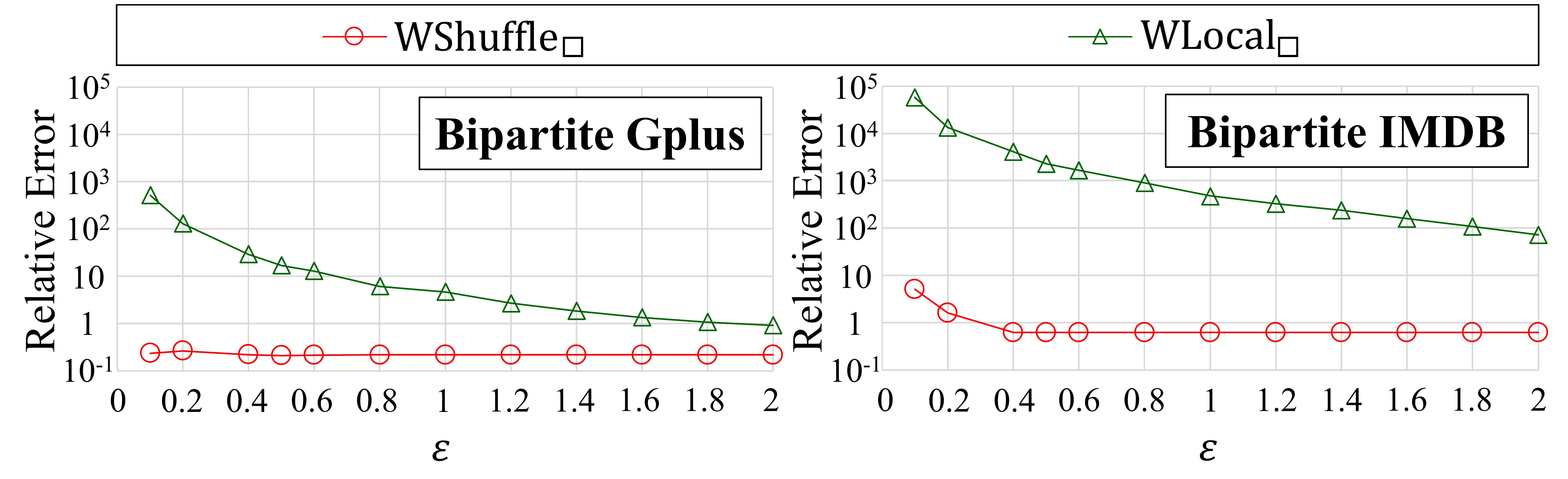}
  \vspace{-2mm}
  \caption{Relative error in the bipartite graph data
  ($n=107614$ in \Gplus{}, $n=896308$ in \IMDB{}, $c=1$).
  }
  \label{fig:res8_Bip}
\vspace{2mm}
  \centering
  \includegraphics[width=0.92\linewidth]{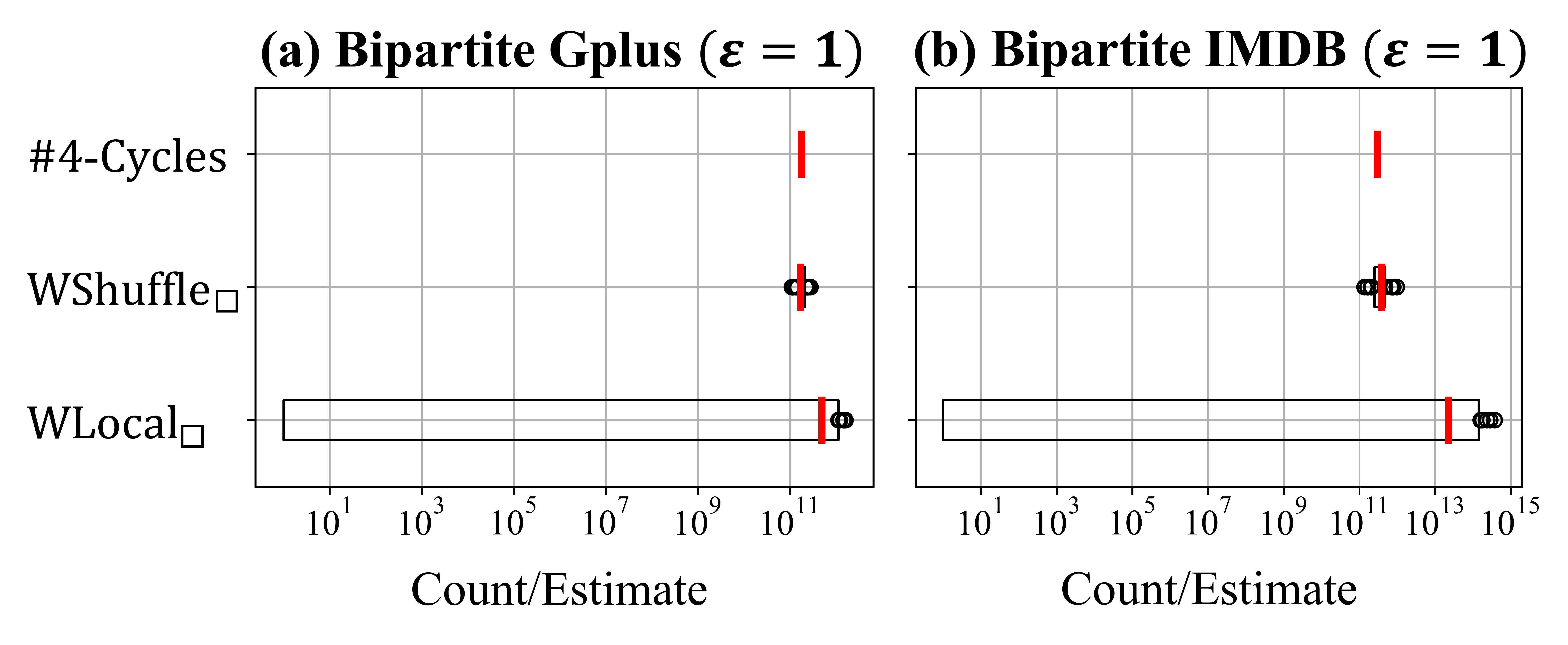}
  \vspace{-2mm}
  \caption{Box plots of counts/estimates in the bipartite graph data
  ($n=107614$ in \Gplus{}, $n=896308$ in \IMDB{}, $c=1$).
  \#4-Cycles represents the true 4-cycle count.
  Each box plot represents the median (red), lower/upper quartile, and outliers (circles) of $20$ estimates.
  The leftmost values are smaller than $1$.
  }
  \label{fig:res8_Bip_box}
\end{figure}

Figure~\ref{fig:res8_Bip} shows the results.
We observe that \AlgWSCyc{} significantly outperforms \AlgWLCyc{} in these datasets.
Compared to Figure~\ref{fig:res1_eps}(b), the relative error of \AlgWSCyc{} is a bit larger in the bipartite graph data.
For example, when $\epsilon=1$, the relative error of \AlgWSCyc{} is $0.147$, $0.308$, $0.217$, and $0.626$ in \Gplus{}, \IMDB{}, the bipartite \Gplus{}, and the bipartite \IMDB{}, respectively.
This is because the 4-cycle count is reduced by removing edges within each group.
In \Gplus{}, the 4-cycle count is reduced from $1.42 \times 10^{12}$ to $1.77 \times 10^{11}$.
In \IMDB{}, it is reduced from $2.37 \times 10^{12}$ to $2.96 \times 10^{11}$.
Consequently, the denominator in the relative error becomes smaller, and the relative error becomes larger.

Although the relative error of \AlgWSCyc{} with $\epsilon=1$ is about $0.6$ in the bipartite \IMDB{}, \AlgWSCyc{} still calculates a rough estimate of the 4-cycle count.
Figure~\ref{fig:res8_Bip_box} shows box plots of counts or estimates in the bipartite graph data.
This figure shows that \AlgWLCyc{} fails to estimate the 4-cycle count.
In contrast, \AlgWSCyc{} successfully calculates a rough estimate of the 4-cycle count in all cases.

In summary, our \AlgWSCyc{} significantly outperforms \AlgWLCyc{} and accurately counts 4-cycles in the bipartite graphs as well.

\section{Standard Error of the Average Relative Error}
\label{sec:standard_error}

In Section~\ref{sec:experiments}, we evaluated the average relative error over $20$ runs for each algorithm. 
In this appendix, we evaluate the standard error of the average relative error. 

\begin{figure}[t]
  \centering
  \includegraphics[width=0.99\linewidth]{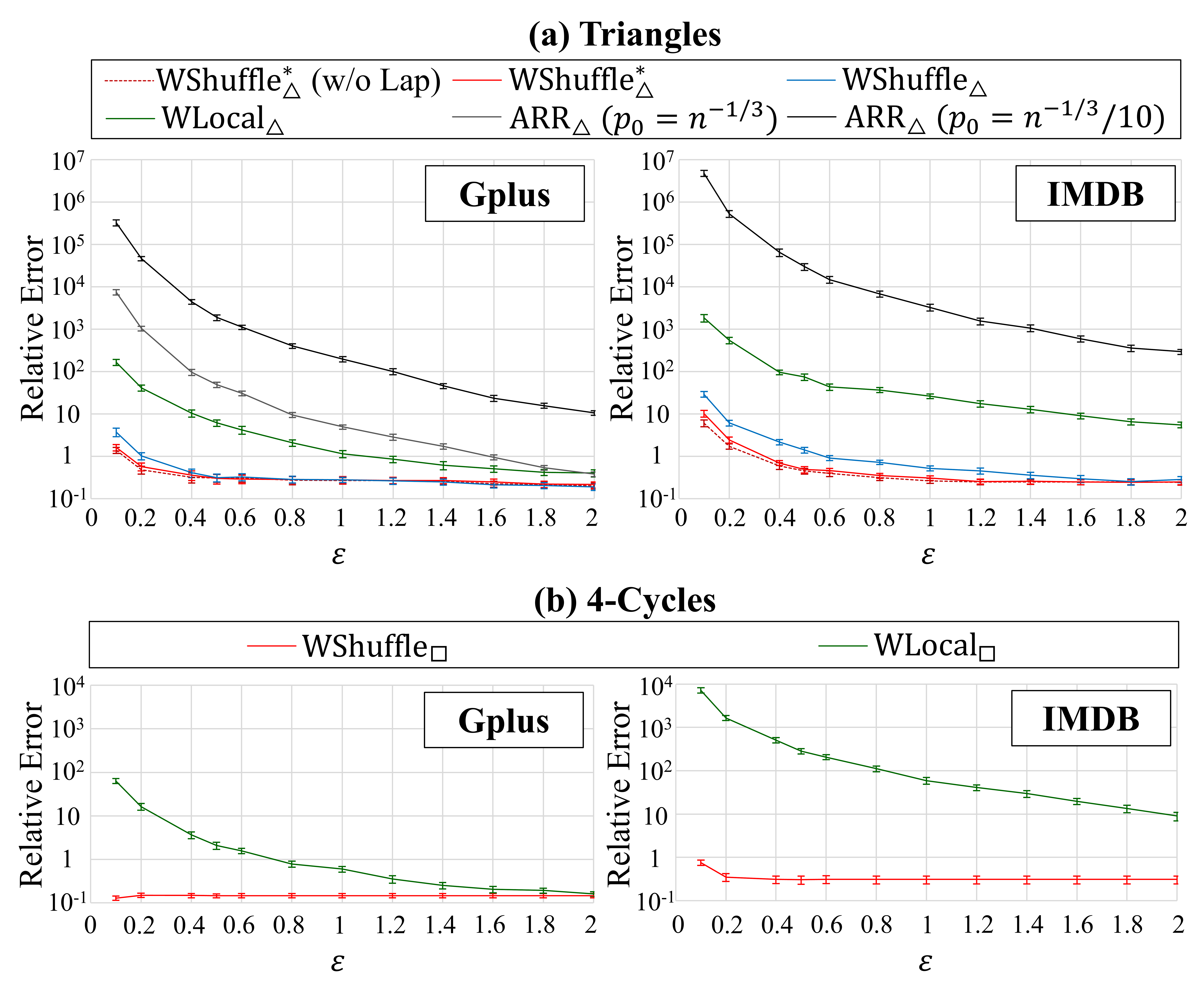}
  \vspace{-4mm}
  \caption{Standard error of the average relative error in Figure~\ref{fig:res1_eps}. 
  Each error bar represents $\pm$ standard error.
  }
  \label{fig:res9_std_err}
\end{figure}

Figure~\ref{fig:res9_std_err} shows the standard error of the average relative error in Figure~\ref{fig:res1_eps}. 
We observe that the standard error is small. 
For example, as shown in Table~\ref{tab:res1_eps_tri_time} (a), the average relative error of \AlgWSTriVR{} is $0.298$ ($\epsilon=0.5$) or $0.277$ ($\epsilon=1$) in \Gplus{}. 
Figure~\ref{fig:res9_std_err} shows that the corresponding standard error is $0.078$ ($\epsilon=0.5$) or $0.053$ ($\epsilon=1$). Thus, we conclude that $20$ runs are sufficient in our experiments. 

\section{MSE of the Existing One-Round Local Algorithms}
\label{sec:upper}

Here, we show the MSE of the existing one-round local algorithms \AlgARRTri{} \cite{Imola_USENIX22} and \AlgRRTri{} \cite{Imola_USENIX21}.
Specifically, we prove the MSE of \AlgARRTri{} because \AlgARRTri{} includes \AlgRRTri{} as a special case; i.e., the MSE of \AlgRRTri{} is immediately derived from that of \AlgARRTri{}.

We also note that the MSE of \AlgRRTri{} is proved in \cite{Imola_USENIX21} under the assumption that a graph is generated from the Erd\"os-R\'enyi graph model \cite{NetworkScience}.
However, this assumption does not hold in practice, because the Erd\"os-R\'enyi graph does not have a power-law degree distribution.
In contrast, we prove the MSE of \AlgARRTri{} (hence \AlgRRTri{}) without making any assumption on graphs.

\smallskip
\noindent{\textbf{Algorithm.}}~~First, we briefly explain \AlgARRTri{}.
In this algorithm, each user $v_i$ obfuscates her neighbor list $\bma_i \in \{0,1\}^n$ using the ARR (Asymmetric Randomized Response) whose input domain and output range are $\{0,1\}$.
Specifically, the ARR has two parameters $\epsilon \in \nnreals$ and $\mu \in [0,\frac{e^\epsilon}{e^\epsilon + 1}]$.
Given $1$ (resp.~$0$), the ARR outputs $1$ with probability $\mu$ (resp.~$\mu e^{-\epsilon}$).
This mechanism is equivalent to $\epsilon$-RR
followed by edge sampling, which samples each 1 with probability $p_0$ satisfying $\mu = \frac{e^\epsilon}{e^\epsilon+1} p_0$.
User $v_i$ applies the ARR to bits $a_{i,1}, \ldots, a_{i,i-1}$ for smaller user IDs in her neighbor list $\bma_i$ (i.e., lower triangular part of $\bmA$) and sends the noisy bits to the data collector.
Then, the data collector constructs a noisy graph $G^*$ based on the noisy bits.

The data collector counts triangles, 2-edges (three nodes with two edges), 1-edge (three nodes with one edge), and no-edges (three nodes with no edges) in the noisy graph $G^*$.
Let $m_3^*, m_2^*, m_1^*, m_0^* \in \nnints$ be the numbers of triangles, 2-edges, 1-edge, and no-edges, respectively, in $G^*$.
Note that $m_3^* + m_2^* + m_1^* + m_0^* = \binom{n}{3}$.
Finally, the data collector estimates the number $f^\triangle(G)$ of triangles as follows:
\begin{align}
    \hf^\triangle(G) = \frac{1}{(e^\epsilon-1)^3}(e^{3\epsilon} \hat{m}_3 - e^{2\epsilon} \hat{m}_2 + e^\epsilon \hat{m}_1 - \hat{m}_0),
    \label{eq:ARR_hf}
\end{align}
where
\begin{align}
\hat{m}_3 &= \textstyle{\frac{m_3^*}{p_0^3}} \label{eq:ARR_m3}\\
\hat{m}_2 &= \textstyle{\frac{m_2^*}{p_0^2} - 3(1-p_0)\hat{m}_3} \label{eq:ARR_m2}\\
\hat{m}_1 &= \textstyle{\frac{m_1^*}{p_0} - 3(1-p_0)^2\hat{m}_3 - 2(1-p_0)\hat{m}_2} \label{eq:ARR_m1}\\
\hat{m}_0 &= \textstyle{\binom{n}{3} - \hat{m}_3 - \hat{m}_2 - \hat{m}_1}. \label{eq:ARR_m0}
\end{align}
\AlgRRTri{} is a special case of \AlgARRTri{} where $\mu = \frac{e^\epsilon}{e^\epsilon + 1}$ ($p_0 = 1$), i.e. without edge sampling.

\smallskip
\noindent{\textbf{Privacy and Time Complexity.}}~~The ARR is equivalent to $\epsilon$-RR followed by edge sampling, as explained above.
Therefore, \AlgARRTri{} provides $\epsilon$-edge LDP by the post-processing invariance \cite{DP}.

The time complexity of \AlgARRTri{} is dominated by counting the number $m_3$ of triangles in the noisy graph $G^*$.
The expectation of $m_3$ is upper bounded as $\E[\mu_3^*] \leq \mu^3 n^3$, as each user-pair has an edge in $G^*$ with probability at most $\mu$.
Thus, the time complexity of \AlgARRTri{} can be expressed as $O(\mu^3 n^3)$.
This is $O(n^2)$ when $\mu^3 = O(\frac{1}{n})$.

\smallskip
\noindent{\textbf{MSE.}}~~Below, we analyze the MSE of \AlgARRTri{}.
First, we show that \AlgARRTri{} provides an unbiased estimate\footnote{It is informally explained in \cite{Imola_USENIX22} that the estimate of \AlgARRTri{} is unbiased. We formalize their claim.}:
\begin{theorem}
\label{thm:unbiased_ARR}
In \AlgARRTri{}, $\E[\hf^\triangle(G)] = f^\triangle(G)$.
\end{theorem}
\begin{proof}
Let $m_3, m_2, m_1, m_0 \in \nnints$ be the numbers of triangles, 2-edges, 1-edge, and no-edges, respectively, in the noisy graph $G'$ obtained by applying only $\epsilon$-RR.
Because the ARR independently samples each edge with probability $p_0$, we have:
\begin{align}
\E[m_3^*] &= p_0^3 m_3 \label{eq:E_m3*}\\
\E[m_2^*] &= 3 p_0^2 (1 - p_0) m_3 + p_0^2 m_2 \label{eq:E_m2*}\\
\E[m_1^*] &= 3 p_0 (1 - p_0)^2 m_3 + 2 p_0 (1 - p_0) m_2 + p_0 m_1. \label{eq:E_m1*}
\end{align}
By (\ref{eq:ARR_hf}), we have:
\begin{align}
& \E[\hf^\triangle(G)] \nonumber\\
&= \frac{1}{(e^\epsilon-1)^3}(e^{3\epsilon} \E[\hat{m}_3] - e^{2\epsilon} \E[\hat{m}_2] + e^\epsilon \E[\hat{m}_1] - \E[\hat{m}_0]) \nonumber\\
&= \frac{1}{(e^\epsilon-1)^3}(e^{3\epsilon} \E[\hat{m}_3] - e^{2\epsilon} \E[\hat{m}_2] + e^\epsilon \E[\hat{m}_1] - \E[\hat{m}_0]). \label{eq:E_hf_triangle_G}
\end{align}
By (\ref{eq:ARR_m3}), (\ref{eq:ARR_m2}), (\ref{eq:ARR_m1}), (\ref{eq:ARR_m0}), (\ref{eq:E_m3*}), (\ref{eq:E_m2*}), and (\ref{eq:E_m1*}), we have:
\begin{align*}
\E[\hat{m}_3]
&= \textstyle{\frac{\E[m_3^*]}{p_0^3}}
= \E[m_3] \\
\E[\hat{m}_2]
&= \textstyle{\frac{\E[m_2^*]}{p_0^2} - 3(1-p_0)\E[\hat{m}_3]} \\
&= 3 (1 - p_0) m_3 + m_2 - 3(1-p_0)m_3 \\
&= \E[m_2] \\
\E[\hat{m}_1]
&= \textstyle{\frac{\E[m_1^*]}{p_0} - 3(1-p_0)^2\E[\hat{m}_3] - 2(1-p_0)\E[\hat{m}_2]} \\
&= 3 (1 - p_0)^2 \E[m_3] + 2 (1 - p_0) \E[m_2] + \E[m_1] \\
& \hspace{3mm} - 3(1-p_0)^2 \E[m_3] - 2(1-p_0) \E[m_2] \\ &= \E[m_1] \\
\E[\hat{m}_0]
&= \textstyle{\binom{n}{3} - \E[\hat{m}_3] - \E[\hat{m}_2] - \E[\hat{m}_1]} \\
&= \textstyle{\binom{n}{3} - \E[m_3] - \E[m_2] - \E[m_1]} \\
&= \E[m_0].
\end{align*}
Thus, the equality (\ref{eq:E_hf_triangle_G}) can be written as follows:
\begin{align}
& \E[\hf^\triangle(G)] \nonumber\\
&= \frac{1}{(e^\epsilon-1)^3}(e^{3\epsilon} \E[m_3] - e^{2\epsilon} \E[m_2] + e^\epsilon \E[m_1] - \E[m_0]).
\label{eq:E_hf_traingle_G2}
\end{align}
Finally, we use the following lemma:

\begin{lemma}\label{lemma:triangle_emp}[Proposition 2 in \cite{Imola_USENIX21}]
  \begin{align}
      \textstyle{\mathbb{E}\left[ \frac{e^{3\epsilon}}{(e^\epsilon-1)^3} m_3 \hspace{-0.5mm}-\hspace{-0.5mm} \frac{e^{2\epsilon}}{(e^\epsilon-1)^3} m_2 \hspace{-0.5mm}+\hspace{-0.5mm} \frac{e^\epsilon}{(e^\epsilon-1)^3} m_1 \hspace{-0.5mm}-\hspace{-0.5mm} \frac{1}{(e^\epsilon-1)^3} m_0 \right] \hspace{-0.5mm} = \hspace{-0.5mm} f_\triangle(G).}
      \label{eq:triangle_emp}
  \end{align}
\end{lemma}
See \cite{Imola_USENIX21} for the proof of Lemma~\ref{lemma:triangle_emp}.
By (\ref{eq:E_hf_traingle_G2}) and Lemma~\ref{lemma:triangle_emp}, we have $\E[\hf^\triangle(G)] = f_\triangle(G)$.
\end{proof}

Next, we show the MSE ($=$ variance) of \AlgARRTri{}:
\begin{theorem}
\label{thm:var_ARR}
When we treat $\epsilon$ as a constant,
\AlgARRTri{} provides the following utility guarantee:
\begin{align*}
& \MSE(\hf^\triangle(G)) = \V[\hf^\triangle(G)] = O\left(\frac{n^4}{\mu^6}\right).
\end{align*}
\end{theorem}
By Theorem~\ref{thm:var_ARR}, the MSE of \AlgARRTri{} is $O(n^6)$ when we set $\mu^3=O(\frac{1}{n})$ so that the time complexity is $O(n^2)$.
The MSE of \AlgRRTri{} ($\mu=1$) is $O(n^4)$.
Below, we prove Theorem~\ref{thm:var_ARR}.

\begin{proof}
Let $d_3 = \frac{e^{3\epsilon}}{(e^\epsilon-1)^3}$, $d_2 = - \frac{e^{2\epsilon}}{(e^\epsilon-1)^3}$,
$d_1 = \frac{e^{\epsilon}}{(e^\epsilon-1)^3}$, and
$d_0 = - \frac{1}{(e^\epsilon-1)^3}$.
Then, by (\ref{eq:ARR_hf}), we have:
\begin{align}
    \V[\hf^\triangle(G)]
    &= \V[d_3 \hat{m}_3 + d_2 \hat{m}_2 + d_1 \hat{m}_1 + d_0 \hat{m}_0] \nonumber\\
    &= d_3^2 \V[\hat{m}_3] + d_2^2 \V[\hat{m}_2] + d_1^2 \V[\hat{m}_1] + d_0^2 \V[\hat{m}_0] \nonumber\\
    & \hspace{3mm} + \sum_{i \ne j} d_i d_j \cov(\hat{m}_i, \hat{m}_j).
    \label{eq:V_hf_cov}
\end{align}
By the Cauchy-Schwarz inequality,
\begin{align}
|\cov(\hat{m}_i, \hat{m}_j)|
&\leq \sqrt{\V[\hat{m}_i]\V[\hat{m}_j]} \nonumber\\
&\leq \max\{\V[\hat{m}_i], \V[\hat{m}_j]\} \nonumber\\
&\leq \V[\hat{m}_i] + \V[\hat{m}_j].
\label{eq:cov_hm_ij}
\end{align}
Therefore, we have:
\begin{align}
    \V[\hf^\triangle(G)]
    = O(\V[\hat{m}_3] + \V[\hat{m}_2] + \V[\hat{m}_1] + \V[\hat{m}_0]).
    \label{eq:V_hf_m3210}
\end{align}

Below, we upper bound $\V[\hf^\triangle(G)]$ in (\ref{eq:V_hf_m3210}) by bounding $\V[m_3^*], \ldots, \allowbreak \V[m_0^*]$ and then $\V[\hat{m}_3], \ldots, \allowbreak \V[\hat{m}_0]$.
Let $T_{i,j,k} \in \{0,1\}$ be a random variable that takes $1$ if and only if $v_i$, $v_j$, and $v_k$ form a triangle in the noisy graph $G^*$.
Then we have:
\begin{align*}
    \V[m_3^*]
    &= \V\left[\sum_{i,j,k} T_{i,j,k}\right] \\
    &= \sum_{i<j<k} \sum_{i'<j'<k'} \cov[T_{i,j,k}, T_{i',j',k'}].
\end{align*}
If $v_i, v_j, v_k$ and $v_{i'}, v_{j'}, v_{k'}$ intersect in zero or one node, then $T_{i,j,k}$ and $T_{i',j',k'}$ are independent
and their covariance is $0$.
There are only $O(n^4)$ choices of $v_i, v_j, v_k$ and $v_{i'}, v_{j'}, v_{k'}$ that intersect in two or more nodes, as there can only be $4$ distinct nodes.
Therefore, we have $\V[m_3^*] = O(n^4)$.
Similarly, we can prove $\V[m_2^*] = \V[m_1^*] = \V[m_0^*] = O(n^4)$ by regarding $T_{i,j,k}$ as a random variable that takes $1$ if and only if $v_i$, $v_j$, and $v_k$ form a 2-edge, 1-edge, and no-edges, respectively.
In summary, we have:
\begin{align}
    \V[m_3^*] = \V[m_2^*] = \V[m_1^*] = \V[m_0^*] = O(n^4).
    \label{eq:V_m3210_ast}
\end{align}
By (\ref{eq:V_m3210_ast}) and $\mu = \frac{e^\epsilon}{e^\epsilon+1} p_0$, we can upper bound the variance of $\hat{m}_3$ in (\ref{eq:ARR_m3}) as follows:
\begin{align*}
    \V[\hat{m}_3]
    = \frac{\V[m_3^*]}{p_0^6}
    = O\left(\frac{n^4}{\mu^6}\right).
\end{align*}
As with (\ref{eq:V_hf_cov}), (\ref{eq:cov_hm_ij}), and (\ref{eq:V_hf_m3210}), we can upper bound the variance of $\hat{m}_2$ in (\ref{eq:ARR_m2}) by using the Cauchy-Schwarz inequality as follows:
\begin{align*}
    & \V[\hat{m}_2] \\
    &= \textstyle{\V\left[\frac{m_2^*}{p_0^2} - \frac{3(1-p_0)}{p_0^3}m_3^*\right]} \\
    &= \frac{\V[m_2^*]}{p_0^4} + \frac{9(1-p_0)^2\V[m_3^*]}{p_0^6} - \frac{6(1-p_0)\cov(m_2^*,m_3^*)}{p_0^5} \\
    &\leq \frac{\V[m_2^*]}{p_0^4} + \frac{9(1-p_0)^2\V[m_3^*]}{p_0^6} + \frac{6(1-p_0)(\V[m_2^*]+\V[m_3^*])}{p_0^5} \\
    &= O\left(\frac{n^4}{\mu^6}\right) \\
\end{align*}
Similarly, we have
\begin{align*}
    &\V[\hat{m}_1] \\
    &= \textstyle{\V\left[\frac{m_1^*}{p_0} + \frac{3(1-p_0)^2}{p_0^3}m_3^* - \frac{2(1-p_0)}{p_0^2}m_2^* \right]} \\
    &= O\left(\frac{n^4}{\mu^6}\right) \\
    &\V[\hat{m}_0] \\
    &= \V[\hat{m}_3 + \hat{m}_2 + \hat{m}_1]\\
    &=
    \textstyle{\V\left[\frac{m_3^*}{p_0^3}
    + \frac{m_2^*}{p_0^2} - \frac{3(1-p_0)}{p_0^3}m_3^*
    + \frac{m_1^*}{p_0} + \frac{3(1-p_0)^2}{p_0^3}m_3^* - \frac{2(1-p_0)}{p_0^2}m_2^*
     \right]} \\
    &= O\left(\frac{n^4}{\mu^6}\right).
\end{align*}
In summary,
\begin{align}
    \V[\hat{m}_3] = \V[\hat{m}_2] = \V[\hat{m}_1] = \V[\hat{m}_0] = O\left(\frac{n^4}{\mu^6}\right).
    \label{eq:V_hm3210}
\end{align}
By (\ref{eq:V_hf_m3210}) and (\ref{eq:V_hm3210}),
$\V[\hf^\triangle(G)] = O\left(\frac{n^4}{\mu^6}\right)$.
\end{proof}

\section{Proofs of Statements in Section 5}
\label{sec:triangle_proof}

For these proofs, we will write $f_{i,\sigma}(G)$ as a shorthand for
$f_{\sigma(i), \sigma(i+1)}$ and $\hf_{i,\sigma}(G)$ as a shorthand for
$\hf_{\sigma(i), \sigma(i+1)}$.

\subsection{Proof of Theorem~\ref{thm:unbiased_I}}
\label{sub:unbiased_I_proof}
    From~\eqref{eq:hfij_triangle}, the quantity $\hf_{i,j}^\triangle(G)$ can be written as follows:
    \begin{align}
        \hf_{i,j}^\triangle(G)
        &= \frac{1}{2} (\hf_{i,j}^{(1)}(G) + \hf_{i,j}^{(2)}(G)),
        \label{eq:est_decomp}
    \end{align}
    where
    \begin{align}
        \hf_{i,j}^{(1)}(G)
        &= \frac{(z_{i,j}-q)\sum_{k \in I_{-(i,j)}} (y_{k} - q_L)}{(1-2q)(1-2q_L)} \label{eq:est1}\\
        \hf_{i,j}^{(2)}(G))
        &= \frac{(z_{j,i}-q)\sum_{k \in I_{-(i,j)}} (y_{k} - q_L)}{(1-2q)(1-2q_L)}, \label{eq:est2}
    \end{align}
    and each variable $y_k$ represents the output of the RR for the existence of wedge $w_{i-k-j}$ (see Algorithm~\ref{alg:WShuffle}).
    We call $\hf_{i,j}^{(1)}(G)$ and $\hf_{i,j}^{(2)}(G)$ the first and second estimates, respectively.

    \smallskip
    \noindent{\textbf{First Estimate.}}~~Since variables $z_{i,j}$ and $y_{k}$ in~\eqref{eq:est1} are independent, we have
    \begin{equation*}
        \E[\hf_{i,j}^{(1)}(G)] = \frac{\E[z_{i,j}-q]\sum_{k \in I_{-(i,j)}}
        \E[y_{k} - q_L]}{(1-2q)(1-2q_L)}.
    \end{equation*}
    First, suppose $(v_i,v_j) \notin E$. Then, $z_{i,j} = 1$ with probability $q$, and $\E[z_{i,j} - q] = 0$. This means $\E[\hf_{i,j}^{(1)}(G)] = f_{i,j}^\triangle(G) = 0$.
    Second, suppose $(v_i,v_j) \in E$. We have $z_{i,j}=1$ with probability $1-q$.
    We have $\E[z_{i,j} - q] = 1-2q$. For any $k \in I_{-(i,j)}$, if
    $w_{i-k-j} = 0$, then we have $\E[y_{k} - q_L] = 0$. If
    $w_{i-k-j} = 1$, then $\E[y_{k}-q_L] = 1-2q_L$. Written concisely, we
    can say $\E[y_{k}-q_L] = (1-2q_L)w_{i-k-j}$. Putting this together, we have
    \begin{align*}
        \E[\hf_{i,j}^{(1)}(G)] &= \frac{(1-2q)\sum_{k \in
        I_{-(i,j)}}(1-2q_L)w_{i-k-j}}{(1-2q)(1-2q_L)} \\
        &= \sum_{k \in I_{-(i,j)}} w_{i-k-j} \\
        &= f_{i,j}^\triangle(G).
    \end{align*}
    Thus, $\E[\hf_{i,j}^{(1)}(G)] = f_{i,j}^\triangle(G)$ holds for both cases.

    \smallskip
    \noindent{\textbf{Second and Average Estimates.}}~~Similarly, we can prove that $\E[\hf_{i,j}^{(2)}(G)] = f_{i,j}^\triangle(G)$ holds.
    Then, by (\ref{eq:est_decomp}), $\E[\hf_{i,j}^\triangle(G)] = f_{i,j}^\triangle(G)$ holds. \qed

\subsection{Proof of Theorem~\ref{thm:l2-loss_I}}
\label{sub:l2-loss_I_proof}

Recall that $\hf_{i,j}^\triangle(G)$ is an average of the first estimate $\hf_{i,j}^{(1)}(G)$ and second estimate $\hf_{i,j}^{(2)}(G)$; see (\ref{eq:est_decomp}), (\ref{eq:est1}), and (\ref{eq:est2}). We bound the variance of the first estimate.

\smallskip
\noindent{\textbf{First Estimate.}}~~Let $H = (z_{i,j}-q)\sum_{k \in I_{-(i,j)}} (y_{k} - q_L)$.
Using the law of total variance, we have
\begin{align*}
    &\V[\hf_{i,j}^{(1)}(G)] \\
    &= \frac{1}{(1-2q)^2(1-2q_L)^2}\left(\E_z[\V_y[H | z_{i,j}]] + \V_z[\E_y[H | z_{i,j}]]\right),
\end{align*}
where $\E_z$ (resp.~$\E_y$) represents the expectation over $z_{i,j}$ (resp.~$y_{k}$).
$\V_z$ (resp.~$\V_y$) represents the variance over $z_{i,j}$ (resp.~$y_{k}$).

We upper bound $\V_z[\E_y[H | z_{i,j}]]$ first. Let $E_y = \E_y[\sum_{k \in
I_{-(i,j)}} y_{k} \allowbreak -q_L]$. When $z_{i,j} = 0$, we have
\begin{align*}
\E_y[H] &= -q\E_y[\textstyle{\sum_{k \in I_{-(i,j)}} y_{k}-q_L}] = -qE_y.
\end{align*}
When $z_{i,j} = 1$, we have
\[
  \E_y[H] = (1-q)\E_y[\textstyle{\sum_{k \in I_{-(i,j)}} y_{k}-q_L}] = (1-q)E_y.
\]
The difference between these two quantities is $E_y$.
Since $z_{i,j}$ is a Bernoulli random variable with bias $q$ on either $0$ or
$1$, we have $\V_z[\E_y[H|z_{i,j}]] = E_y^2 q(1-q)$.
Recalling that $\E_y[y_{k}-q_L] = (1-2q_L)w_{i-k-j}$ (see the proof of Theorem~\ref{thm:unbiased_I}), by linearity of expectation we have that
\begin{align*}
    E_y &= \sum_{k \in I_{-(i,j)}} (1-2q_L)w_{i-k-j} \\
    &\leq (1-2q_L) d_{max}.
\end{align*}
Thus,
\[
  \V_z[\E_y[H|z_{i,j}]] \leq q(1-2q_L)^2 d_{max}^2.
\]

Now, we upper bound $\E_z[\V_y[H | z_{i,j}]]$. When $z_{i,j} = 0$, we have $H =
-q \sum_{k \in I_{-(i,j)}} (y_{k}-q)$, which is a sum of $n-2$ Bernoulli random variables.
Thus, $\V_y[H | z_{i,j}] = q^2 (n-2) q_L(1-q_L) \leq q^2 q_L n$. When $z_{i,j} = 1$, we have
by a similar argument that $\V_y[H | z_{i,j}] \leq (1-q)^2 q_L n$. Regardless of the value
of $z_{i,j}$, both values attainable by $\V_y[H | z_{i,j}]$ are at most $nq_L$.
Thus,
\[
\E_z[\V_y[H | z_{i,j}]] \leq n q_L.
\]

Putting all this together, we have the following upper-bound:
\[
  \V[\hf_{i,j}^{(1)}(G)] \leq \frac{n q_L + q(1-2q_L)^2 d_{max}^2}{(1-2q)^2(1-2q_L)^2}.
\]

\smallskip
\noindent{\textbf{Second and Average Estimates.}}~~Similarly, we can prove the same upper-bound for the second estimate $\hf_{i,j}^{(2)}(G)$.
Using~\eqref{eq:est_decomp} and Lemma~\ref{lem:var-sum} at the end of Appendix~\ref{sub:l2-loss_I_proof}, we have
$\V[\hf_{i,j}^\triangle(G)] \leq \V[\hf_{i,j}^{(1)}(G)]$, and the result follows.

\smallskip
\noindent{\textbf{Effect of Shuffling.}}~~When $\epsilon$ and $\delta$ are constants
and $\epsilon_L \geq \log n + O(1)$, we
have $nq_L = \frac{n}{e^{\epsilon_L} + 1} = O(1)$ and
$q_L = O(1/n)$,
and the bound becomes
\[
  \V[\hf_{i,j}^\triangle(G)] \leq O(d_{max}^2).
\]\qed

\begin{lemma}\label{lem:var-sum}
  Let $X,Y$ be two real-valued random variables. Then $\V[X+Y] \leq 4\max\{\V[X], \V[Y]\}$.
\end{lemma}
\begin{proof}
  We have $\V[X+Y] = \V[X] + \V[Y] + 2Cov(X,Y)$. By the Cauchy-Schwarz
  inequality, we have $Cov(X,Y) \leq \sqrt{\V[X]\V[Y]} \leq \max\{\V[X],
  \V[Y]\}$. Our result follows by observing $\V[X] + \V[Y] \leq 2\max\{\V[X],
  \V[Y]\}$.
\end{proof}
\subsection{Proof of Theorem~\ref{thm:DP_II}}
\label{sub:DP_II_proof}
First, we show that $\AlgWSLE{}$ meets the desired privacy requirements.
Let $x_k = w_{i-k-j}$. In step 1 of \AlgWSLE{},
for each $k \in I_{-(i,j)}$, user $v_k$
sends $y_k = \calR_{\epsilon_L}^W(x_k)$ to the shuffler.
Then, the shuffler sends $\{y_{\pi(k)} | k \in I_{-(i,j)}\}$ to the data collector.
Thus, by Theorem~\ref{thm:shuffle},
$\{x_k | k \in I_{-(i,j)}\}$
is protected with $(\epsilon, \delta)$-DP, where $\epsilon = f(n-2, \epsilon_L, \delta)$.
Note that changing $a_{k,i}$ will change $x_k$ if and only if $a_{k,j} = 1$.
Thus, for any $k \in I_{-(i,j)}$, $a_{k,i}$ and $a_{k,j}$ are protected with $(\epsilon, \delta)$-DP.

In step 1, user $v_i$ (resp.~$v_j$) sends $z_{i,j} = \calR_{\epsilon}^W(a_{i,j})$ (resp.~$z_{j,i} = \calR_{\epsilon}^W(a_{j,i})$) to the data collector.
Since $\calR_{\epsilon}^W$ provides $\epsilon$-DP, $a_{i,j}$ and $a_{j,i}$ are protected with $\epsilon$-DP.

Putting all this together, each element of the $i$-th and $j$-th columns in the adjacency matrix $\bmA$ is protected with $(\epsilon, \delta)$-DP.
Thus, by Proposition~\ref{prop:element_edge_DP},
\AlgWSLE{}
provides $(\epsilon, \delta)$-element-level DP and ($2\epsilon, 2\delta$)-edge DP.

\AlgWSTri{} interacts with $\bmA$ by calling \AlgWSLE{} on \[(v_{\sigma(1)},
v_{\sigma(2)}), \ldots, (v_{\sigma(2t-1)}, v_{\sigma(2t)}).\] Each of these calls
use disjoint elements of $\bmA$, and thus each element of $\bmA$ is still
protected by $(\epsilon, \delta)$-level DP and by $(2\epsilon, 2\delta)$-edge
DP.\qed

\subsection{Proof of Theorem~\ref{thm:unbiased_II}}
\label{sub:unbiased_II_proof}
Notice that the number of triangles in a graph can be computed as
\begin{align}
  6f^\triangle(G) &= \sum_{1 \leq i,j \leq n, i\neq j} f^\triangle_{i,j}(G)
  \nonumber \\
  &= n(n-1) \E_{\sigma}[f^\triangle_{i,
  \sigma}(G)],\label{eq:triangle-sample-est}
\end{align}
where $i \in [n]$ is arbitrary and $\sigma$ is a random permutation on
$[n]$. The constant $6$ appears because each triangle
appears six times when summing up $f^\triangle_{i,j}(G)$; e.g., a triangle $(v_1, v_2, v_3)$ appears in $f^\triangle_{1,2}(G)$, $f^\triangle_{2,1}(G)$,
$f^\triangle_{1,3}(G)$,
$f^\triangle_{3,1}(G)$,
$f^\triangle_{2,3}(G)$, and
$f^\triangle_{3,2}(G)$.

Note that there are two kinds of randomness in $\hf^\triangle(G)$: randomness in
choosing a permutation $\sigma$ and randomness in Warner's RR.
By~\eqref{eq:hf_triangle_II}, the expectation can be written as follows:
\begin{align*}
    \E_{\sigma, RR}[\hf^\triangle(G)]
    = \frac{n(n-1)}{6t} \sum_{i=1, 3, \ldots}^{2t-1} \E_{\sigma,
    RR}[\hf_{i,\sigma}^\triangle(G)].
\end{align*}
By the law of total expectation, we have
\begin{align*}
    \E_{\sigma, RR}[\hf_{i,\sigma}^\triangle(G)]
    &= \E_\sigma [\E_{RR}[\hf_{i,\sigma}^\triangle(G) | \sigma]] \\
    &= \E_\sigma [f_{i,\sigma}^\triangle(G)] ~~ \text{(by
    Theorem~\ref{thm:unbiased_I})} \\
    &= \frac{6}{n(n-1)}f^\triangle(G)~~ \text{(by~\eqref{eq:triangle-sample-est})}.
\end{align*}
Putting all together,
\begin{align*}
  \E_{\sigma, RR}[\hf^\triangle(G)] &=
  \frac{n(n-1)}{6t}\sum_{i=1,3,\ldots}^{2t-1} \E_{\sigma,
  RR}[\hf_{i,\sigma}^\triangle(G)] \\
  &= \frac{n(n-1)}{6t} \sum_{i=1,3,\ldots}^{2t-1} \frac{6}{n(n-1)}f^\triangle(G)
  \\
  &= f^\triangle(G).
\end{align*}\qed

\subsection{Proof of Theorem~\ref{thm:l2-loss_II}}
\label{sub:l2-loss_II_proof}
Because $\hf^\triangle(G)$ is unbiased, $\MSE(\hf^\triangle(G)) = \V[\hf^\triangle(G)]$.
By (\ref{eq:hf_triangle_II_ast}), the variance can be written as follows:
\begin{align}
    & \V_{\sigma,RR}[\hf^\triangle(G)] \nonumber\\
    &= \frac{n^2(n-1)^2}{36t^2}
    \V_{\sigma, RR}\left[\sum_{i=1, 3, \ldots}^{2t-1} \hf_{i,\sigma}^\triangle(G)\right] \nonumber\\
    &= \frac{n^2(n-1)^2}{36t^2}
    \V_\sigma\left[\sum_{i = 1, 3, \ldots}^{2t-1}
    \E_{RR}\left[\hf_{i,\sigma}^\triangle(G)\middle|\sigma \right]\right]
    \label{eq:thm:l2-loss_II_ve} \\
    &\qquad + \frac{n^2(n-1)^2}{36t^2}
    \E_\sigma\left[\sum_{i = 1, 3, \ldots}^{2t-1} \V_{RR}\left[\hf_{i,\sigma}^\triangle(G)\middle|\sigma \right]\right],
    \label{eq:thm:l2-loss_II_ev}
\end{align}
where in the
step we used the law of total variance and the independence of each
$\hf_{i,\sigma}^\triangle(G)$ given a fixed $\sigma$.

\noindent\textbf{Bounding~\eqref{eq:thm:l2-loss_II_ve}:}
We can write
\begin{align}
  \V_\sigma\left[\sum_{i = 1, 3, \ldots}^{2t-1} \E_{RR}\left[\hf_{i,\sigma}^\triangle(G)\middle|\sigma \right]\right]
    &= \V_\sigma\left[\sum_{i = 1, 3, \ldots}^{2t-1}
    f_{i,\sigma}^\triangle(G)\right]. \label{eq:tri-sampling-var}
\end{align}
Each random variable $f_{i,\sigma}^\triangle(G)$ represents a uniform draw from the set
$\{f_{i,j}^\triangle(G) : i,j \in [n], i \neq j\}$ without replacement. Applying
Lemma~\ref{lem:sampling_replacement_var}, the variance
in~\eqref{eq:tri-sampling-var} is upper bounded by
$t \V_\sigma[f_{1,\sigma}(G)]$. We can upper bound this final term
in the following:
\begin{align}
&\V_\sigma\left[f_{1,\sigma}^\triangle(G), \right] \nonumber\\
&\leq \E_\sigma[(f_{1,\sigma}^\triangle(G))^2] \nonumber\\
&= \frac{1}{n(n-1)} \sum_{1 \leq i, j \leq n, i \neq j} (f_{i,j}^\triangle(G))^2
\nonumber\\
&= \frac{1}{n(n-1)} \sum_{(i,j) \in E} (f_{i,j}^\triangle(G))^2 ~~ \text{(because $f_{i,j}^\triangle(G) = 0$ for $(i,j) \notin E$)} \nonumber\\
&\leq \frac{d_{max}^3}{n-1} ~~ \text{(because $|E| \leq n d_{max}$ and $f_{i,j}^\triangle(G) \leq d_{max}$)}.
\label{eq:thm:DP_II_cov_sigma_eq}
\end{align}

Plugging this into~\eqref{eq:tri-sampling-var}, we obtain
\begin{align*}
  \V_\sigma\left[\sum_{i = 1, 3, \ldots}^{2t-1}
    \E_{RR}\left[\hf_{i,\sigma}^\triangle(G)\middle|\sigma \right]\right]
    &\leq \frac{td_{max}^3}{n-1}.
\end{align*}

\noindent\textbf{Bounding~\eqref{eq:thm:l2-loss_II_ev}:}
For any value of $i$ and permutation $\sigma$, we have
from Theorem~\ref{thm:l2-loss_I} that
\[
  \V_{RR}[\hf_{i,\sigma}^\triangle(G)|\sigma] \leq
  err_{\AlgWSLE{}}(n,d_{max}, q, q_L).
\]
Thus,
\[
  \E_\sigma\left[\sum_{i = 1, 3, \ldots}^{2t-1}
    \V_{RR}\left[\hf_{i,\sigma}^\triangle(G)\middle|\sigma \right]\right]
    \leq t \cdot err_{\AlgWSLE{}}(n,d_{max}, q, q_L).
\]
\noindent\textbf{Putting it together:}
Plugging the upper bounds in, we obtain a
final bound of
\[
  \V[\hf^\triangle(G)] \leq
  \frac{n^4}{36t}err_{\AlgWSLE{}}(n,d_{max},q,q_L) +
  \frac{n^3}{36t}d_{max}^3. 
\]
Finally, when $\epsilon$ and $\delta$ are constants, $\epsilon_L = \log(n) + O(1)$, and $t = \lfloor\frac{n}{2}\rfloor$,
\[err_{\AlgWSLE{}}(n,d_{max},q,q_L) = O(d_{max}^2),\] and we obtain
\[
  \V[\hf^\triangle(G)] = O(n^3d_{max}^2 + n^2d_{max}^3).
\]
We can verify that $n^3d_{max}^2 \geq n^2d_{max}^3$ for all values of $d_{max}$
between $1$ and $n$, and thus the bound simplifies to $O(n^3 d_{max}^2)$.\qed

\begin{lemma}\label{lem:sampling_replacement_var}
Let $\calX$ be a finite subset of real numbers of size $n$.
Suppose $X_1, X_2, \ldots, X_k$ for $k \leq n$ are sampled
uniformly from $\mathcal{X}$ without replacement. Then,
\[
  \V[X_1 + \cdots + X_k] \leq k \V[X_1].
\]
\end{lemma}
\begin{proof}
  We have
  \[
    \V[X_1 + \cdots + X_k] = \sum_{i,j=1}^k Cov(X_i, X_j).
  \]
  We are done by observing each $X_i$ has the same distribution, and so $\V[X_i] = \V[X_1]$, and by showing $Cov(X_i, X_j) \leq 0$ when $i \neq j$. To prove the latter statement,
  let $\calX = \{x_1, \ldots, x_n\}$ with $x_1 \leq x_2 \leq \cdots \leq x_n$.
  Notice that for any $i \neq j$, the distribution $X_i | X_j$ is uniformly distributed
  on $\calX \setminus \{X_j\}$. This implies that $\E[X_i | X_j = x_1] \geq \E[X_i | X_j = x_2]
  \geq \cdots \geq \E[X_i | X_j = x_n]$. Let
  $y_\ell = \E[X_i | X_j = x_\ell]$ for $i,j \in [n]$. We have
  $y_1 \geq y_2 \geq \cdots \geq y_n$ for any $i \in [n]$.

  Because $X_j$ is uniformly distributed across $\calX$, we have $\E[X_j] = \frac{1}{n} \sum_{\ell=1}^n x_i$.
  Next, we can observe $\E[X_i] = \sum_{\ell=1}^n \E[X_i | X_j = x_\ell] \allowbreak \Pr[X_j = x_\ell] = \frac{1}{n} \sum_{\ell=1}^n y_\ell$.
  Finally, we have $\E[X_i X_j] = \sum_{\ell=1}^n x_\ell \allowbreak \E[X_i | X_j = x_\ell] \Pr[X_j = x_\ell] = \frac{1}{n} \sum_{\ell=1}^n x_\ell y_\ell $. Using Chebyshev's sum inequality, we are able to deduce that $\E[X_iX_j] \leq \E[X_i]\E[X_j]$, implying $Cov(X_i, X_j) \leq 0$.
\end{proof}

\subsection{Proof of Theorem~\ref{thm:DP_II_ast}}
\label{sub:DP_II_ast_proof}
\AlgWSTriVR{} interacts with $\bmA$ in the same way as \AlgWSTri{}, plus the additional degree estimates $\td_i$. The first calculation is protected by
$(\epsilon_2, \delta)$-element DP by Theorem~\ref{thm:DP_II}. The noisy degrees
are calculated with the Laplace mechanism, which provides a protection of
$(\epsilon_1, 0)$-element DP. Using composition, the entire computation provides
$(\epsilon_1 + \epsilon_2, \delta)$-element DP, and by
Proposition~\ref{prop:element_edge_DP}, the computation also provides
$(2(\epsilon_1 + \epsilon_2), 2\delta)$-edge DP.\qed

\subsection{Proof of Theorem~\ref{thm:bias_II_ast}}
\label{sub:bias_II_ast_proof}
Let $f^\triangle_*$ be the estimator returned by \AlgWSTriVR{} to distinguish
it from that returned by \AlgWSTri{}.
Define $V^+ = \{i : i \in [n], \td_i \geq c \td_{avg} \}$, and let $V^- = [n]
\setminus V^+$.
The randomness in \AlgWSTriVR{} comes
from randomized response, from the choice of $\sigma$, and from the choice of
$D$. Note that $i \in D$ if and only if $\textbf{1}_{\sigma(i) \in V^+}
\textbf{1}_{\sigma(i+1) \in V^+}$. For any $V^+$ and $V^-$,
\begin{align}
  &\E[\hf^\triangle_*(G) | V^+,V^-] \nonumber \\
  &= \E_{\sigma,RR}\left[\frac{n(n-1)}{6t}\sum_{i \in D}
  \hf^\triangle_{i, \sigma}(G) \middle| V^+,V^- \right] \nonumber \\
  &= \E_{\sigma,RR}\left[\frac{n(n-1)}{6t}\sum_{i = 1,3,\ldots}^{2t-1}
  \textbf{1}_{\sigma(i) \in V^+}\textbf{1}_{\sigma(i+1) \in V^+} \hf^\triangle_{i,
  \sigma}(G) \middle| V^+,V^- \right] \nonumber \\
  &= \frac{n(n-1)}{6t}\sum_{i = 1,3,\ldots}^{2t-1}
  \E_{\sigma,RR}[\textbf{1}_{\sigma(i) \in V^+}\textbf{1}_{\sigma(i+1) \in V^+}
  \hf^\triangle_{i, \sigma}(G) | V^+,V^-] \nonumber \\
  &= \frac{n(n-1)}{6t}\sum_{i = 1,3,\ldots}^{2t-1}
  \E_\sigma[\textbf{1}_{\sigma(i) \in V^+}\textbf{1}_{\sigma(i+1) \in V^+} f^\triangle_{i,
  \sigma}(G) | V^+,V^-], \label{eq:exp-expand}
\end{align}
where the last line uses the fact that for a fixed $i,\sigma$,
$\E_{RR}[\hf_{i,\sigma}(G)] = f_{i,\sigma}(G)$
(Theorem~\ref{thm:unbiased_I}). Using the inequality $f_{i,\sigma}^\triangle(G)
\leq \min\{d_{\sigma(i)}, \allowbreak d_{\sigma(i+1)}\}$, and the fact that
$f_{i,\sigma}^\triangle(G) = 0$ unless $(v_{\sigma(i)},v_{\sigma(i+1)}) \in E$, we obtain
\begin{align*}
  &|\E_{\sigma}[f_{i,\sigma}^\triangle(G) - \textbf{1}_{\sigma(i)\in V^+}
  \textbf{1}_{\sigma(i+1)\in V^+} f_{i,\sigma}^\triangle(G)|V^+,V^-]| \\
  &\leq
  \E_{\sigma}[\min\{d_{\sigma(i)}, d_{\sigma(i+1)}\}\textbf{1}_{v_{\sigma(i)},v_{\sigma(i+1)} \in E}\textbf{1}_{\sigma(i) \in V^-
  \vee \sigma(i+1) \in V^-} \\
  & \hspace{10mm} | V^+,V^-].
\end{align*}
Expanding the second equation, we obtain
\begin{align*}
  &\frac{1}{n(n-1)} \sum_{1\leq i,j \leq n, i \neq j}
  \min\{d_i, d_j\}\textbf{1}_{(i,j) \in E}\textbf{1}_{i \in V^-
  \vee j \in V^-} \\
  &= \frac{1}{n(n-1)}\left( \sum_{j \in V^-, (i,j) \in E} \min\{d_i, d_j\} \right. \\
  &\hspace{18mm} \left. +
  \sum_{i \in V^-, j \in V^+, (i,j) \in E} \min\{d_i, d_j\}\right) \\
  &\leq \frac{1}{n(n-1)}\left( \sum_{j \in V^-, (i,j) \in E} d_j +
  \sum_{i \in V^-, j \in V^+, (i,j) \in E} d_i\right) \\
  &\leq \frac{1}{n(n-1)}\left( \sum_{j \in V^-} d_j^2 +
  \sum_{i \in V^-} d_i^2 \right) \\
  &\leq \frac{2}{n(n-1)}d_{sum,-}^{(2)},
\end{align*}
where $d_{sum,-}^{(2)} = \sum_{i \in V^-} d_i^2$.

Applying the triangle inequality,
\begin{align*}
  &\left|\sum_{i=1,3,\ldots}^{2t-1}\E_\sigma [\textbf{1}_{\sigma(i)\in V}\textbf{1}_{\sigma(i+1) \in V}
  f_{i,\sigma}^\triangle(G) | V] -
  \sum_{i=1,3,\ldots}^{2t-1}\E_\sigma[f_{i,\sigma}^\triangle(G) ]\right| \\
  &\leq
  \frac{2t}{n(n-1)} d_{sum,-}^{(2)}
\end{align*}
Plugging in~\eqref{eq:exp-expand} and~\eqref{eq:triangle-sample-est}, we obtain
\begin{align}
  \hspace{-1mm}\left|\frac{6t}{n(n-1)}\E[\hf^\triangle_*(G)|V^+,V^-] -
  \frac{6t}{n(n-1)}f^\triangle(G)\right| &\leq
  \frac{2t}{n(n-1)} d_{sum,-}^{(2)} \nonumber \\
  \hspace{-1mm}\left|\E[\hf^\triangle_*(G)|V^+,V^-] -
  f^\triangle(G)\right| &\leq \frac{1}{3} d_{sum,-}^{(2)}.\label{eq:bias-vv}
\end{align}

Marginalizing over all $V^+$ and $V^-$, and applying the triangle inequality again, we obtain
\[
  \left|\E[\hf^\triangle_*(G)] -
  f^\triangle(G)\right| \leq
  \frac{1}{3} \E[d_{sum,-}^{(2)}].
\]

Now, we have
\begin{align*}
  \E[d_{sum,-}^{(2)}] &= \sum_{i = 1}^n d_i^2 \Pr[i \in V^-] \\
  &\leq n(c d_{avg})^2 + \sum_{i \in [n], d_i \geq c d_{avg}} d_i^2 \Pr[i \in
  V^-].
\end{align*}
In the second line, we split the sum into those nodes where $d_i \geq c d_{avg}$
(of which there are only $n^\alpha$, since $c \geq \lambda$), and other nodes.
In order for $i$ to be in $V^-$ for a node such that $d_i \geq c d_{avg}$, we
must have that $Lap(\frac{1}{\epsilon_1}) \leq d_i - c d_{avg}$. The probability
of this occuring is at most $e^{-(d_i - c d_{avg})\epsilon_1}$. Using calculus, we
can show the expression $d_i^2 e^{-(d_i - c d_{avg})\epsilon_1}$ is maximized when 
$d_i = c d_{avg}$ (when 
$c d_{avg} \geq \frac{2}{\epsilon_1}$), and when $d_i = \frac{2}{\epsilon_1}$
(otherwise).

In the first case, we have
\[
  \sum_{i \in [n], d_i \geq c d_{avg}} d_i^2 \Pr[i \in V^-] \leq n^\alpha
  (cd_{avg})^2.
\]
In the second, we have
\[
  \sum_{i \in [n], d_i \geq c d_{avg}} d_i^2 \Pr[i \in V^-] \leq
  n^\alpha\left(\frac{2}{\epsilon_1}\right)^2 e^{-2 + cd_{avg} \epsilon_1} \leq
  n^\alpha\left(\frac{2}{\epsilon_1}\right)^2.
\]
Thus, our overall bound becomes 
\begin{align*}
  \E[d_{sum,-}^{(2)}] &= \sum_{i = 1}^n d_i^2 \Pr[i \in V^-] \\
  &\leq n(c d_{avg})^2 + n^\alpha \left(\frac{2}{\epsilon_1}\right)^2,
\end{align*}
and therefore 
\[
  \left|\E[\hf^\triangle_*(G)] -
  f^\triangle(G)\right| \leq
  \frac{n c^2 d_{avg}^2}{3}  + \frac{4 n^\alpha}{3 \epsilon_1^2}.
\]\qed

\subsection{Proof of Theorem~\ref{thm:var_II_ast}}
\label{sub:var_II_ast_proof}

Define $f^\triangle_*$, $V^+$, and $V^-$ be as they are defined in
Appendix~\ref{sub:bias_II_ast_proof}. 
Using the law of total variance, we have
\begin{equation}\label{eq:total-var-vv}
  \V[\hf_*^\triangle(G)] = \E_{V^+,V^-}\V[\hf_*^\triangle(G)|V^+,V^-] +
  \V_{V^+,V^-}\E[\hf_*^\triangle(G)|V^+,V^-].
\end{equation}

From~\eqref{eq:bias-vv}, $\E[\hf_*^\triangle(G)|V^+,V^-]$ is always in the
range $[f^\triangle(G) - \frac{\overline{d}}{3}, f^\triangle(G) + \frac{
  \overline{d}}{3}]$, where $\overline{d} = 
\max_{V^- \subseteq [n]} \sum_{i \in V^-} d_i^2 \leq \sum_{i=1}^n d_i^2 \leq n d_{max}^2$.
Because the maximum variance of a variable bounded between $[A,B]$ is
$\frac{(B-A)^2}{4}$, 
the second term of (\ref{eq:total-var-vv}) can be written as 
$\V_{V^+,V^-}\E[\hf_*^\triangle(G)|V^+,V^-] 
\leq \allowbreak \frac{(\sum_{i=1}^n d_i^2)^2}{9} \leq \frac{n^2 d_{max}^4}{9} $.

Again using the law of total variance on $\V[\hf_*^\triangle(G)|V^+,V^-]$, we
obtain, similar to~\eqref{eq:thm:l2-loss_II_ve} and \eqref{eq:thm:l2-loss_II_ev}:
\begin{align}
  &\left(\frac{6t}{n(n-1)}\right)^2\V[\hf_*^\triangle(G)|V^+,V^-] \nonumber \\
  &=\V_{\sigma,RR}\left[\sum_{i=1,3,\ldots}^{2t-1} \textbf{1}_{\sigma(i)\in
  V^+}\textbf{1}_{\sigma(i+1) \in V^+} \hf_{i,\sigma}^\triangle(G)\middle| V^+,V^-\right] \nonumber \\
  &=\V_{\sigma} \left[ \sum_{i=1,3,\ldots}^{2t-1} \textbf{1}_{\sigma(i) \in
  V^+}\textbf{1}_{\sigma(i+1) \in V^+}\E_{RR}\left[\hf_{i,\sigma}^\triangle(G)\middle|
  \sigma \right]\middle| V^+,V^-\right]\label{eq:thm:l2-loss_II_ast_ve} \\
  &\hspace{1mm} +\E_{\sigma} \left[\sum_{i=1,3,\ldots}^{2t-1}\textbf{1}_{\sigma(i) \in
  V^+}\textbf{1}_{\sigma(i+1) \in V^+}\V_{RR}\left[\hf_{i,\sigma}^\triangle(G)\middle|
  \sigma \right]\middle| V^+,V^-\right]\label{eq:thm:l2-loss_II_ast_ev}
\end{align}
\noindent\textbf{Bounding~\eqref{eq:thm:l2-loss_II_ast_ve}:}
Similar to the process for bounding~\eqref{eq:thm:l2-loss_II_ve}, we have that
$\E_{RR}[\hf_{i,\sigma}^\triangle(G)|\sigma] = f_{i,\sigma}^\triangle(G)$. The
collection $\{\textbf{1}_{\sigma(i) \in V^+}\textbf{1}_{\sigma(i+1)\in V^+}
\allowbreak f_{i,\sigma}^\triangle(G) \}$ for $i \in \{1, 3, \ldots, 2t-1\}$ consists of draws without replacement from the set
\[
  \mathcal{T} = \{\textbf{1}_{i \in V^+} \textbf{1}_{j \in V^+}
  f_{i,j}^\triangle(G) : i,j \in [n], i \neq j\},
\]
We can then apply Lemma~\ref{lem:sampling_replacement_var} for
an upper bound of
\[t \V_\sigma[\textbf{1}_{\sigma(1)\in
V^+}\textbf{1}_{\sigma(2) \in V^+}f_{1,\sigma}(G)].\] Using the same line of
reasoning as that to obtain~\eqref{eq:thm:DP_II_cov_sigma_eq}, we obtain an
upper bound of
\begin{align*}
  &\V_\sigma[\textbf{1}_{\sigma(1)\in
  V^+}\textbf{1}_{\sigma(2) \in V^+}f^\triangle_{1,\sigma}(G)] \\
  &\leq
  \E_\sigma[\textbf{1}_{\sigma(1)\in
  V^+}\textbf{1}_{\sigma(2) \in V^+}f^\triangle_{1,\sigma}(G)^2] \\
  &\leq \frac{1}{n(n-1)} \sum_{(v_i,v_j) \in E, i \in V^+, j \in V^+}
  f_{i,j}^\triangle(G)^2 \\
  &\leq \frac{1}{n(n-1)} |V^+| d_{max}^3,
\end{align*}
with the last line holding because there are at most $|V^+| d_{max}$ edges
within $V^+$. Thus,~\eqref{eq:thm:l2-loss_II_ast_ve} is at most
$\frac{t|V^+|d_{max}^3}{n(n-1)}$.

\noindent\textbf{Bounding~\eqref{eq:thm:l2-loss_II_ast_ev}:}
Similar to the process for bounding~\eqref{eq:thm:l2-loss_II_ev}, we use
Theorem~\ref{thm:l2-loss_I}, along with the fact that
$\E_\sigma[\textbf{1}_{\sigma(i) \in V^+} \textbf{1}_{\sigma(i+1) \in V^+}] =
\frac{|V^+|(|V^+|-1)}{n(n-1)} \leq \frac{|V^+|^2}{n^2}$, to obtain an upper bound of
\[
\frac{t|V^+|^2}{n^2} err_{\AlgWSLE}(n,d_{max},q,q_L).
\]

\noindent\textbf{Putting it together:}
Summing together~\eqref{eq:thm:l2-loss_II_ast_ve}
and~\eqref{eq:thm:l2-loss_II_ast_ev},
manipulating constants, and taking an expectation over $V$, we obtain
\begin{align*}
  &\E_{V^+,V^-}\V[\hf_*^\triangle(G)|V^+,V^-] \\
  &\leq \frac{n^2 d_{max}^3 \E[|V^+|]}{36t} +
  \frac{n^2\E[|V^+|^2]}{36t}err_{\AlgWSLE}(n, d_{max}, q, q_L).
\end{align*}
Plugging back into~\eqref{eq:total-var-vv}, we have
\begin{align*}
  &\V[\hf_*^\triangle(G)] \leq \frac{n^2 d_{max}^4}{9} + \\
  &\frac{n^2\E[|V^+|^2]}{36t}err_{\AlgWSLE}(n,d_{max},q,q_L) + \frac{n^2
  d_{max}^3 \E[|V^+|]}{36t}.
\end{align*}
We are given that there are just $n^\alpha$ nodes with degrees higher than $\lambda
d_{avg}$. 
Let $L_i$ be the Laplace random variable added to $d_i$.
We have $\Pr[L_i \geq T] \leq e^{-T \epsilon_1}$ for $T > 0$.
Observe
\[
  |V^+| \leq n^\alpha + \sum_{i \in [n], d_i \leq \lambda d_{avg}}
  \textbf{1}_{L_i \geq (c-\lambda) d_{avg}}.
\]
We have $\E[\textbf{1}_{L_i \geq (c-\lambda) d_{avg}}] =
e^{-(c-\lambda)\epsilon_1 d_{avg}}$, which by the condition on $c$, is at most
$n^{\alpha-1}$. Thus, 
$\E[|V^+|] \leq n^{\alpha} + n \cdot  n^{\alpha-1} = 2 n^{\alpha}$. 
Similarly, 
\begin{align*}
  \E[|V^+|^2] &\leq n^{2\alpha} + 2n^{\alpha} \E[|V^+|] \\& \hspace{4mm}+2\sum_{i,j \in [n], i \ne j, d_i,d_j \leq \lambda d_{avg}} n^{2(\alpha-1)} + \sum_{i \in [n], d_i \leq
  \lambda d_{avg}} n^{\alpha-1} \\
  &\leq n^{2\alpha} + 4n^{2\alpha} + 2n^{2\alpha} + n^\alpha \\
  &\leq 8n^{2\alpha}.
\end{align*}
Plugging in, we obtain
\begin{align*}
  &\V[\hf_*^\triangle(G)] \leq \frac{n^2 d_{max}^4}{9} + \\
  &\frac{2n^{2+2\alpha}}{9t}err_{\AlgWSLE}(n,d_{max},q,q_L) + \frac{n^{2+\alpha}
  d_{max}^3 }{36t}.
\end{align*}

When $\epsilon$ and $\delta$, are treated
as constants, $\epsilon_L = \log n + O(1)$, and $t = \lfloor\frac{n}{2}\rfloor$, then
$err_{\AlgWSLE}(n,d_{max},q,q_L) = O(d_{max}^2)$, and
\begin{align*}
  \V[\hf_*^\triangle(G)] 
  &\leq O(n^2 d_{max}^4 + n^{1+2\alpha} d_{max}^2 + n^{1+\alpha} d_{max}^3) \\
  &= O(n^2 d_{max}^4 + n^{1+2\alpha} d_{max}^2).
\end{align*}
We can remove the third term because it is always smaller than the first term.
\qed

\section{Proofs of Statements in Section 6}
\label{sec:4cycle_proof}
In the following, we define
$f_{i,\sigma}^\wedge(G) = f^\wedge_{\sigma(2i), \sigma(2i-1)}(G)$ and
$\hf_{i,\sigma}^\wedge(G) = \hf^\wedge_{\sigma(2i), \sigma(2i-1)}(G)$ as a shorthand.
Similarly, we do the same with $f_{i,\sigma}^\square$, $\hf_{i,\sigma}^\square$
by replacing $\wedge$ with $\square$.

\subsection{Proof of Theorem~\ref{thm:DP_IV}}
\label{sub:DP_IV_proof}
\AlgWSCyc{} interacts with $\bmA$ in the same way as \AlgWSTri{}. 
The subsequent processes (lines 7-10 in Algorithm~\ref{alg:wshuffle_cycle}) are post-processing on $\{y_{\pi_i(k)} | k \in I_{-(\sigma(i),\sigma(i+1))}\}$. 
Thus, by the post-processing invariance \cite{DP}, \AlgWSCyc{} provides
$(\epsilon, \delta)$-element DP and $(2\epsilon, 2\delta)$-element DP in the same way as \AlgWSTri{} (see Appendix~\ref{sub:DP_II_proof} for the proof of DP for \AlgWSTri{}). \qed

\subsection{Proof of Theorem~\ref{thm:unbiased_IV}}
\label{sub:unbiased_IV_proof}
  Notice that the number of $4$-cycles can be computed as
  \begin{align}
    4f^\square(G) &= \sum_{1 \leq i,j \leq n, i \neq j} f_{i,j}^\square(G) \\
    &= n(n-1)\E_{\sigma}[f_{i,\sigma}^\square(G)],\label{eq:4-cycle-calc}
  \end{align}
  where the $4$ appears because for every $4$-cycle, there are $4$ choices for
  diagonally opposite nodes.

    We will show that $\hf_{i,j}^\wedge(G) = \sum_{k \in I_{-(i,j)}}
    \frac{y_k-q_L}{1-2q_L}$ is an unbiased estimate of $f_{i,j}^\wedge(G)$.
    Since we use $\epsilon_L$-RR in \AlgWS{}, we have
    \begin{align*}
      \E[y_{k}] = (1-q_L) w_{i-k-j} + q_L (1-w_{i-k-j}).
    \end{align*}
    Thus,  we have:
    \begin{align*}
        \E\left[\frac{y_{k} - q_L}{1-2q_L}\right]
        &= \frac{(1-q_L) w_{i-k-j} - q_L w_{i-k-j}}{1-2q_L} \\
        &= w_{i-k-j}.
    \end{align*}
    The sum of these is clearly the number of wedges connected to users $v_i$ and $v_j$.
    Therefore, $\E[\hf_{i,j}^\wedge(G)] = f_{i,j}^\wedge(G)$.

    Furthermore, $(1-2q_L)\hf_{i,j}(G)$ is a sum of $(n-2)$ Bernoulli random
    variables shifted by $(n-2)q_L$. Each Bernoulli r.v. has variance
    $q_L(1-q_L)$, and thus $\V[\hf_{i,j}(G)] = \frac{(n-2)q_L(1-q_L)}{(1-2q_L)^2}$.
    This information is enough to verify that
    \begin{align*}
        \E[\hf_{i,j}^\wedge(G)^2] &= \E[\hf_{i,j}^\wedge(G)]^2 + \V[\hf_{i,j}^\wedge(G)] \\
        &= f_{i,j}^\wedge(G)^2 + \frac{(n-2)q_L(1-q_L)}{(1-2q_L)^2}.
    \end{align*}
    Putting this together and plugging into~\eqref{eq:4-cycle-ij-estimate}, we
    obtain
    \begin{align}
        \E_{RR}\left[\hf^\square_{i,j}(G)\right] &=
        \E_{RR}\left[\frac{\hf^\wedge_{i,j}(G)(\hf^\wedge_{i,j}(G)-1)}{2} -
        \frac{n-2}{2}\frac{q_L(1-q_L)}{(1-2q_L)^2} \right] \nonumber \\
        &=
        \frac{f^\wedge_{i,j}(G)(f_{i,j}^\wedge(G)-1)}{2},\label{eq:4-cycle-estimate_proof}
    \end{align}
    In the above equation, we emphasize that the randomness in the expectation is over
    the randomized response used by the estimator $\hf$.
    From~\eqref{eq:4-cycle-estimate_proof}, the estimate $\hf^\square(G)$ satisfies:
    \begin{align*}
      &\E[\hf^\square(G)] = \E_\sigma[\E_{RR}[\hf^\square(G)|\sigma]] \\
        &= \frac{n(n-1)}{4t} \hspace{-1mm}\sum_{i=1}^t \E_{\sigma}\hspace{-1mm}\left[\E_{RR}\left[
          \frac{\hf_i(G)(\hf_{i}(G)-1)}{2} - \frac{n-2}{2}\frac{q_L(1-q_L)}{(1-2q_L)^2} \middle| \sigma\right]\right] \\
        &= \frac{n(n-1)}{4t} \sum_{i=1}^t
        \E_\sigma\left[\frac{\hf_{i,\sigma}(G)(\hf_{i,\sigma}(G)-1)}{2}\right] \\
        &= \frac{n(n-1)}{4}\E_{\sigma}\left[f_{i,\sigma}^\square(G)\right] \\
        &= f^\square(G).
    \end{align*}\qed

\subsection{Proof of Theorem~\ref{thm:l2-loss_IV}}
\label{sub:l2-loss_IV_proof}
From~\eqref{eq:4-cycle-estimate}, we have
\begin{align*}
  \V[\hf^\square(G)] = \left(\frac{n(n-1)}{4t}\right)^2 \V\left[\sum_{i=1}^t
  \hf_{i,\sigma}^\square(G)\right].
\end{align*}
Using the law of total variance, along with the fact that the
$\hf_{i,\sigma}^\square(G)$ are mutually independent given $\sigma$, we write
\begin{align}
  \V_{\sigma, RR}\left[\sum_{i=1}^t \hf_{i,\sigma}^\square(G)\right] = \E_\sigma\left[\sum_{i=1}^t
    \V_{RR}\left[\hf_{i,\sigma}^\square(G) \middle| \sigma\right]\right] \label{eq:4-cycle-ev}\\
    + \V_\sigma\left[ \sum_{i=1}^t
    \E_{RR}\left[\hf_{i,\sigma}^\square(G)\middle| \sigma
    \right]\right]. \label{eq:4-cycle-ve}
\end{align}
We will now shift our attention to upper bounding the terms on the left- and right-hand sides of the above sum.

\noindent\textbf{Upper bounding~\eqref{eq:4-cycle-ev}:}
Our analysis will apply to any fixed $\sigma$, and we will assume for the rest
of the proof that $\sigma$ is a fixed constant permutation.
By~\eqref{eq:4-cycle-ij-estimate}, we have
\[\V_{RR}[\hf_{i,\sigma}^\square(G)|\sigma] =
\V_{RR}\left[\frac{\hf_{i,\sigma}^\wedge(G)(\hf_{i,\sigma}^\wedge(G)-1)}{2}\middle|\sigma\right].\] Each term
can be upper bounded using the fact that $\V[X+Z] \leq 4 \max\{\V[X], \V[Z]\}$
for random variables $X,Z$.
Thus, for any $i$, we have
$\V[\frac{\hf_{i,\sigma}^\wedge(G)(\hf_{i,\sigma}^\wedge(G)-1)}{2}]
\leq \max\{\V[\hf^\wedge_{i,\sigma}(G)^2], \V[\hf_{i,\sigma}^\wedge(G)]\}$. In
Appendix~\ref{sub:unbiased_IV_proof}, we showed
\begin{equation}\label{eq:4-cycle-l2-varI}
  \V[\hf^\wedge_{i,\sigma}(G)] \leq \frac{n q_L(1-q_L)}{(1-2q_L)^2}.
\end{equation}
To bound $\V[\hf^\wedge_{i,\sigma}(G)^2]$, first we define $I = I_{-(\sigma(2i), \sigma(2i-1))}$ as a shorthand and
plug
in~\eqref{eq:wedge-estimate}:
\begin{align*}
    \V[\hf_{i,\sigma}^\wedge(G)^2] = \frac{1}{(1-2q_L)^4} \V\left[ \sum_{k,\ell \in I} y_{k} y_{\ell}\right].
\end{align*}
We can write
\begin{align*}
    \V\left[ \sum_{k,\ell \in I} y_k y_\ell \right] &= \sum_{k,\ell, k', \ell' \in I} Cov(y_k y_\ell, y_{k'} y_{\ell'}) \\
    &= 4 \sum_{k,\ell,k' \text{ distinct in } I} Cov(y_k y_\ell, y_{k'} y_\ell) \\
    &\qquad+ 4 \sum_{k,\ell \text{ distinct in } I} Cov(y_k y_\ell, y_k^2) \\
    &\qquad+ \sum_{k \in I} Cov(y_k^2, y_k^2),
\end{align*}
where in the second equality we have canceled the covariances equal to $0$ due to independence.
Note that the coefficient in the first term is 4 because $Cov(y_k y_\ell, \allowbreak y_{k'} y_\ell)$ captures
$Cov(y_k y_\ell, \allowbreak y_{k'} y_\ell)$, $Cov(y_k y_\ell,\allowbreak  y_\ell y_{k'})$, $Cov(y_\ell y_k, \allowbreak y_{k'} y_\ell)$, and $Cov(y_\ell y_k, \allowbreak y_\ell y_{k'})$.
In other words, there are four possible combinations depending on the positions of two $y_\ell$'s.
Similarly, the coefficient in the second term is 4 because $Cov(y_k y_\ell, \allowbreak y_k^2)$ captures
$Cov(y_k y_\ell, \allowbreak y_k^2)$,
$Cov(y_\ell y_k, \allowbreak y_k^2)$,
$Cov(y_k^2, \allowbreak y_k y_\ell)$, and
$Cov(y_k^2, \allowbreak y_\ell y_k)$.

Now, we have that
\begin{align*}
    Cov(y_k y_\ell, y_{k'}y_\ell) &= \E[y_ky_{\ell}y_{k'}y_\ell ] - \E[y_ky_\ell]\E[y_{k'}y_\ell] \\
    &= \E[ y_\ell^2]\E[y_{k}]\E[y_{k'}] - \E[y_\ell]^2\E[y_{k}]\E[y_{k'}] \\
    &= \E[y_k]\E[y_{k'}] \V[y_\ell].
\end{align*}

We unconditionally have that $\V[y_\ell] \leq q_L$, and there are
at most $d_{max}$ choices for $k$
such that $\E[y_k] = 1-q_L$, and the remaining choices satisfy $\E[y_k] = q_L$. Finally, there are at most $n$ choices for $\ell$ in $I$. Thus,
\begin{align*}
    &\sum_{k,\ell,k' \text{ distinct in } I} Cov(y_k y_\ell, y_{k'}y_{\ell}) \\
    &\qquad \leq \sum_{k,\ell,k' \text{ distinct in } I} \E[y_k]\E[y_{k'}] \V[y_\ell] \\
    &\qquad \leq nq_L ( d_{max}^2 (1-q_L)^2 + 2d_{max}(n-d_{max}) (1-q_L) q_L \\
    &\hspace{16mm} + (n-d_{max})^2 q_L^2) \\
    &\qquad \leq nq_L ( d_{max}^2 + 2nd_{max} q_L + n^2q_L^2) \\
    &\qquad \leq nq_L (d_{max} + nq_L)^2.
\end{align*}
Now,
using the fact that $y_k$ is zero-one valued,
we have $Cov(y_k y_\ell, \allowbreak y_k^2) =
Cov(y_k y_\ell, y_k)= \E[y_\ell] \V(y_k)$. There are at most $n$ choices for $k$, and there are
at most $d_{max}$ choices such that $\E[y_\ell] = 1-q_L$, and
the remaining choices satisfy
$\E[y_\ell] = q_L$. Thus,
\begin{align*}
    &\sum_{k,\ell \text{ distinct in } I} Cov(y_k y_\ell, y_k^2) \\
    &\qquad \leq \sum_{k,\ell \text{ distinct in } I} \E[y_\ell]\V[y_k] \\
    &\qquad \leq nq_L (d_{max}(1-q_L) + (n-d_{max})q_L) \\
    &\qquad \leq nq_L (d_{max} + nq_L).
\end{align*}
Finally, $\V[y_k^2] = \V[y_k] \leq q_L$, and so $\sum_{k \in I} Cov(y_k^2, y_k^2) \leq nq_L$. Thus,
\begin{align*}
  \V\left[\sum_{k,\ell \in I} y_ky_\ell \right] &\leq \big(4nq_L(d_{max} +
    nq_L)^2 \\ &\qquad+ 4nq_L(d_{max} + nq_L) + nq_L\big) \\
    &\leq 9nq_L(d_{max} + nq_L)^2.
\end{align*}

This implies
\begin{equation}\label{eq:4-cycle-l2-varII}
  \V[\hf_{i,\sigma}^\wedge(G)^2] \leq
    \frac{9nq_L(d_{max} + nq_L)^2}{(1-2q_L)^4}.
\end{equation}
We clearly have that \eqref{eq:4-cycle-l2-varII} is bigger
than~\eqref{eq:4-cycle-l2-varI}, and thus~\eqref{eq:4-cycle-l2-varII} is an upper
bound for $\V_{RR}[\hf_{i,\sigma}^\square(G)|\sigma]$.
Thus,~\eqref{eq:4-cycle-ev} is upper bounded by
    $\frac{9ntq_L(d_{max} + nq_L)^2}{(1-2q_L)^4}$.

\noindent\textbf{Upper bounding~\eqref{eq:4-cycle-ve}:}
By~\eqref{eq:4-cycle-estimate}, we can write
\[
  \V_\sigma\left[\sum_{i=1}^t \E_{RR}\left[\hf_{i,\sigma}^\square(G)\middle|
  \sigma \right]\right] = \V_\sigma\left[\sum_{i=1}^t f_{i,\sigma}^\square(G)\right].
\]
When $\sigma$ is chosen randomly, the random variables $f_{i,\sigma}(G)$
for $1 \leq i \leq t$ are a uniform sampling without replacement from the set \[
\mathcal{F} =
\left\{f_{i,j}^\square(G): i,j \in [n], i\neq j\right\}.
\]
Applying Lemma~\ref{lem:sampling_replacement_var}, we have
\[
  \V_\sigma\left[\sum_{i=1}^t f_{i,\sigma}^\square(G)\right]
\leq t \V_\sigma[ f_{1,\sigma}^\square(G)].
\]
Now, we have
\begin{align*}
  \V_{\sigma} [f_{1,\sigma}^\square(G)] &\leq
    \E_{\sigma}[f_{1,\sigma}^\square(G)^2] \\
    &=  \frac{1}{4}\E_{\sigma}\left[(f_{i,\sigma}^\wedge(G)(f_{i,\sigma}^\wedge(G)-1))^2\right] \\
    &\leq  \frac{1}{4}\E_{\sigma}\left[f_{i,\sigma}^\wedge(G)^4\right] \\
    &=  \frac{1}{4n(n-1)}\sum_{1 \leq i,j \leq n, i\neq j}f_{i,j}^\wedge(G)^4.
\end{align*}
Let $E^2$ be the set of node pairs for which there exists a $2$-hop path between them in $G$. We have $|E^2| \leq nd_{max}^2$. Now, we can write
\begin{align*}
\frac{1}{4n(n-1)}\sum_{1 \leq i,j \leq n, i\neq j}f_{i,j}^\wedge(G)^4
    &= \frac{1}{4n(n-1)} \sum_{(i,j) \in E^2}f_{i,j}^\wedge(G)^4 \\
    &\leq \frac{1}{4n(n-1)} \sum_{(i,j) \in E^2} d_{max}^4 \\
    &\leq \frac{d_{max}^6}{4(n-1)}.
\end{align*}
This allows to conclude that the variance of~\eqref{eq:4-cycle-ve} is at most
$\frac{t d_{max}^6}{4(n-1)}$.

\noindent\textbf{Putting it together:}
Substituting in~\eqref{eq:4-cycle-ev} and~\eqref{eq:4-cycle-ve}, we obtain
\begin{align*}
    \V[\hf^\square(G)] &\leq
    \frac{n^2(n-1)^2}{16t^2}\left(\frac{td_{max}^6}{4(n-1)} +
    \frac{9ntq_L(d_{max} + nq_L)^2}{(1-2q_L)^4} \right) \\
    &\leq \frac{9n^5 q_L(d_{max} +
    nq_L)^2}{16t(1-2q_L)^4} + \frac{n^3d_{max}^6}{64t}.
\end{align*}\qed

\end{document}